\documentclass[11pt]{amsart}

\usepackage{a4wide, amsmath, amsfonts, amssymb, mathrsfs, amsthm, breqn, mathtools}
\usepackage[hyperfootnotes=false]{hyperref}

\allowdisplaybreaks[2]

\numberwithin{equation}{section}

\newtheorem{theorem}{Theorem}[section]
\newtheorem{lemma}[theorem]{Lemma}
\newtheorem{corollary}[theorem]{Corollary}
\newtheorem{remark}[theorem]{Remark}
\newtheorem{proposition}[theorem]{Proposition}
\newtheorem{assumption}[theorem]{Assumption}
\newtheorem*{RIE}{Property \textup{(RIE)}}

\newcommand{\dd}{\,\mathrm{d}}
\renewcommand{\d}{\mathrm{d}}
\newcommand{\D}{\mathrm{D}}
\renewcommand{\epsilon}{\varepsilon}
\let\oldphi\phi
\renewcommand{\phi}{\varphi}
\newcommand{\E}{\mathbb{E}}
\newcommand{\N}{\mathbb{N}}
\renewcommand{\P}{\mathbb{P}}
\newcommand{\R}{\mathbb{R}}
\newcommand{\W}{\mathbb{W}}
\newcommand{\X}{\mathbb{X}}
\newcommand{\Z}{\mathbb{Z}}
\newcommand{\1}{\mathbf{1}}
\newcommand{\bZ}{\mathbf{Z}}
\newcommand{\bW}{\mathbf{W}}
\newcommand{\bX}{\mathbf{X}}
\newcommand{\tX}{\widetilde{X}}
\newcommand{\tY}{\widetilde{Y}}
\newcommand{\tbbX}{\widetilde{\X}}
\newcommand{\tbX}{\widetilde{\bX}}
\newcommand{\tb}{\tilde{b}}
\newcommand{\tsigma}{\tilde{\sigma}}
\newcommand{\cA}{\mathcal{A}}
\newcommand{\cD}{\mathcal{D}}
\newcommand{\cE}{\mathcal{E}}
\newcommand{\cF}{\mathcal{F}}
\newcommand{\cL}{\mathcal{L}}
\newcommand{\cP}{\mathcal{P}}
\newcommand{\cV}{\mathcal{V}}
\newcommand{\crptX}{\mathcal{V}^p_{\widetilde{X}}}
\newcommand{\crpW}{\mathcal{V}^p_{(\cdot,W)}}
\newcommand{\crpX}{\mathcal{V}^p_X}

\DeclareMathOperator*{\Sym}{Sym}

\title[Pathwise analysis of log-optimal portfolios]{Pathwise analysis of log-optimal portfolios}

\author[Allan]{Andrew L. Allan}
\address{Andrew L. Allan, Durham University, United Kingdom}
\email{andrew.l.allan@durham.ac.uk}

\author[Kwossek]{Anna P. Kwossek}
\address{Anna P. Kwossek, University of Vienna, Austria}
\email{anna.paula.kwossek@univie.ac.at}

\author[Liu]{Chong Liu}
\address{Chong Liu, ShanghaiTech University, China}
\email{liuchong@shanghaitech.edu.cn}

\author[Pr{\"o}mel]{David J. Pr{\"o}mel}
\address{David J. Pr{\"o}mel, University of Mannheim, Germany}
\email{proemel@uni-mannheim.de}

\date{\today}

\begin{document}

\begin{abstract}
Based on the theory of c{\`a}dl{\`a}g rough paths, we develop a pathwise approach to analyze stability and approximation properties of portfolios along individual price trajectories generated by standard models of financial markets. As a prototypical example from portfolio theory, we study the log-optimal portfolio in a classical investment-consumption optimization problem on a frictionless financial market modelled by an It{\^o} diffusion process. We identify a fully deterministic framework that enables a pathwise construction of the log-optimal portfolio, for which we then establish pathwise stability estimates with respect to the underlying model parameters. We also derive pathwise error estimates arising from the time-discretization of the log-optimal portfolio and its associated capital process.
\end{abstract}

\maketitle




\section{Introduction}

A central challenge in mathematical finance, financial economics, and related fields is to understand the decision making of rational agents facing financial markets with their random evolution of asset prices. A major approach to this challenge, initiated by Merton \cite{Merton1969}, is the study of utility maximization problems in continuous-time markets. By now, a vast number of researchers have contributed to this approach and investigated various facets of utility maximization problems; see, e.g., \cite{Jarrow2018} and the references therein. For instance, a substantial body of research is devoted to constructing closed-form solutions, which are of particular interest from a practitioner's perspective; see, e.g., \cite{Merton1971,Kim1996,Zariphopoulou1999,Kraft2005,Guasoni2012}.

In classical portfolio theory, the utility maximization problems are considered and solved with the implicit assumption that the underlying model for the asset prices is perfectly specified, that is, the model parameters (trend and volatility) are fully known. Consequently, essentially all ``optimal'' portfolios in the literature depend on the underlying model parameters. However, in reality, due to the necessity to use statistical estimation to determine the underlying models, there is always a natural uncertainty about the model parameters and, even worse, about the underlying model itself. In particular, estimating the trend of the time-evolution of an asset price on a financial market is known to be a notoriously difficult problem; cf.~\cite{Rogers2001}. Hence, to deal with model uncertainty and to understand its implications is of utmost importance in portfolio theory.

Various approaches have been developed in mathematical finance to treat model uncertainty in the context of portfolio optimization. Let us briefly mention the three major research areas which are most related to the present work. The \emph{sensitivity analysis} of utility maximization problems is based on classical probabilistic modelling, and studies the impact of model perturbations to decision making; see, e.g., \cite{Kardaras2011,Larsen2007,Weston2016,Mostovyi2019}. \emph{Robust} portfolio theory does not fix a fully specified underlying model---instead, it introduces a ``worst-case'' approach, also called the Knightian approach, aiming to solve utility maximization problems simultaneously for a family of models; see, e.g., \cite{Tevzadze2013,Biagini2017,Neufeld2018,Pham2022}. In \emph{model-free} portfolio theory, portfolios are constructed without any underlying probabilistic framework, and their performance is analyzed in an entirely pathwise manner; see, e.g., \cite{Schied2018,Cuchiero2019,Karatzas2020,Allan2023a}.

In this paper we develop a methodology for the pathwise analysis of portfolios, applied to individual price trajectories generated by standard models of financial markets. This yields, in particular, a model-free approach to the classical sensitivity analysis of utility maximization problems, enabling a clear distinction between Knightian uncertainty and the intrinsic randomness inherent in the underlying mathematical problem. Our primary focus is on the pathwise stability of portfolios with respect to model parameters, as this appears to be a stability property of paramount importance for both the practical and mathematical perspectives; see, e.g., \cite{Weston2016} and the discussion therein. That said, we emphasize that the proposed methodology is sufficiently general to handle pathwise stability with respect to various types of perturbations, including changes in the noise process and initial prices, as well as discretization schemes. As an illustration of this flexibility, we analyze the time-discretization error of portfolios, as detailed below.

As a prototypical example of an ``optimal'' portfolio, we investigate the log-optimal portfolio of a classical investment-consumption optimization problem in a frictionless financial market, modelled by an It{\^o} diffusion process; see \cite{Merton1969,Merton1971}. Let us briefly recall this classical utility maximization problem. We assume that the discounted price process $(\bar{S}_t)_{t \in [0,T]}$ is given by
\begin{equation*}
\bar{S}_t = s_0 + \int_0^t \bar{b}_s \dd s + \int_0^t \bar{\sigma}_s \dd \bar{W}_s, \qquad t \in [0,T],
\end{equation*}
where $\bar{b}$, $\bar{\sigma}$ are suitable predictable processes, and $\bar{W}$ is a Brownian motion. It is well-known that there exists a log-optimal portfolio $(\bar{\phi}, \bar{\kappa})$ given a consumption clock $K$, that is,
\begin{equation*}
\E \bigg[ \int_0^T \log (\bar{\kappa}_t) \dd K_t \bigg] = \sup_{(\bar{\oldphi},\bar{\chi})} \E \bigg[ \int_0^T \log (\bar{\chi}_t) \dd K_t \bigg],
\end{equation*}
where the supremum is taken over all admissible portfolios $(\bar{\oldphi}, \bar{\chi})$; see, e.g., \cite{Goll2000, Goll2003}. Before presenting our pathwise analysis of the log-optimal portfolio, we would like to emphasize that, based on the developed methodology, an analogous pathwise analysis can be carried out for numerous portfolios which are known to be ``optimal'' from classical portfolio theory.

As a foundation, we set up a suitable pathwise It{\^o}-type integration, relying on the theory of c{\`a}dl{\`a}g rough paths, as in \cite{Friz2017,Friz2018}, and, more specifically, the so-called Property~\textup{(RIE)} as introduced in \cite{Perkowski2016, Allan2023b}. While rough path theory, of course, provides a comprehensive theory of rough integration, some new results essential in the specific context of mathematical finance are required, and care is needed to obtain the natural economic interpretation of all involved integrals and related objects. In particular, assuming that a ``noise'' path $W$ satisfies Property~\textup{(RIE)}, the discounted price path $(S_t)_{t \in [0,T]}$ can be modelled by the (rough) differential equation
\begin{equation}\label{eq:price intro}
S_t = s_0 + \int_0^t \hat{b}_s \dd s + \int_0^t \hat{\sigma}_s \dd W_s, \qquad t \in [0,T],
\end{equation}
where $\hat{b}$, $\hat{\sigma}$ are suitable paths. The rationale behind the deterministic price dynamics \eqref{eq:price intro} is that $W$ corresponds to a fixed realization of the noise, modelling the randomness of price processes, and $\hat{b}, \hat{\sigma}$ are the model parameters specifying the dynamics of the underlying model adopted for the asset prices. Hence, \eqref{eq:price intro} provides a transparent distinction of model uncertainty and randomness. In this paper, we work in each of the following settings for \eqref{eq:price intro}:
\begin{itemize}
  \item[$\vcenter{\hbox{\scriptsize$\bullet$}}$] local volatility models: $\hat{b}_s = b(s,S_s) $ and $\hat{\sigma}_s = \sigma(s,S_s)$, with $b \in C^3_b(\R^{m+1};\R^m)$ and $\sigma \in C^3_b(\R^{m+1};\mathcal{L}(\R^d;\R^m))$,
  \item[$\vcenter{\hbox{\scriptsize$\bullet$}}$] Black--Scholes-type models: $ \hat{b}_s= b_s S_s$ and $\hat{\sigma}_s = \sigma_s S_s$, where the paths $b$ and $\sigma$ are controlled (in the rough path sense) with respect to $W$.
\end{itemize}
We note that it is necessary to differentiate between these two settings, since rough differential equations with unbounded coefficients are a delicate challenge, and can only be treated in specific situations; see, e.g., \cite{Lejay2012}.

Based on the developed pathwise It{\^o}-type integration, we can construct the log-optimal portfolio of Merton's investment-consumption problem entirely pathwise, given the model parameters $b$, $\sigma$ and a fixed ``noise'' path $W$. Moreover, assuming that the noise paths $W$ are realizations of a Brownian motion $\bar{W}$, the pathwise construction of the log-optimal portfolio is, indeed, a solution to Merton's classical investment-consumption problem in a frictionless financial market, modelled by an It{\^o} diffusion process. We therefore call this pathwise constructed portfolio the \emph{pathwise log-optimal portfolio} $(\phi, \kappa)$. However, let us remark that for the construction of this portfolio, as well as its pathwise analysis, the ``noise'' path $W$ can be a rather general deterministic path, and does not need to be a sample path of any particular stochastic process.

The present pathwise framework and the pathwise construction of the log-optimal portfolio allow us, in particular, to analyze the dependency of the pathwise log-optimal portfolio on the model parameters for a fixed noise path. Relying on continuity estimates for rough integration and rough differential equations, we prove that the pathwise log-optimal portfolio and its associated capital process depend in a locally Lipschitz continuous way on the model parameters $b$, $\sigma$. For instance, in the case of local volatility models, the stability of the pathwise log-optimal portfolio and its associated capital process read as follows:
\begin{equation*}
\|(\phi^{(b,\sigma)}, \kappa^{(b,\sigma)}) - (\phi^{(\tilde{b},\tilde{\sigma})}, \kappa^{(\tilde{b},\tilde{\sigma})})\|_\infty \lesssim \|b - \tilde{b}\|_{C^2_b} + \|\sigma - \tilde{\sigma}\|_{C^2_b}
\end{equation*}
and
\begin{equation*}
\|V^{(b,\sigma)} - V^{(\tilde{b},\tilde{\sigma})}\|_\infty \lesssim \|b - \tilde{b}\|_{C^2_b} + \|\sigma - \tilde{\sigma}\|_{C^2_b},
\end{equation*}
where $(\phi^{(b,\sigma)}, \kappa^{(b,\sigma)})$, $(\phi^{(\tilde{b},\tilde{\sigma})},\kappa^{(\tb,\tsigma)})$ denote the pathwise log-optimal portfolios, and $V^{(b,\sigma)}$, $V^{(\tilde{b},\tilde{\sigma})}$ the associated capital processes, given the model parameters $b$, $\sigma$ and $\tilde{b}$, $\tilde{\sigma}$, respectively. The precise statements of the pathwise stability estimates with respect to model parameters can be found in Sections~\ref{subsec: stability in local vola model} and \ref{subsec: stability in Black Scholes model}.

Model uncertainty is, of course, not the only key challenge when aiming to implement a theoretically optimal portfolio on a real financial market. Another major obstacle is the necessary time-discretization of portfolios, and trading strategies in general; see, e.g., \cite{Gobet2001,Rosenbaum2014,Fukasawa2014,Cai2016}. Indeed, while trading can be done at very high frequency, there is still some gap between high-frequency trading and continuous-time trading, and, for various reasons, there is often a desire to rebalance a portfolio at a lower frequency. With this in mind, we prove the convergence of the time-discretized version of the pathwise log-optimal portfolio to its continuous-time counterpart, as well as the convergence of the associated capital processes, along with quantitative bounds for the discretization error. The precise estimates for the discretization errors can be found in Sections~\ref{subsec: discretization in local vola model} and \ref{subsec: discretization in Black Scholes model}.

\medskip

\noindent \textbf{Organization of the paper:} Section~\ref{sec: portfolio in probabilistic setting} presents the classical investment-consumption optimization problem in a probabilistic setting. In Section~\ref{sec: pathwise analysis}, we recall some essential background from rough path theory and set up the pathwise approach to stochastic It{\^o} integration. In the case of price trajectories generated by local volatility models, the pathwise analysis of the log-optimal portfolio is developed in Section~\ref{sec: local vol models}, and, in the case of price trajectories generated by Black--Scholes-type models, in Section~\ref{sec: BS model}. Appendix~\ref{sec: appendix} establishes several elementary results in the theory of c{\`a}dl{\`a}g rough paths.

\medskip

\noindent\textbf{Acknowledgments:} A.~P.~Kwossek and D.~J.~Pr{\"o}mel gratefully acknowledge financial support by the Baden-W{\"u}rttemberg Stiftung. A.~P.~Kwossek  was affiliated with the University of Mannheim for the majority of this project's duration.

\section{Portfolio optimization in a probabilistic setting}\label{sec: portfolio in probabilistic setting}

Before setting up a pathwise stability analysis of optimal portfolios, let us recall the classical formulation of (and the well-known solution to) an optimal investment-consumption problem in a probabilistic setting {\`a} la Merton \cite{Merton1971}. For this purpose, we fix an underlying probability space $(\Omega, \mathcal{F}, \P)$ with a filtration $(\mathcal{F}_t)_{t \in [0,T]}$ satisfying the usual conditions, i.e., completeness and right-continuity.

Following \cite{Goll2000,Goll2003}, we consider an optimal investment-consumption problem, where ``optimal'' refers to the maximization of the expected $\log$-utility from the investor's consumption over a finite time horizon $T > 0$. In the next subsection, we give a precise formulation of the investment-consumption problem.

\subsection{A classical investment-consumption optimization problem}\label{subsec: setting}

The underlying frictionless financial market consists of $m+1$ assets, where the discounted price process $\bar{S} = (\bar{S}^0, \bar{S}^1, \ldots, \bar{S}^m) = (\bar{S}_t)_{t \in [0,T]}$ is an $\R^{m+1}$-valued $\cF_t$-adapted semimartingale on $(\Omega, \mathcal{F}, \P)$, with $\bar{S}^0 \equiv 1$.

Following \cite{Goll2000}, we adopt the following standard setup:
\begin{itemize}
  \item[$\vcenter{\hbox{\scriptsize$\bullet$}}$] A self-financing trading strategy $\bar{\phi} \in L(\bar{S})$, where $L(\bar{S})$ denotes the space of all $\bar{S}$-integrable predictable processes in the spirit of stochastic It{\^o} integration, is called \emph{admissible}, and denoted by $\bar{\phi} \in \mathfrak{S}$, if $\int_0^t \bar{\phi}_s^\top \dd \bar{S}_s \geq -1$ for all $t \in [0,T]$, $\P$-almost surely, where $(\cdot)^\top$ denotes matrix transposition.
  \item[$\vcenter{\hbox{\scriptsize$\bullet$}}$] The \emph{consumption clock} $K \colon [0,T] \to \R$ is an increasing deterministic c{\`a}dl{\`a}g function, and $\mathfrak{K}$ denotes the set of all non-negative optional processes $\bar{\kappa}$, called the \emph{consumption rate}, such that $\int_0^t \bar{\kappa}_s \dd K_s < \infty$ for all $t \in [0,T]$, $\P$-almost surely. For $\bar{\kappa} \in \mathfrak{K}$, the consumption process is given by $\int_0^\cdot \bar{\kappa}_s \dd K_s$.
  \item[$\vcenter{\hbox{\scriptsize$\bullet$}}$] A pair $(\bar{\phi}, \bar{\kappa}) \in \mathfrak{S} \times \mathfrak{K}$ belongs to the set $\mathfrak{P}$ of \emph{admissible portfolios} if the discounted wealth process $(\bar{V}_t)_{t\in [0,T]}$, given by
  \begin{equation}\label{eq: defn wealth process}
    \bar{V}_t(\bar{\phi}, \bar{\kappa}) := 1 + \int_0^t \bar{\phi}_s^\top \dd \bar{S}_s - \int_0^t \bar{\kappa}_s \dd K_s, \qquad t \in [0,T],
  \end{equation}
  is non-negative, $\P$-almost surely.
\end{itemize}

Typical choices of the consumption clock are $K_t = \1_{\{T\}}(t)$, i.e., consumption only at time $T$, and $K_t = \sum_{s \leq t} \1_{\N}(s) = \lfloor t \rfloor$, i.e., consumption only at integer times.

Occasionally, as will become apparent, we will identify $\bar{S}$ with the $\R^m$-valued process $(\bar{S}^1, \dots, \bar{S}^m)$, and similarly for $\bar{\phi}$.

\medskip

A pair $(\bar{\phi}, \bar{\kappa}) \in \mathfrak{P}$ is called a \emph{log-optimal portfolio} if $(\bar{\phi}, \bar{\kappa})$ maximizes the map $\Phi_{\log} \colon \mathfrak{P} \to \R$, given by
\begin{equation*}
(\bar{\oldphi}, \bar{\chi}) \mapsto \E \bigg[\int_0^T \log (\bar{\chi}_t) \dd K_t \bigg],
\end{equation*}
over all $(\bar{\oldphi}, \bar{\chi}) \in \mathfrak{P}$.

\subsection{The log-optimal portfolio for the investment-consumption problem}

Finding log-optimal portfolios in the context of expected utility maximization is a well-studied mathematical problem; see, e.g., \cite{Korn1997} for a classical introduction. For example, in a general semimartingale framework, the works of Goll and Kallsen \cite{Goll2000, Goll2003} provide explicit formulae in terms of semimartingale characteristics. In the following, we recall the result of \cite{Goll2000} in the case that the discounted price process is modelled by an It{\^o} process.

Let $\bar{W} = (\bar{W}_t)_{t \in [0,T]}$ be a $d$-dimensional Brownian motion defined on $(\Omega, \cF, \P)$ with respect to $(\cF_t)_{t \in [0,T]}$, and suppose that the discounted price process $(\bar{S}_t)_{t \in [0,T]}$ is given by
\begin{equation*}
  \bar{S}_t = s_0 + \int_0^t \bar{b}_s \dd s + \int_0^t \bar{\sigma}_s \dd \bar{W}_s, \qquad t \in [0,T],
\end{equation*}
where $s_0 \in \R^m$, $\bar{b}$ is a predictable, locally integrable $\R^m$-valued process, and $\bar{\sigma}$ is a predictable, locally square integrable $\R^{m \times d}$-valued process such that $\bar{\sigma}_t \bar{\sigma}_t^\top$ is a positive definite $m \times m$-matrix for every $t$, where each coefficient is bounded away from zero.

In the present setting, \cite[Theorem~3.1]{Goll2000}, which formulates the solution to the optimal investment-consumption problem, reads as follows.

\begin{theorem}\label{thm: Goll-Kallsen's main theorem Ito processes}
  Assume that there exists an $\R^m$-valued process $\bar{H} \in L(\bar{S})$ such that
  \begin{equation*}
    \bar{b}_t - \bar{c}_t \bar{H}_t = 0 \qquad \text{with} \qquad \bar{c}_t := \bar{\sigma}_t \bar{\sigma}_t^\top, \qquad t \in [0,T],
  \end{equation*}
  holds $\P \otimes \d t$-almost everywhere, and set
  \begin{align*}
    &\bar{\kappa}_t := \frac{1}{K_T} \bar{\cE} \bigg( \int_0^\cdot \bar{H}^\top_s \dd \bar{S}_s \bigg)_{\hspace{-2pt}t}, \qquad \bar{V}_t := \bar{\kappa}_t (K_T - K_t),\\
    &\bar{\phi}_t^i := \bar{H}_t^i \bar{V}_{t-}, \quad i = 1, \ldots, m, \qquad \bar{\phi}^0_t := \int_0^t \bar{\phi}_s^\top \dd \bar{S}_s - \sum_{i=1}^m \bar{\phi}^i_t \bar{S}^i_t, \qquad t \in [0,T],
  \end{align*}
  where we set $\bar{V}_{0-} := 0$, and $\bar{\cE}$ denotes the stochastic exponential. Then, $(\bar{\phi}, \bar{\kappa}) \in \mathfrak{P}$ is a log-optimal portfolio with discounted wealth process $(\bar{V}_t)_{t \in [0,T]}$.
\end{theorem}

\begin{remark}\label{remark: Goll Kallsen}
  If the price process $\bar{S}$ is given by a linear stochastic differential equation, i.e., if $\bar{b}^i = \bar{S}^i \hat{b}^i$ and $\bar{\sigma}^{i,j} = \bar{S}^i \hat{\sigma}^{ij}$, for some predictable $\hat{b}^i$, $\hat{\sigma}^{ij}$, $i = 1, \ldots, m$, $j = 1, \ldots, d$, then the previous theorem can be rephrased as follows; see also \cite[Example~4.2]{Goll2000}.

  Assume that there exists a predictable, $\R^m$-valued process $\bar{h}$ such that
  \begin{equation*}
    \hat{b}_t - \hat{c}_t \bar{h}_t = 0 \qquad \text{with} \qquad \hat{c}_t = \hat{\sigma}_t \hat{\sigma}_t^\top, \qquad t \in [0,T],
  \end{equation*}
  holds $\P \otimes \d t$-almost everywhere, and set
  \begin{equation*}
    \bar{H}^i_t := \frac{\bar{h}^i_t}{\bar{S}^i_t}, \qquad t \in [0,T], \qquad i = 1, \ldots, m,
  \end{equation*}
  and $\bar{\kappa}$, $\bar{V}$, $\bar{\phi}^i$, $i = 0, \ldots, m$, as defined in Theorem~\ref{thm: Goll-Kallsen's main theorem Ito processes}. Then, $(\bar{\phi},\bar{\kappa}) \in \mathfrak{P}$ is a log-optimal portfolio with discounted wealth process $(\bar{V}_t)_{t \in [0,T]}$.
\end{remark}

\section{Pathwise stochastic analysis}\label{sec: pathwise analysis}

Developing a methodology that allows for a pathwise analysis of optimal portfolios requires, unsurprisingly, an underlying pathwise framework. To that end, we rely on the theory of rough paths---see, e.g., \cite{FrizHairer2020} for an introductory textbook---and, more specifically, the so-called Property \textup{(RIE)} as introduced in \cite{Perkowski2016, Allan2023b}, which provides a suitable foundation for the use of rough path theory in mathematical finance. We start by recalling some essentials from the theory of c{\`a}dl{\`a}g rough paths. For a more comprehensive introduction we refer to \cite{Friz2017,Friz2018}.

\subsection{Essentials of rough path theory}\label{subsec: rough path theory}

A \emph{partition} $\mathcal{P}$ of an interval $[s,t]$ is a finite set of points between and including the points $s$ and $t$, i.e., $\mathcal{P} = \{s = u_0 < u_1 < \cdots < u_N = t\}$ for some $N \in \N$, and its mesh size is denoted by $|\mathcal{P}| := \max\{|u_{i+1} - u_i| \, : \, i = 0, \ldots, N-1\}$.

Throughout, we let $T > 0$ be a fixed finite time horizon. We let $\Delta_T := \{(s,t) \in [0,T]^2 \, : \, s \leq t\}$ denote the standard $2$-simplex. A function $w \colon \Delta_T \to [0,\infty)$ is called a \emph{control function} if it is superadditive, in the sense that $w(s,u) + w(u,t) \leq w(s,t)$ for all $0 \leq s \leq u \leq t \leq T$. For two vectors $x = (x^1, \ldots, x^d)^\top, y = (y^1, \ldots, y^d)^\top \in \R^d$ we use the usual tensor product
\begin{equation*}  
  x \otimes y := (x^i y^j)_{i, j = 1, \ldots, d} \in \R^{d \times d}.
\end{equation*}
Whenever $(B,\|\hspace{0.5pt}\cdot\hspace{0.5pt}\|)$ is a normed space and $f, g \colon B \to \R$ are two functions on $B$, we shall write $f \lesssim g$ or $f \leq C g$ to mean that there exists a constant $C > 0$ such that $f(x) \leq C g(x)$ for all $x \in B$. The constant $C$ may depend on the normed space, e.g., through its dimension or regularity parameters.

\medskip

For two vector spaces, the space of linear maps from $E_1 \to E_2$ is denoted by $\cL(E_1; E_2)$, and we write, e.g., $C^k_b = C^k_b(\R^n;\mathcal{L}(\R^d;\R^m))$ for the space of $m$-times differentiable (in the Fr{\'e}chet sense) functions $f \colon \R^n \to \mathcal{L}(\R^d;\R^m)$ such that $f$ and all its derivatives up to order $k$ are continuous and bounded. We equip this space with the norm
\begin{equation*}
  \|f\|_{C^k_b} := \|f\|_\infty + \|\D f\|_\infty + \cdots + \|\D^k f\|_\infty,
\end{equation*}
where $\D^r f$ denotes the $r$-th order derivative of $f$, and $\|\hspace{0.5pt}\cdot\hspace{0.5pt}\|_{\infty}$ denotes the supremum norm on the corresponding spaces of operators.

\medskip

For a normed space $(E,|\hspace{0.5pt}\cdot\hspace{0.5pt}|)$, we let $D([0,T];E)$ denote the set of c{\`a}dl{\`a}g (right-continuous with left-limits) paths from $[0,T] \to E$. For $X \in D([0,T];E)$, the supremum norm of the path $X$ is given by
\begin{equation*}
  \|X\|_{\infty} := \sup_{t \in [0,T]} |X_t|,
\end{equation*}
and, for $p \geq 1$, the $p$-variation of the path $X$ is given by
\begin{equation*}
  \|X\|_p := \|X\|_{p,[0,T]} \qquad \text{with} \qquad \|X\|_{p,[s,t]} := \bigg( \sup_{\mathcal{P} \subset [s,t]} \sum_{[u,v] \in \mathcal{P}} |X_v - X_u|^p \bigg)^{\hspace{-2pt}\frac{1}{p}}, \quad (s,t) \in \Delta_T,
\end{equation*}
where the supremum is taken over all possible partitions $\mathcal{P}$ of the interval $[s,t]$. We recall that, given a path $X$, we have that $\|X\|_p < \infty$ if and only if there exists a control function $w$ such that\footnote{Here and throughout, we adopt the convention that $\frac{0}{0} := 0$.}
\begin{equation*}
  \sup_{(u,v) \in \Delta_T} \frac{|X_v - X_u|^p}{w(u,v)} < \infty.
\end{equation*}
We write $D^p = D^p([0,T];E)$ for the space of paths $X \in D([0,T];E)$ which satisfy $\|X\|_p < \infty$.

\medskip

For a path $X \in D([0,T];E)$, we will use the shorthand notation:
\begin{equation*}
  X_{s,t} := X_t - X_s \quad \text{and} \quad X_{t-} := \lim_{u \nearrow t} X_u, \qquad \text{for} \quad (s,t) \in \Delta_T.
\end{equation*}
For $r \geq 1$ and a two-parameter function $\mathbb{X} \colon \Delta_T \to E$, we similarly define
\begin{equation*}
  \|\mathbb{X}\|_r := \|\mathbb{X}\|_{r,[0,T]} \qquad \text{with} \qquad \|\mathbb{X}\|_{r,[s,t]} := \bigg(\sup_{\mathcal{P} \subset [s,t]} \sum_{[u,v] \in \mathcal{P}} |\mathbb{X}_{u,v}|^r\bigg)^{\hspace{-2pt}\frac{1}{r}}, \quad (s,t) \in \Delta_T.
\end{equation*}
We write $D_2^r = D_2^r(\Delta_T;E)$ for the space of all functions $\mathbb{X} \colon \Delta_T \to E$ which satisfy $\|\X\|_r < \infty$, and are such that the maps $s \mapsto \mathbb{X}_{s,t}$ for fixed $t$, and $t \mapsto \mathbb{X}_{s,t}$ for fixed $s$, are both c{\`a}dl{\`a}g.

\medskip

For $p \in [2,3)$, a pair $\bX = (X,\X)$ is called a \emph{c{\`a}dl{\`a}g rough path} over $\R^d$ if
\begin{enumerate}
  \item[(i)] $X \in D^p([0,T];\R^d)$ and $\X \in D_2^{\frac{p}{2}}(\Delta_T;\R^{d \times d})$, and
  \item[(ii)] Chen's relation: $\mathbb{X}_{s,t} = \mathbb{X}_{s,u} + \mathbb{X}_{u,t} + X_{s,u} \otimes X_{u,t}$ holds for all $0 \leq s \leq u \leq t \leq T$.
\end{enumerate}
In component form, condition (ii) states that $\mathbb{X}^{ij}_{s,t} = \mathbb{X}^{ij}_{s,u} + \mathbb{X}^{ij}_{u,t} + X^i_{s,u} X^j_{u,t}$ for every $i$ and~$j$. We will denote the space of c{\`a}dl{\`a}g rough paths by $\cD^p = \cD^p([0,T];\R^d)$. On the space $\cD^p([0,T];\R^d)$, we use the natural seminorm
\begin{equation*}
  \|\bX\|_{p} := \|\bX\|_{p,[0,T]} \qquad \text{with} \qquad \|\bX\|_{p,[s,t]} := \|X\|_{p,[s,t]} + \|\X\|_{\frac{p}{2},[s,t]}
\end{equation*}
for $(s,t) \in \Delta_T$, and the induced distance
\begin{equation*}
  \|\bX;\tbX\|_p :=  \|\bX;\tbX\|_{p,[0,T]} \qquad \text{with} \qquad \|\bX;\tbX\|_{p,[s,t]} := \|X - \tX\|_{p,[s,t]} + \|\X - \tbbX\|_{\frac{p}{2},[s,t]},
\end{equation*}
whenever $\bX = (X,\X), \tbX = (\tX,\tbbX) \in \cD^p([0,T];\R^d)$. Recall that the rough path $\bX = (X,\X)$ above a path $X$ is not unique.

\medskip

Let $X \in D^p([0,T];\R^d)$ for some $p \in [2,3)$. We say that a pair $(Y,Y')$ is a \emph{controlled path} (with respect to $X$), if
\begin{equation*}
Y \in D^p([0,T];E), \quad Y' \in D^p([0,T];\cL(\R^d;E)) \quad \text{and} \quad R^Y \in D_2^{\frac{p}{2}}(\Delta_T;E),
\end{equation*}
where $R^Y$ is defined by
\begin{equation*}
  Y_{s,t} = Y'_s X_{s,t} + R^Y_{s,t} \qquad \text{for all} \quad (s,t) \in \Delta_T.
\end{equation*}
We write $\cV^p_X = \cV^p_X([0,T];E)$ for the space of $E$-valued controlled paths, which becomes a Banach space when equipped with the norm $(Y,Y') \mapsto \|Y,Y'\|_{\crpX}$, where
\begin{equation*}
  \|Y,Y'\|_{\crpX} := \|Y,Y'\|_{\crpX,[0,T]},
\end{equation*}
with
\begin{equation*}
  \|Y,Y'\|_{\crpX,[s,t]} := |Y_s| + |Y'_s| + \|Y'\|_{p,[s,t]} + \|R^Y\|_{\frac{p}{2},[s,t]},
\end{equation*}
for $(s,t) \in \Delta_T$. It is straightforward to see that
\begin{equation*}
  \|Y\|_p \leq \|Y'\|_\infty \|X\|_p + \|R^Y\|_{\frac{p}{2}} \qquad \text{and} \qquad \|Y'\|_\infty \leq |Y'_0| + \|Y'\|_p,
\end{equation*}
so that in particular
\begin{equation}\label{eq: bound on sup-norm}
  \|Y\|_\infty \leq (1 + \|X\|_p) \|Y,Y'\|_{\crpX}.
\end{equation}

Throughout the paper, for notational brevity, we will usually leave the Gubinelli derivative $Y'$ as implicit, and write simply $Y \in \crpX$, and use $\|Y\|_{\crpX}$ as shorthand for $\|Y,Y'\|_{\crpX}$.

We further introduce the standard ``distance''
\begin{equation*}
  \|Y;\tY\|_{\crpX,\crptX} := \|Y;\tY\|_{\crpX,\crptX,[0,T]}
\end{equation*}
with
\begin{equation*}
\|Y;\tY\|_{\crpX,\crptX,[s,t]} := |Y_s - \tY_s| + |Y'_s - \tY'_s| + \|Y' - \tY'\|_{p,[s,t]} + \|R^Y - R^{\tY}\|_{\frac{p}{2},[s,t]}
\end{equation*}
for $(s,t) \in \Delta_T$, whenever $(Y,Y') \in \crpX$, $(\tY,\tY') \in \crptX$. Note that, in general, $\crpX$ and $\crptX$ are different Banach spaces; if $X = \tX$, we will write $\|Y;\tY\|_{\crpX} := \|Y;\tY\|_{\crpX,\crpX}$.

We also note that
\begin{equation}\label{eq: bound on p var norm}
\|Y - \tY\|_p \leq C (\|Y;\tY\|_{\crpX,\crptX} + \|X - \tX\|_p)
\end{equation}
for some constant $C$ which depends only on $\|Y\|_{\crpX}$, $\|\tY\|_{\crptX}$, $\|X\|_p$ and $\|\tX\|_p$.

\medskip

Given $p \in (2,3)$, $\bX = (X,\X) \in \cD^p([0,T];\R^d)$ and $(Y,Y') \in \crpX([0,T];\cL(\R^d;\R^m))$, the (forward) rough integral
\begin{equation}\label{eq: rough integral}
  \int_s^t Y_u \dd \bX_u := \lim_{|\cP^n| \to 0} \sum_{[u,v] \in \cP^n} (Y_u X_{u,v} + Y'_u \X_{u,v}), \qquad (s,t) \in \Delta_T,
\end{equation}
exists (in the classical mesh Riemann--Stieltjes sense), where the limit is taken along any sequence of partitions $(\cP^n)_{n \in \N}$ of the interval $[s,t]$ such that $|\cP^n| \to 0$ as $n \to \infty$. To be precise, in writing the product $Y_u X_{u,v}$, we apply the operator $Y_u \in \cL(\R^d;\R^m)$ to $X_{u,v} \in \R^d$, and in writing the product $Y'_u \X_{u,v}$, we use the natural identification of $\cL(\R^d;\cL(\R^d;\R^m))$ with $\cL(\R^d \otimes \R^d;\R^m)$. The rough integral comes with the estimate
\begin{equation*}
  \bigg| \int_s^t Y_u \dd \bX_u - Y_s X_{s,t} - Y'_s \X_{s,t} \bigg| \leq C \Big(\|R^Y\|_{\frac{p}{2},[s,t)} \|X\|_{p,[s,t]} + \|Y'\|_{p,[s,t)} \|\X\|_{\frac{p}{2},[s,t]}\Big)
\end{equation*}
for some constant $C$ which depends only on $p$ (see, e.g., \cite[Proposition~2.4 and Remark~2.5]{Allan2023b}), where
\begin{equation*}
  \|Y'\|_{p,[s,t)} := \sup_{u < t} \|Y'\|_{p,[s,u]} \qquad \text{and} \qquad \|R^Y\|_{\frac{p}{2},[s,t)} := \sup_{u < t} \|R^Y\|_{\frac{p}{2},[s,u]}.
\end{equation*}
This implies that $(\int_0^\cdot Y_u \dd \bX_u, Y) \in \crpX$ is a controlled path with respect to $X$, and satisfies
\begin{equation}\label{eq: estimate for rough integral}
  \bigg\|\int_0^\cdot Y_u \dd \bX_u\bigg\|_{\crpX} \leq C,
\end{equation}
for some constant $C$ depending only on $p, \|Y\|_{\crpX}$ and $\|\bX\|_p$.

\medskip

Given a rough path $\bX = (X,\X) \in \cD^p([0,T];\R^d)$ with $p \in [2,3)$, there exists a unique controlled path $(Y,Y') \in \cV^p_X([0,T];\R^m)$ satisfying the rough differential equation (RDE)
\begin{equation*}
  Y_t = y_0 + \int_0^t b(s,Y_s) \dd s + \int_0^t \sigma(s,Y_s) \dd \bX_s, \qquad t \in [0,T],
\end{equation*}
if $b \in C_b^2(\R^{m+1};\R^m)$ and $\sigma \in C_b^3(\R^{m+1};\mathcal{L}(\R^d;\R^m))$; see, e.g., \cite[Theorem~2.2]{Allan2023c}.

\subsection{Pathwise It{\^o}-type integration}\label{subsec: pathwise ito integration}

Rough path theory provides a pathwise approach to stochastic integration and stochastic differential equations. In particular, it allows to recover the stochastic It\^o and Stratonovich integrals by choosing the corresponding rough path lift of a semimartingale. Consequently, from a financial modelling perspective, choosing the rough path above a given path without care could create arbitrage. Moreover, the definition of the rough integral~\eqref{eq: rough integral} lacks a canonical interpretation in mathematical finance. To overcome these issues, we rely on the so-called Property \textup{(RIE)}, as introduced in \cite{Perkowski2016,Allan2023b}.

\begin{RIE}
  Let $p \in (2,3)$ and let $\cP^n = \{0 = t^n_0 < t^n_1 < \cdots < t^n_{N_n} = T\}$, $n \in \N$, be a sequence of partitions of the interval $[0,T]$ such that $|\mathcal{P}^n| \to 0$ as $n \to \infty$. For $X \in D([0,T];\R^d)$, and each $n \in \N$, we define $X^n \colon [0,T] \to \R^d$ by
  \begin{equation*}
    X^n_t = X_T \1_{\{T\}}(t) + \sum_{k=0}^{N_n - 1} X_{t^n_k} \1_{[t^n_k,t^n_{k+1})}(t), \qquad t \in [0,T].
  \end{equation*}
  We assume that:
  \begin{enumerate}
    \item[(i)] the sequence of paths $(X^n)_{n \in \N}$ converges uniformly to $X$ as $n \to \infty$,
    \item[(ii)] the Riemann sums $\int_0^t X^n_u \otimes \d X_u := \sum_{k=0}^{N_n-1} X_{t^n_k} \otimes X_{t^n_k \wedge t,t^n_{k+1} \wedge t}$ converge uniformly as $n \to \infty$ to a limit, which we denote by $\int_0^t X_u \otimes \d X_u$, $t \in [0,T]$,
    \item[(iii)] and there exists a control function $w$ such that
	\begin{equation}\label{eq: RIE inequality}
      \sup_{(s,t) \in \Delta_T} \frac{|X_{s,t}|^p}{w(s,t)} + \sup_{n \in \N} \, \sup_{0 \leq k < \ell \leq N_n} \frac{|\int_{t^n_k}^{t^n_\ell} X^n_u \otimes \d X_u - X_{t^n_k} \otimes X_{t^n_k,t^n_\ell}|^{\frac{p}{2}}}{w(t^n_k,t^n_\ell)} \leq 1.
    \end{equation}
  \end{enumerate}
\end{RIE}

We say that a path $X \in D([0,T];\R^d)$ satisfies Property \textup{(RIE)} relative to $p$ and $(\mathcal{P}^n)_{n \in \N}$, if $p$, $(\mathcal{P}^n)_{n \in \N}$ and $X$ together satisfy Property \textup{(RIE)}.

\medskip

It is known that, if a path $X \in D([0,T];\R^d)$ satisfies Property \textup{(RIE)}, then $X$ extends canonically to a rough path $\mathbf{X} = (X,\mathbb{X}) \in \mathcal{D}^p([0,T];\R^d)$, where the lift $\X$ is defined by
\begin{equation}\label{eq: RIE rough path}
  \mathbb{X}_{s,t} := \int_s^t X_u \otimes \d X_u - X_s \otimes (X_t - X_s), \qquad (s,t) \in \Delta_T,
\end{equation}
with $\int_s^t X_u \otimes \d X_u := \int_0^t X_u \otimes \d X_u - \int_0^s X_u \otimes \d X_u$, and the existence of the integral $\int_0^t X_u \otimes \d X_u$ is ensured by condition (ii) of Property \textup{(RIE)}; see \cite[Lemma~2.13]{Allan2023b}. When assuming Property \textup{(RIE)} for a path $X$, we will always work with the rough path $\mathbf{X} = (X,\mathbb{X})$ defined via~\eqref{eq: RIE rough path}, and note that $\mathbf{X} = (X,\mathbb{X})$ corresponds to the It{\^o} rough path lift of a stochastic process, since the ``iterated integral'' $\mathbb{X}$ is given as a limit of left-point Riemann sums, analogously to the stochastic It{\^o} integral.

\medskip

Property \textup{(RIE)} not only ensures the existence of a suitable rough path lift of a path, but also allows the rough integral to be expressed as a classical limit of Riemann sums. Consequently, the rough integral possesses the natural interpretation in a financial context as the capital process of a portfolio. The next theorem is a slight generalization of \cite[Theorem~2.15]{Allan2023b}.

\begin{theorem}\label{thm: rough int against controlled path under RIE}
Let $p \in (2,3)$, and let $\cP^n = \{0 = t^n_0 < t^n_1 < \cdots < t^n_{N_n} = T\}$, $n \in \N$, be a sequence of partitions such that $|\cP^n| \to 0$ as $n \to \infty$. Suppose that $X \in D([0,T];\R^d)$ satisfies Property \textup{(RIE)} relative to $p$ and $(\cP^n)_{n \in \N}$, and let $\bX$ be the canonical rough path lift of $X$, as constructed in \eqref{eq: RIE rough path}. Let $(F,F'), (G,G') \in \crpX$ be controlled paths with respect to $X$, and suppose that $J_F \subseteq \liminf_{n \to \infty} \cP^n := \bigcup_{m \in \N} \bigcap_{n \geq m} \cP^n$, where $J_F := \{t \in (0,T] : F_{t-} \neq F_t\}$ denotes the set of jump times of $F$. Then, the limit
  \begin{equation*}
    \int_0^t F_u \dd G_u := \lim_{n \to \infty} \sum_{k=0}^{N_n-1} F_{t^n_k} G_{t^n_k \wedge t,t^n_{k+1} \wedge t},
  \end{equation*}
  exists, where the convergence holds uniformly for $t \in [0,T]$, and it coincides with the rough integral of $(F,F')$ against $(G,G')$, as defined in \eqref{eq: integration against controlled path}.
\end{theorem}

The proof of Theorem~\ref{thm: rough int against controlled path under RIE} follows the proof of \cite[Theorem~2.15]{Allan2023b} almost verbatim. The only difference is that, rather than using \cite[Proposition~2.14]{Allan2023b} to establish the uniform convergence of $F^n$ to $F$, we can instead use \cite[Proposition~B.1]{Allan2023c} (which does not require the sequence of partitions to be nested).

\medskip

A crucial observation for our pathwise analysis of log-optimal portfolios is that, if a path $X$ satisfies Property \textup{(RIE)}, then suitable controlled paths with respect to $X$ do as well. This is made precise in Theorem~\ref{thm: controlled path satisfies RIE} below. In particular, a corollary of this result is that if $X$ satisfies Property \textup{(RIE)}, and $Y$ is the solution to an RDE driven by the canonical rough path lift of $X$, then $Y$ itself satisfies Property \textup{(RIE)} with respect to the same sequence of partitions.

\begin{theorem}\label{thm: controlled path satisfies RIE}
  Suppose that $X \in D([0,T];\R^d)$ satisfies Property \textup{(RIE)} relative to some $p \in (2,3)$ and a sequence of partitions $\cP^n = \{0 = t^n_0 < t^n_1 < \cdots < t^n_{N_n} = T\}$, $n \in \N$. Let $(Y,Y') \in \cV^p_X$ be a controlled path such that $J_Y \subseteq \liminf_{n \to \infty} \cP^n$, where $J_Y := \{t \in (0,T] : Y_{t-} \neq Y_t\}$ denotes the set of jump times of $Y$. Then $Y$ satisfies Property \textup{(RIE)} relative to $p$ and $(\cP^n)_{n \in \N}$.
\end{theorem}

\begin{proof}
  For each $n \in \N$, let
  \begin{equation*}
    Y^n_t = Y_T \1_{\{T\}}(t) + \sum_{k=0}^{N_n - 1} Y_{t^n_k} \1_{[t^n_k,t^n_{k+1})}(t), \qquad t \in [0,T],
  \end{equation*}
  be the piecewise constant approximation of $Y$ along $\cP^n$. Since $J_Y \subseteq \liminf_{n \to \infty} \cP^n$, we have from \cite[Proposition~B.1]{Allan2023c} that $Y^n \to Y$ uniformly as $n \to \infty$, so that part~(i) of Property \textup{(RIE)} holds.

  By Lemma~\ref{lem: integration against controlled path}, we can define the rough integral of the controlled path $(Y,Y')$ against itself as\footnote{In writing $Y'_u \otimes Y'_u$, we technically mean the $4$-tensor whose $ijk\ell$ component is given by $[Y'_u \otimes Y'_u]^{ijk\ell} = (Y'_u)^{ij} (Y'_u)^{k\ell}$, and we interpret the ``multiplication'' $(Y'_u \otimes Y'_u) \X_{u,v}$ as the matrix whose $ik$ component is given by $[(Y'_u \otimes Y'_u) \X_{u,v}]^{ik} = \sum_j \sum_\ell (Y'_u)^{ij} (Y'_u)^{k\ell} \X_{u,v}^{j\ell}$.}
  \begin{equation*}
    \int_0^t Y_r \otimes \d Y_r := \lim_{|\cP| \to 0} \sum_{[u,v] \in \cP} Y_u \otimes Y_{u,v} + (Y'_u \otimes Y'_u) \X_{u,v}, \qquad t \in [0,T],
  \end{equation*}
  relative to the rough path $\bX = (X,\X)$, where the limit exists along any sequence of partitions $\cP$ of the interval $[0,t]$ with mesh size tending to zero. We have from Theorem~\ref{thm: rough int against controlled path under RIE} that
  \begin{equation}\label{eq: Riemann Stieltjes integral Y^n dY}
    \int_0^t Y^n_r \otimes \d Y_r = \sum_{k=0}^{N_n-1} Y_{t^n_k} \otimes Y_{t^n_k \wedge t,t^n_{k+1} \wedge t} \, \longrightarrow \, \int_0^t Y_r \otimes \d Y_r \qquad \text{as} \qquad n \, \longrightarrow \, \infty,
  \end{equation}
  where the convergence is uniform in $t \in [0,T]$, which gives part~(ii) of Property \textup{(RIE)}.

  As the piecewise constant approximation $X^n$ as defined in Property \textup{(RIE)} has finite $1$-variation, we also have that $\bX^n = (X,X^n,\X^n)$ is a c{\`a}dl{\`a}g rough path in the sense of \cite[Definition~2.1]{Allan2023b}, where
  \begin{equation}\label{eq: defn bbX^n}
    \X^n_{s,t} := \int_s^t X^n_u \otimes \d X_u - X^n_s \otimes X_{s,t},  \qquad (s,t) \in \Delta_T.
  \end{equation}
  We note that $(Y^n,Y')$ is a controlled path with respect to $X^n$. We can therefore consider the rough integral of $(Y^n,Y')$ against $(Y,Y')$ relative to the rough path $\bX^n$ in the sense of \cite[Proposition~2.4]{Allan2023b}, which is given by
  \begin{equation*}
    \int_0^t Y^n_r \otimes \d Y_r = \lim_{|\cP| \to 0} \sum_{[u,v] \in \cP} Y^n_u \otimes Y_{u,v} + (Y'_u \otimes Y'_u) \X^n_{u,v}.
  \end{equation*}
  For any refinement $\widetilde{\cP}$ of the partition $(\cP^n \cup \{t\}) \cap [0,t]$ and any $[u,v] \in \widetilde{\cP}$, there exists a $k$ such that $t^n_k \leq u < v \leq t^n_{k+1}$ which, recalling \eqref{eq: defn bbX^n}, implies that $\X^n_{u,v} = 0$. Thus,
  \begin{equation*}
    \int_0^t Y^n_r \otimes \d Y_r = \lim_{|\widetilde{\cP}| \to 0} \sum_{[u,v] \in \widetilde{\cP}} Y^n_u \otimes Y_{u,v} = \sum_{k=0}^{N_n-1} Y_{t^n_k} \otimes Y_{t^n_k \wedge t,t^n_{k+1} \wedge t},
  \end{equation*}
  so that the rough integral $\int_0^t Y^n_r \dd Y_r$ coincides with the Riemann--Stieltjes integral on the left-hand side of~\eqref{eq: Riemann Stieltjes integral Y^n dY}.

  Let us fix $0 \leq k < \ell \leq N_n$. By the estimate in \cite[Proposition~2.4]{Allan2023b}, we have that
  \begin{equation}\label{eq: first RIE bound on int Y^n dY}
    \begin{split}
    &\bigg|\int_{t^n_k}^{t^n_\ell} Y^n_r \otimes \d Y_r - Y_{t^n_k} \otimes Y_{t^n_k,t^n_\ell} - (Y'_{t^n_k} \otimes Y'_{t^n_k}) \X^n_{t^n_k,t^n_\ell}\bigg|\\
    &\quad\lesssim \|Y'\|_\infty (\|Y'\|_{p,[t^n_k,t^n_\ell]}^p + \|X^n\|_{p,[t^n_k,t^n_\ell]}^p)^{\frac{2}{p}} \|X\|_{p,[t^n_k,t^n_\ell]} + \|Y^n\|_{p,[t^n_k,t^n_\ell]} \|R^Y\|_{\frac{p}{2},[t^n_k,t^n_\ell]}\\
    &\quad\hspace{30pt} + \|R^{Y^n}\|_{\frac{p}{2},[t^n_k,t^n_\ell]} \|Y'\|_\infty \|X\|_{p,[t^n_k,t^n_\ell]} + \|Y' \otimes Y'\|_{p,[t^n_k,t^n_\ell]} \|\X^n\|_{\frac{p}{2},[t^n_k,t^n_\ell]}.
    \end{split}
  \end{equation}
  It is clear that the functions given by $w_1(s,t) := \|Y'\|_{p,[s,t]}^p$, $w_2(s,t) := \|X\|_{p,[s,t]}^p$, and $w_3(s,t) := \|R^Y\|_{\frac{p}{2},[s,t]}^{\frac{p}{2}}$ for $(s,t) \in \Delta_T$ are all controls. Since $t^n_k, t^n_\ell \in \cP^n$, we have that
  \begin{align*}
    \|X^n\|_{p,[t^n_k,t^n_\ell]} &\leq \|X\|_{p,[t^n_k,t^n_\ell]} =  w_2(t^n_k,t^n_\ell)^{\frac{1}{p}},\\
    \text{and} \qquad \|Y^n\|_{p,[t^n_k,t^n_\ell]} &\leq \|Y\|_{p,[t^n_k,t^n_\ell]} \leq \|Y\|_p.
  \end{align*}
  Let $w$ denote the control with respect to which~\eqref{eq: RIE inequality} holds for $X$. Note that $X$ also satisfies Property \textup{(RIE)} over the subinterval $[t^n_k,t^n_\ell]$, with respect to $p$, the sequence of partitions $(\cP^m \cap [t^n_k,t^n_\ell])_{m \geq n}$, and the same control $w$. It then follows from \cite[Lemma~2.12]{Allan2023b} that
  \begin{equation*}
    \sup_{m \geq n} \|\X^m\|_{\frac{p}{2},[t^n_k,t^n_\ell]} \lesssim w(t^n_k,t^n_\ell)^{\frac{2}{p}}.
  \end{equation*}
  We also infer from the proof of \cite[Theorem~2.15]{Allan2023b} that
  \begin{align*}
    \sup_{m \geq n} \|R^{Y^m}\|_{\frac{p}{2},[t^n_k,t^n_\ell]}^{\frac{p}{2}} &\lesssim \|Y'\|_{p,[t^n_k,t^n_\ell]}^p + \|X\|_{p,[t^n_k,t^n_\ell]}^p + \|R^Y\|_{\frac{p}{2},[t^n_k,t^n_\ell]}^{\frac{p}{2}}\\
    &= w_1(t^n_k,t^n_\ell) + w_2(t^n_k,t^n_\ell) + w_3(t^n_k,t^n_\ell).
  \end{align*}
  The estimate in \eqref{eq: first RIE bound on int Y^n dY} then implies that
  \begin{equation*}
    \begin{split}
    &\bigg|\int_{t^n_k}^{t^n_\ell} Y^n_u \otimes \d Y_u - Y_{t^n_k} \otimes Y_{t^n_k,t^n_\ell} - (Y'_{t^n_k} \otimes Y'_{t^n_k}) \X^n_{t^n_k,t^n_\ell}\bigg|^{\frac{p}{2}} \\
    &\quad\lesssim \|Y'\|_\infty^{\frac{p}{2}} (w_1(t^n_k,t^n_\ell) + w_2(t^n_k,t^n_\ell)) \|X\|_p^{\frac{p}{2}} + \|Y\|_p^{\frac{p}{2}} w_3(t^n_k,t^n_\ell) \\
    &\hspace{30pt} + (w_1(t^n_k,t^n_\ell) + w_2(t^n_k,t^n_\ell) + w_3(t^n_k,t^n_\ell)) \|Y'\|_\infty^{\frac{p}{2}} \|X\|_p^{\frac{p}{2}} + \|Y' \otimes Y'\|_p^{\frac{p}{2}} w(t^n_k,t^n_\ell).
    \end{split}
  \end{equation*}
  Since we can also bound
  \begin{equation*}
    |(Y'_{t^n_k} \otimes Y'_{t^n_k}) \X^n_{t^n_k,t^n_\ell}|^{\frac{p}{2}} \leq \|Y' \otimes Y'\|_\infty^{\frac{p}{2}} \|\X^n\|_{\frac{p}{2},[t^n_k,t^n_\ell]}^{\frac{p}{2}} \lesssim \|Y' \otimes Y'\|_\infty^{\frac{p}{2}} w(t^n_k,t^n_\ell),
  \end{equation*}
  it is then clear how to choose a control $w_4$ such that
  \begin{equation*}
    \bigg|\int_{t^n_k}^{t^n_\ell} Y^n_u \otimes \dd Y_u - Y_{t^n_k} \otimes Y_{t^n_k,t^n_\ell}\bigg|^{\frac{p}{2}} \leq w_4(t^n_k,t^n_\ell).
  \end{equation*}
  Since $w_4$ does not depend on the choices of $n \in \N$ or $0 \leq k < l \leq N_n$, we have established part~(iii) of Property \textup{(RIE)}.
\end{proof}

\subsection{Consistency of rough and stochastic integration}

In this subsection we briefly discuss the relation between the deterministic theory of rough integration, as developed in Sections~\ref{subsec: rough path theory} and \ref{subsec: pathwise ito integration}, and stochastic integration. As before, we fix a probability space $(\Omega, \mathcal{F}, \P)$ with a filtration $(\mathcal{F}_t)_{t \in [0,T]}$ satisfying the usual conditions. As shown in \cite[Section~3]{Allan2023c}, the sample paths of various stochastic processes, such as Brownian motion, It\^o processes and L\'evy processes, almost surely satisfy Property \textup{(RIE)} relative to $p \in (2,3)$ and suitable sequences of partitions $(\cP^n)_{n \in \N}$. In the current work, we fundamentally rely on Brownian motion, and thus recall the corresponding result in the following remark, which combines the content of \cite[Lemma~3.1 and Proposition~3.2]{Allan2023c}.

\begin{remark}\label{remark: BM satisfies RIE}
  Let $\bar{W} = (\bar{W}_t)_{t \in [0,T]}$ be $d$-dimensional Brownian motion, $p \in (2,3)$ and $\cP^n = \{0 = t_0^n < t_1^n < \cdots < t_{N_n}^n = T\}$, $n \in \N$, be a sequence of equidistant partitions of the interval $[0,T]$, so that, for each $n \in \N$, there exists some $\pi_n > 0$ such that $t^n_{i+1} - t^n_i = \pi_n$ for each $0 \leq i < N_n$. If $\pi_n^{2 - \frac{4}{p}} \log(n) \to 0$ as $n \to \infty$, then, for almost every $\omega \in \Omega$, the sample path $W(\omega)$ satisfies Property \textup{(RIE)} relative to $p$ and $(\cP^n)_{n \in \N}$.

  Moreover, the canonical rough path $\bar{\bW}(\omega) = (\bar{W}(\omega),\bar{\mathbb{W}}(\omega))$ defined via Property \textup{(RIE)} corresponds almost surely to the random rough path defined via It{\^o} integration, namely, where
  \begin{equation*}
    \bar{\mathbb{W}}_{s,t} := \int_s^t \bar{W}_{s,r} \otimes \d \bar{W}_r = \int_s^t \bar{W}_r \otimes \d \bar{W}_r - \bar{W}_s \otimes \bar{W}_{s,t}, \qquad (s,t) \in \Delta_T.
  \end{equation*}
\end{remark}

Property \textup{(RIE)} also ensures that a random rough integral against a semimartingale coincides almost surely with the associated stochastic It\^o integral.

\begin{lemma}\label{lemma: rough integral is equal to stochastic integral}
Let $X = (X_t)_{t \in [0,T]}$ be a $d$-dimensional c{\`a}dl{\`a}g semimartingale, and let $(Y,Y')$ be a c{\`a}dl{\`a}g stochastic process adapted to $(\mathcal{F}_t)_{t \in [0,T]}$. Let $p \in (2,3)$. By part~(i) of \cite[Proposition~3.11]{Allan2023c}, there exists an adapted sequence of partitions $\cP^n = \{\tau^n_k\}$, $n \in \N$, (so that each $\tau^n_k \in \cP^n$ is a stopping time), such that, for almost every $\omega \in \Omega$, the path $X(\omega)$ satisfies Property \textup{(RIE)} relative to $p$ and $(\cP^n(\omega))_{n \in \N}$. Suppose that, for almost every $\omega \in \Omega$, $(Y(\omega),Y'(\omega))$ is a controlled path in $\cV^p_{X(\omega)}$ with $J_{Y(\omega)} \subseteq \liminf_{n \to \infty} \cP^n(\omega)$, where $J_{Y(\omega)}$ denotes the set of jump times of $Y(\omega)$. Then the rough and It{\^o} integrals of $Y$ against $X$ coincide $\P$-almost surely, that is,
  \begin{equation*}
    \int_0^t Y_s(\omega) \dd \bX_s(\omega) = \bigg( \int_0^t Y_{s-} \dd X_s \bigg)(\omega) \qquad \text{for all } \ t \in [0,T],
  \end{equation*}
  holds for almost every $\omega \in \Omega$, where $\bX(\omega)$ is the canonical rough path lift of $X(\omega)$ as defined via Property \textup{(RIE)}.
\end{lemma}

\begin{proof}
  By, e.g., \cite[Ch.~II, Theorem~21]{Protter2005}, we have that
  \begin{equation*}
    \sum_{k=0}^{N_n-1} Y_{\tau_k^n} X_{\tau_k^n \wedge t, \tau_{k+1}^n \wedge t} \, \longrightarrow \, \int_0^t Y_{s-} \dd X_s \qquad \text{as} \quad n \to \infty,
  \end{equation*}
  where the convergence holds uniformly (in $t \in [0,T]$) in probability. By taking a subsequence if necessary, we can then assume that the (uniform) convergence holds almost surely. On the other hand, by Theorem~\ref{thm: rough int against controlled path under RIE}, we know that, for almost every $\omega \in \Omega$,
  \begin{equation*}
    \sum_{k=0}^{N_n-1} Y_{\tau_k^n(\omega)}(\omega) X_{\tau_k^n(\omega) \wedge t, \tau_{k+1}^n(\omega) \wedge t}(\omega) \, \longrightarrow \, \int_0^t Y_s(\omega) \dd \bX_s(\omega) \qquad \text{as} \quad n \to \infty,
  \end{equation*}
  uniformly for $t \in [0,T]$. The result thus follows by the uniqueness of limits.
\end{proof}

\section{Local volatility models: pathwise analysis of log-optimal portfolios}\label{sec: local vol models}

In this section we shall study log-optimal portfolios for the investment-consumption problem, acting on deterministic price paths generated by local volatility models, defined in a pathwise manner. To this end, we make the following assumption throughout this section.

\begin{assumption}\label{ass: partitions exhausting the jumps of K}
  Let $p \in (2,3)$ and let $\cP^n = \{0 = t_0^n < t_1^n < \cdots < t_{N_n}^n = T\}$, $n \in \N$, be a sequence of equidistant partitions of the interval $[0,T]$, such that,
  \begin{itemize}
    \item[$\vcenter{\hbox{\scriptsize$\bullet$}}$] for each $n \in \N$, there exists some $\pi_n > 0$ such that $t^n_{i+1} - t^n_i = \pi_n$ for each $0 \leq i < N_n$,
    \item[$\vcenter{\hbox{\scriptsize$\bullet$}}$] $\pi_n^{2 - \frac{4}{p}} \log(n) \to 0$ as $n \to \infty$,
    \item[$\vcenter{\hbox{\scriptsize$\bullet$}}$] $J_K \subseteq \liminf_{n \to \infty} \cP^n$ with $J_K := \{t \in (0,T] : K_{t-} \neq K_t\}$,
  \end{itemize}
  where the consumption clock $K \colon [0,T] \to \R$ is fixed as in Section~\ref{subsec: setting}. Moreover, the deterministic path $W \colon [0,T] \to \R^d$ satisfies Property \textup{(RIE)} relative to $p$ and $(\mathcal{P}^n)_{n \in \N}$.
\end{assumption}

We suppose that the discounted price path $(S_t)_{t \in [0,T]}$ satisfies the rough differential equation
\begin{equation}\label{eq: pathwise local vola model}
  S_t = s_0 + \int_0^t b(s,S_s) \dd s + \int_0^t \sigma(s, S_s) \dd \bW_s, \qquad t \in [0,T],
\end{equation}
where $s_0 \in \R^m$, $b \in C^3_b(\R^{m+1};\R^m)$, $\sigma \in C^3_b(\R^{m+1};\cL(\R^d;\R^m))$, and $\bW = (W,\mathbb{W})$ is the canonical rough path lift of $W$ as defined in \eqref{eq: RIE rough path}.

\begin{remark}\label{remark: pathwise local vola model}
  If $W$ is a realization of a Brownian motion, the dynamics of the RDE \eqref{eq: pathwise local vola model} can be seen as a fixed realization of a local volatility model for a financial market.

  Indeed, let us assume that $\bar{W} = (\bar{W}_t)_{t \in [0,T]}$ is a $d$-dimensional Brownian motion on a probability space $(\Omega,\mathcal{F},\P)$ with respect to an underlying filtration $(\cF_t)_{t \in [0,T]}$. It is well-known that the stochastic differential equation (SDE)
  \begin{equation}\label{eq: Ito diffusion}
    \bar{S}_t = s_0 + \int_0^t b(s,\bar{S}_s) \dd s + \int_0^t \sigma(s,\bar{S}_s) \dd \bar{W}_s, \qquad t \in [0,T],
  \end{equation}
  has a unique strong solution, where $\int_0^t \sigma(s,\bar{S}_s) \dd \bar{W}_s$ denotes the stochastic It{\^o} integral; see, e.g., \cite[Ch.~V, Theorem~6]{Protter2005}. Note that the It{\^o} diffusion $(\bar{S}_t)_{t \in [0,T]}$ represents many standard models for financial markets, including local volatility models.

  Recall that for almost every $\omega \in \Omega$, the sample path $\bar{W}(\omega)$ of a Brownian motion satisfies Property \textup{(RIE)} relative to $p$ and $(\cP^n)_{n \in \N}$; see Remark~\ref{remark: BM satisfies RIE}. Hence, for almost every $\omega \in \Omega$, the solution $(\bar{S}_t(\omega))_{t \in [0,T]}$ of the SDE \eqref{eq: Ito diffusion} driven by $\bar{W}$, and the solution $(S_t)_{t \in [0,T]}$ of the rough differential equation driven by the rough path $\bW = (W,\mathbb{W}) := \bar{\bW}(\omega) = (\bar{W}(\omega), \bar{\mathbb{W}}(\omega))$ coincide; see \cite[Lemma~3.1]{Allan2023c}. In other words, $(S_t)_{t \in [0,T]}$ can be understood as a fixed realization of the probabilistic model $(\bar{S}_t)_{t \in [0,T]}$.
\end{remark}

In the present setting, it will be convenient to equivalently reformulate the RDE \eqref{eq: pathwise local vola model} into the RDE
\begin{equation}\label{eq: dynamics of price process via RDE}
  S_t = s_0 + \int_0^t (b, \sigma) (s,S_s) \dd (\cdot,\bW)_s, \qquad t \in [0,T],
\end{equation}
where $(\cdot,\bW)$ denotes the time-extended rough path of $\bW$, i.e., the path-level of $(\cdot,\bW)$ is given by $(t,W_t)_{t \in [0,T]}$, and the missing integrals $\int W^j \dd t$, $\int t \dd W^j$, $j = 1, \ldots, d$, to define a rough path are canonically defined as Riemann--Stieltjes integrals. By classical rough path theory (e.g., \cite[Theorem~2.5]{Allan2020}), for any $b \in C^3_b(\R^{m+1};\R^m)$, $\sigma \in C^3_b(\R^{m+1};\cL(\R^d;\R^m))$, there exists a unique solution $(S,S') \in \cV^p_{(\cdot,W)}$ to the RDE \eqref{eq: dynamics of price process via RDE}, where $S' = (b,\sigma) (\cdot, S)$. Moreover, $(S_t)_{t \in [0,T]}$ satisfies the RDE \eqref{eq: dynamics of price process via RDE} if and only if it satisfies the RDE \eqref{eq: pathwise local vola model}. For later reference, we also remark that $(\cdot,W)$ satisfies Property \textup{(RIE)} relative to $p$ and $(\cP^n)_{n \in \N}$ by \cite[Proposition~2.12]{Allan2023c}.

\subsection{Pathwise construction of log-optimal portfolios}

As a first step to a pathwise analysis of optimal portfolios, we prove a pathwise construction of the log-optimal portfolio, supposing that the underlying price dynamics of the financial market are given by a local volatility model. Recall that in the probabilistic setting the log-optimal portfolio is well-known and was presented in Theorem~\ref{thm: Goll-Kallsen's main theorem Ito processes}, which is due to \cite{Goll2000}.

\smallskip

In the following, we write $\cA \subset C^3_b([0,T] \times \R^m;\cL(\R^d;\R^m))$ for the class of functions $\sigma \in C^3_b$ such that $\sigma(t,x) \sigma(t,x)^\top \in \mathrm{GL}(\R^{m \times m})$ for all $(t,x)$, such that the determinant is uniformly bounded away from zero. Here, $\mathrm{GL}(\R^{m \times m})$ denotes the general linear group of degree $m$. For a $m \times m$-matrix we write $\det(\cdot)$ for its determinant, and $(\cdot)^\top$ denotes matrix transposition. Given a path $W$, the time-augmented path is denoted by $(\cdot,W) = (t,W)_{t \in [0,T]}$.

\begin{lemma}\label{lem: pathwise optimal portfolio}
  For $(b,\sigma) \in C^3_b([0,T] \times \R^m; \R^m) \times \cA$, let
  \begin{equation*}
    H_t :=  H^{(b,\sigma)}_t := c(t,S_t)^{-1} b(t,S_t) \qquad \text{with} \qquad c(t,S_t) := c^{(b,\sigma)}_t := \sigma(t,S_t) \sigma(t,S_t)^\top,
  \end{equation*}
  for $t \in [0,T]$, and set $(\phi,\kappa) := (\phi^{(b,\sigma)},\kappa^{(b,\sigma)}) := (\phi^{(b,\sigma),0}, \dots, \phi^{(b,\sigma),m}, \kappa^{(b,\sigma)})$, with
  \begin{align*}
    &\kappa_t := \kappa_t^{(b,\sigma)} := \frac{1}{K_T} \cE \bigg( \sum_{i=1}^m \int_0^\cdot H_s^i \dd S_s^i \bigg)_{\hspace{-2pt}t}, \qquad V_t := V_t^{(b,\sigma)} := \kappa_t(K_T - K_t),\\
    &\phi_t^i := \phi_t^{(b,\sigma),i} := H_t^i V_t, \qquad i = 1, \ldots, m, \qquad  \phi_t^0 := \phi_t^{(b,\sigma),0} := \sum_{i=1}^m \int_0^t \phi_s^i \dd S_s^i - \phi_t^i S_t^i,
  \end{align*}
  for $t \in [0,T]$, where $\int_0^t \phi_s^i \dd S_s^i$ is the rough integral, and $\cE$ is the rough exponential as defined in Lemma~\ref{lem: dynamics of rough exponential}. Then $\phi, \kappa$ and $V$ are all well-defined and are controlled paths with respect to $W$ and, in particular, with respect to $(\cdot,W)$.
\end{lemma}

\begin{proof}
  It is a well known result in rough path theory (e.g., \cite[Lemma~3.5]{Friz2018}) that the composition of a controlled path with a regular function remains a controlled path. More precisely, since $S \in \cV^p_W$, we have that $\sigma(\cdot,S)$ and $b(\cdot,S)$ are controlled paths in $\cV^p_W$. We recall that $\det(\sigma(\cdot,S) \sigma(\cdot,S)^\top)$ is bounded away from zero. We then obtain (componentwise) that $H$ is a controlled path in $\cV^p_W$, since the sum and the product of two (real-valued) controlled paths is again a controlled path by Lemma~\ref{lem: product of controlled paths}, as well as the inverse of a controlled path which is bounded away from zero (as a composition with the smooth function $x \mapsto \frac{1}{x}$).

  By Lemma~\ref{lem: integration against controlled path}, since $H^i$ and $S^i$ are both controlled paths, the rough integral $\int_0^\cdot H_s^i \dd S_s^i$ is well-defined and is itself a controlled path for each $i = 1, \ldots, m$. By Lemma~\ref{lem: consistency of rough integrals}, the path $Z_t := \sum_{i=1}^m \int_0^t H_s^i \dd S_s^i$ may then be considered as a rough path, and Lemma~\ref{lem: dynamics of rough exponential} then implies that the rough exponential $\cE(Z)$, and hence also $\kappa = \frac{1}{K_T} \cE(Z)$, are controlled paths.

  Since the consumption clock $K$ is a c{\`a}dl{\`a}g (deterministic) and increasing function (and thus of finite $1$-variation), by Lemma~\ref{lem: product of controlled paths}, the wealth process $V$ is a controlled path in $\cV^p_W$, as the product of two controlled paths.

  By similar arguments, we see that $\phi^i \in \cV^p_W$, $i = 0, 1, \ldots, m$ are also all controlled paths with respect to $W$, and hence also with respect to $(\cdot,W)$.
\end{proof}

The portfolio constructed in Lemma~\ref{lem: pathwise optimal portfolio} in a pathwise manner agrees, indeed, with the log-optimal portfolio for the investment-consumption problem as considered in Section~\ref{sec: portfolio in probabilistic setting}, if the underlying frictionless financial market is generated by a local volatility model, such as the stochastic differential equation \eqref{eq: Ito diffusion}. Hence, in the following we shall call the portfolio $(\phi, \kappa) = (\phi^{(b,\sigma)}, \kappa^{(b,\sigma)})$ from Lemma~\ref{lem: pathwise optimal portfolio} a \emph{pathwise log-optimal portfolio}.

\begin{lemma}\label{lem: pathwise portfolio is optimal}
  Suppose that the discounted price process $(\bar{S}_t)_{t \in [0,T]}$ is modelled by the SDE \eqref{eq: Ito diffusion} driven by a Brownian motion $\bar{W}$ on a probability space $(\Omega,\mathcal{F},\P)$ with respect to an underlying filtration $(\cF_t)_{t \in [0,T]}$. Then the log-optimal portfolio $(\bar{\phi}, \bar{\kappa})$, as provided in Theorem~\ref{thm: Goll-Kallsen's main theorem Ito processes}, and the pathwise log-optimal portfolio $(\phi, \kappa)$, as provided in Lemma~\ref{lem: pathwise optimal portfolio}, coincide $\P$-almost surely, where $(\phi, \kappa)$ is constructed given the realization $\bW := \bar{\bW}(\omega)$ of the It\^o rough path lift of the Brownian motion $\bar{W}$, for almost every $\omega \in \Omega$.
\end{lemma}

\begin{remark}
  One may note that in Theorem~\ref{thm: Goll-Kallsen's main theorem Ito processes} we had $\bar{\phi}^i_t = \bar{H}^i_t \bar{V}_{t-}$, but in Lemma~\ref{lem: pathwise optimal portfolio} we have $\phi^i_t = H^i_t V_t$, without taking the left-limit. This is only to be consistent with standard rough analysis in which controlled paths are assumed to be c\`adl\`ag, and it makes no difference to the value of the rough integral $\int_0^t \phi^i_s \dd S^i_s$, as explained in \cite[Remark~3.3]{Allan2023b}.
\end{remark}

\begin{proof}[Proof of Lemma~\ref{lem: pathwise portfolio is optimal}]
  In this proof we consider $S, \phi, \kappa$, etc., as random controlled paths, in the sense that $S$ is a stochastic process such that $S(\omega)$ is a controlled path for almost every $\omega \in \Omega$. In particular, $S$ is defined pathwise as the solution to the RDE \eqref{eq: pathwise local vola model}, and $(\phi, \kappa)$ is defined pathwise via Lemma~\ref{lem: pathwise optimal portfolio}, given a realization of the Brownian motion $\bar{W}$; see Remark~\ref{remark: pathwise local vola model}.

  By the associativity property of rough integrals (see Proposition~\ref{prop: associativity of rough integration}), and the consistency of rough and stochastic integrals (as established in Lemma~\ref{lemma: rough integral is equal to stochastic integral}), we have that
  \begin{align*}
    \int_0^t H_s^i \dd S_s^i &= \int_0^t H_s^i b^i(s,S_s) \dd s + \int_0^t H_s^i \sigma^{i \cdot}(s,S_s) \dd \bW_s\\
    &= \int_0^t \bar{H}_s^i b^i(s,\bar{S}_s) \dd s + \int_0^t \bar{H}_s^i \sigma^{i \cdot}(s,\bar{S}_s) \dd \bar{W}_s = \int_0^t \bar{H}_s^i \dd \bar{S}_s^i
  \end{align*}
  almost surely for each $i = 1, \ldots, m$, where $\int_0^t \bar{H}_s^i \dd \bar{S}_s^i$ is the stochastic It\^o integral of $\bar{H}_t := c(t,\bar{S}_t)^{-1} b(t,\bar{S}_t)$ against the price process $\bar{S}$ in \eqref{eq: Ito diffusion}.

  Let $Z := \sum_{i=1}^m \int_0^\cdot H_s^i \dd S_s^i$ and $\bar{Z} = \sum_{i=1}^m \int_0^\cdot \bar{H}^i_s \dd \bar{S}^i_s = \int_0^\cdot \bar{H}_s^\top \dd \bar{S}_s$. By Lemma~\ref{lem: consistency of rough integrals}, $Z$ admits a canonical rough path lift $\bZ \in \cD^p$. Further, since $(Z,Z') \in \cV^p_W$ and $W$ satisfies Property \textup{(RIE)} relative to $p$ and $(\cP^n)_{n \in \N}$, $Z$ also satisfies Property \textup{(RIE)} by Theorem~\ref{thm: controlled path satisfies RIE}. By a slight extension (to allow non-nested partitions) of \cite[Proposition~2.18]{Allan2023b}, this implies that the rough path bracket $[\bZ]$ coincides with the quadratic variation $[Z]$ of $Z$ along $(\cP^n)_{n \in \N}$ in the sense of F{\"o}llmer, that is
  \begin{equation*}
    [\bZ]_t = [Z]_t =  \lim_{n \to \infty} \sum_{k=0}^{N_n-1} (Z_{t^n_k \wedge t,t^n_{k+1} \wedge t})^2, \qquad t \in [0,T].
  \end{equation*}
  On the other hand, we have that
  \begin{equation*}
    [\bar{Z}]_t = \lim_{n \to \infty} \sum_{k=0}^{N_n-1} (\bar{Z}_{t^n_k \wedge t,t^n_{k+1} \wedge t})^2, \qquad t \in [0,T],
  \end{equation*}
  where the convergence holds uniformly (in $t \in [0,T]$) in probability. By taking a subsequence if necessary, we can then assume that the uniform convergence holds almost surely, and it follows that $[\bZ] = [\bar{Z}]$ almost surely. In particular, we have that
  \begin{equation*}
    \cE(Z)_t = \exp \Big(Z_t - \frac{1}{2}[\bZ]_t\Big) = \exp \Big(\bar{Z}_t - \frac{1}{2}[\bar{Z}]_t\Big) = \bar{\cE}(\bar{Z})_t
  \end{equation*}
  almost surely, where $\bar{\cE}(\bar{Z})$ denotes the stochastic exponential of $\bar Z$. Moreover, it holds that $\sum_{i=1}^m \int_0^\cdot \phi_s^i \dd S_s^i = \int_0^\cdot \bar{\phi}_s^\top \dd \bar{S}_s$ almost surely. Thus, the log-optimal portfolio $(\bar{\phi}, \bar{\kappa})$, as provided in Theorem~\ref{thm: Goll-Kallsen's main theorem Ito processes}, and the pathwise log-optimal portfolio $(\phi, \kappa)$, as provided in Lemma~\ref{lem: pathwise optimal portfolio}, coincide almost surely.
\end{proof}

\begin{remark}
  We take $\bar{W}$ to be a Brownian motion to ensure that the pathwise log-optimal portfolio $(\phi,\kappa)$, as constructed in Lemma~\ref{lem: pathwise optimal portfolio}, is, indeed, a log-optimal portfolio for the investment-consumption problem in the setting of local volatility models. However, we emphasize that the construction of the pathwise portfolio $(\phi,\kappa)$, as well as its pathwise analysis developed in Sections~\ref{subsec: stability in local vola model} and \ref{subsec: discretization in local vola model} below, work for any path $W$ satisfying Assumption~\ref{ass: partitions exhausting the jumps of K}.
\end{remark}

\subsection{Stability of pathwise log-optimal portfolios with respect to drift and volatility}\label{subsec: stability in local vola model}

Having at hand a pathwise construction of log-optimal portfolios, we are in a position to study its pathwise stability properties. In this subsection, we analyze the stability of the log-optimal portfolio and the associated capital process with respect to the model parameters, $b$ and $\sigma$.

In particular, the following result shows that the pathwise log-optimal portfolio $(\phi, \kappa) = (\phi^{(b, \sigma)}, \kappa^{(b, \sigma)})$ and its associated capital processes $V = V^{(b,\sigma)}$ are locally Lipschitz continuous with respect to these parameters.

\begin{theorem}\label{thm: stability of optimal portfolios wrt model uncertainty}
  For $(b, \sigma), (\tb, \tsigma) \in C^3_b([0,T] \times \R^m;\R^m) \times \cA$, let $(\phi^{(b,\sigma)}, \kappa^{(b,\sigma)})$ and $(\phi^{(\tb,\tsigma)}, \kappa^{(\tb,\tsigma)})$ be the corresponding pathwise log-optimal portfolios, as constructed in Lemma~\ref{lem: pathwise optimal portfolio}. Let $M > 0$ be an upper bound for $\|b\|_{C^3_b}, \|\tb\|_{C^3_b}, \|\sigma\|_{C^3_b}, \|\tsigma\|_{C^3_b}$, as well as for
  \begin{equation*}
    \sup_{(t,x)} |\det(\sigma(t,x)\sigma(t,x)^\top)|^{-1}, \ \sup_{(t,x)} |\det(\tsigma(t,x)\tsigma(t,x)^\top)|^{-1} \ \text{ and } \ \|(\cdot, \bW)\|_p.
  \end{equation*}
  We then have that
  \begin{equation}\label{eq: stability est for optimal portfolio}
    \|(\phi^{(b,\sigma)}, \kappa^{(b,\sigma)}); (\phi^{(\tb,\tsigma)}, \kappa^{(\tb,\tsigma)})\|_{\crpW} \lesssim \|b - \tb\|_{C^2_b} + \|\sigma - \tsigma\|_{C^2_b}
  \end{equation}
  and
  \begin{equation*}
    \|V^{(b,\sigma)}; V^{(\tb,\tsigma)}\|_{\crpW} \lesssim \|b - \tb\|_{C^2_b} + \|\sigma -\tsigma\|_{C^2_b},
  \end{equation*}
  and in particular that
  \begin{equation*}
    \|(\phi^{(b,\sigma)}, \kappa^{(b,\sigma)}) - (\phi^{(\tb,\tsigma)}, \kappa^{(\tb,\tsigma)})\|_\infty \lesssim \|b - \tb\|_{C^2_b} + \|\sigma - \tsigma\|_{C^2_b}
  \end{equation*}
  and
  \begin{equation*}
    \|V^{(b,\sigma)} - V^{(\tb,\tsigma)}\|_\infty \lesssim \|b - \tb\|_{C^2_b} + \|\sigma - \tsigma\|_{C^2_b},
  \end{equation*}
  where the implicit multiplicative constants depend only on $p$, $m$, $d$, $M$, $s_0$ and the consumption clock $K$.
\end{theorem}

\begin{proof}
\emph{Step 1.}
Using the classical result from rough path theory (e.g., \cite[Theorem~2.5]{Allan2020}), for any $(b,\sigma) \in C^3_b([0,T] \times \R^m;\R^m) \times \cA$, we recall that there exists a unique controlled path $(S^{(b,\sigma)},(S^{(b,\sigma)})') \in \cV^p_{(\cdot,W)}$ which satisfies the rough differential equation
\begin{equation*}
S_t^{(b,\sigma)} = s_0 + \int_0^t (b, \sigma) (s,S^{(b,\sigma)}_s) \dd (\cdot, \bW)_s, \qquad t \in [0,T],
\end{equation*}
with $(S^{(b,\sigma)})' = (b,\sigma) (\cdot, S^{(b,\sigma)})$. By part~(i) of \cite[Corollary~3.3]{Kwossek2024}, we have that
\begin{equation*}
\|S^{(b,\sigma)}\|_{\cV^p_{(\cdot,W)}} + \|S^{(\tb,\tsigma)}\|_{\cV^p_{(\cdot,W)}} \leq C,
\end{equation*}
where the constant $C > 0$ depends only on $p$, $M$ and $s_0$. Further, by the continuity of the solution map with respect to the controlled path norm (e.g., part~(ii) of \cite[Corollary~3.3]{Kwossek2024}), for any $(b,\sigma), (\tilde b, \tilde \sigma) \in C^3_b([0,T] \times \R^m;\R^m) \times \cA$, it holds that
\begin{equation}\label{eq: Lipschitz estimate for S}
\|S^{(b,\sigma)}; S^{(\tb,\tsigma)}\|_{\crpW} \lesssim \|b - \tb\|_{C^2_b} + \|\sigma - \tsigma\|_{C^2_b},
\end{equation}
where the implicit multiplicative constant depends only on $p$, $M$ and $s_0$.

\emph{Step 2.}
Another well-known result in rough path theory is that compositions of controlled paths with regular functions remain controlled paths, and that such compositions are locally Lipschitz continuous (see, e.g., \cite[Lemma~3.5]{Friz2018}). More precisely, for $(b, \sigma)$ and $(\tb, \tsigma)$ in $C^3_b([0,T] \times \R^m; \R^m) \times \cA$, we have that $b(\cdot, S^{(b,\sigma)})$, $\sigma(\cdot, S^{(b,\sigma)})$, $\tb(\cdot, S^{(\tb,\tsigma)})$ and $\tsigma(\cdot, S^{(\tb,\tsigma)})$ are controlled paths in $\crpW$, and that
\begin{equation}\label{eq: estimate for b and sigma}
\|b(\cdot,S^{(b,\sigma)})\|_{\crpW} + \|\sigma(\cdot,S^{(b,\sigma)})\|_{\crpW} \leq C,
\end{equation}
where $C > 0$ depends only on $p$, $M$ and $s_0$, and of course the same bound holds for $\tb(\cdot,S^{(\tb,\tsigma)})$ and $\tsigma(\cdot,S^{(\tb,\tsigma)})$. It also holds that
\begin{equation*}
\|b(\cdot, S^{(b,\sigma)}); \tb(\cdot, S^{(\tb,\tsigma)})\|_{\crpW} \lesssim \|b - \tb\|_{C^2_b} + \|S^{(b,\sigma)}; S^{(\tb,\tsigma)}\|_{\crpW}
\end{equation*}
and
\begin{equation*}
\|\sigma(\cdot, S^{(b,\sigma)}); \tsigma(\cdot, S^{(\tb,\tsigma)})\|_{\crpW} \lesssim \|\sigma - \tsigma\|_{C^2_b} + \| S^{(b,\sigma)}; S^{(\tb,\tsigma)}\|_{\crpW},
\end{equation*}
where the implicit multiplicative constant depends only on $p$, $M$ and $s_0$. Combining the estimates above with \eqref{eq: Lipschitz estimate for S}, we obtain
\begin{equation}\label{eq: Lipschitz estimate for b and sigma}
\|b(\cdot, S^{(b,\sigma)}); \tb(\cdot, S^{(\tb,\tsigma)})\|_{\crpW} + \|\sigma(\cdot, S^{(b,\sigma)}); \tsigma(\cdot, S^{(\tb,\tsigma)})\|_{\crpW} \lesssim \|b - \tb\|_{C^2_b} + \|\sigma - \tsigma\|_{C^2_b},
\end{equation}
where the implicit multiplicative constant depends only on $p$, $M$ and $s_0$.

\emph{Step 3.}
Let $c^{(b,\sigma)} := \sigma(\cdot,S^{(b,\sigma)}) \sigma(\cdot,S^{(b,\sigma)})^\top$ and $c^{(\tb,\tsigma)} := \tsigma(\cdot,S^{(\tb,\tsigma)}) \tsigma(\cdot,S^{(\tb,\tsigma)})^\top$. We recall that the sum and the product of controlled paths are again controlled paths (see Lemma~\ref{lem: product of controlled paths}), as well as the reciprocal of a controlled path that is bounded away from zero (as a composition with the regular function $x \mapsto \frac{1}{x}$). Since $\det(\sigma(\cdot,S^{(b,\sigma)}) \sigma(\cdot,S^{(b,\sigma)})^\top)$ and $\det(\tilde \sigma(\cdot,S^{(\tilde b,\tilde \sigma)}) \tilde \sigma(\cdot,S^{(\tilde b,\tilde \sigma)})^\top)$ are bounded away from zero by assumption, we have (componentwise) that $(c^{(b,\sigma)})^{-1}$ and $(c^{(\tb,\tsigma)})^{-1}$ are controlled paths in $\crpW$. Applying Lemma~\ref{lem: product of controlled paths} and the bound in \eqref{eq: estimate for b and sigma}, we can show for each component that
\begin{equation}\label{eq: estimate for c}
\|((c^{(b,\sigma)})^{-1})^{ij}\|_{\crpW} \leq C,
\end{equation}
where $C > 0$ depends only on $p$, $m$, $d$, $M$ and $s_0$, and the same holds for $((c^{(\tb,\tsigma)})^{-1})^{ij}$.

By Lemma~\ref{lem: estimate product of controlled paths} and the fact that the composition of a controlled path with a regular function is locally Lipschitz, it follows from \eqref{eq: estimate for b and sigma} and \eqref{eq: Lipschitz estimate for b and sigma} that
\begin{equation}\label{eq: Lipschitz estimate for c}
\|((c^{(b,\sigma)})^{-1})^{ij}; ((c^{(\tb,\tsigma)})^{-1})^{ij}\|_{\crpW} \lesssim \|b - \tb\|_{C^2_b} + \|\sigma - \tsigma\|_{C^2_b},
\end{equation}
where the implicit multiplicative constant depends only on $p$, $m$, $d$, $M$ and $s_0$.

\emph{Step 4.}
We recall from the proof of Lemma~\ref{lem: pathwise optimal portfolio} that
\begin{equation*}
H^{(b,\sigma)} := (c^{(b,\sigma)})^{-1} b(\cdot,S^{(b,\sigma)}) \qquad \text{and} \qquad H^{(\tb,\tsigma)} := (c^{(\tb,\tsigma)})^{-1} \tb(\cdot,S^{(\tb,\tsigma)})
\end{equation*}
are controlled paths in $\crpW$. Moreover, Lemma~\ref{lem: product of controlled paths} and the estimates in \eqref{eq: estimate for b and sigma} and \eqref{eq: estimate for c} imply that
\begin{equation}\label{eq: estimate for H}
\|H^{(b,\sigma),i}\|_{\crpW} \leq C
\end{equation}
for each $i = 1, \ldots, m$, where $H^{(b,\sigma),i} = \sum_{j=1}^m ((c^{(b,\sigma)})^{-1})^{ij} b(\cdot,S^{(b,\sigma)})^j$, and $C > 0$ depends only on $p$, $m$, $d$, $M$ and $s_0$, and the same also holds for $H^{(\tb,\tsigma),i}$.

We then have from Lemma~\ref{lem: estimate product of controlled paths} that
\begin{align*}
\|&H^{(b,\sigma),i}; H^{(\tb,\tsigma),i}\|_{\crpW}\\
&\lesssim \sum_{j=1}^m \|((c^{(b,\sigma)})^{-1})^{ij}; ((c^{(\tb,\tsigma)})^{-1})^{ij}\|_{\crpW} + \|b(\cdot,S^{(b,\sigma)})^j; \tb(\cdot,S^{(\tb,\tsigma)})^j\|_{\crpW},
\end{align*}
and it follows from the estimates in \eqref{eq: Lipschitz estimate for b and sigma} and \eqref{eq: Lipschitz estimate for c} that
\begin{equation}\label{eq: Lipschitz estimate for H}
\|H^{(b,\sigma),i}; H^{(\tb,\tsigma),i}\|_{\crpW} \lesssim \|b - \tb\|_{C^2_b} + \|\sigma - \tsigma\|_{C^2_b},
\end{equation}
where the implicit multiplicative constant depends only on $p$, $m$, $d$, $M$ and $s_0$.

\emph{Step 5.}
Let $\vartheta^{(b,\sigma)} := \sigma(\cdot,S^{(b,\sigma)})^\top H^{(b,\sigma)}$ and $\vartheta^{(\tb,\tsigma)} := \tsigma(\cdot,S^{(\tb,\tsigma)})^\top H^{(\tb,\tsigma)}$. Then
\begin{equation*}
\theta^{(b,\sigma)} := \Big(\frac{1}{2} (\vartheta^{(b,\sigma)})^\top \vartheta^{(b,\sigma)}, (\vartheta^{(b,\sigma)})^\top\Big) \quad \text{and} \quad \theta^{(\tb,\tsigma)} := \Big(\frac{1}{2} (\vartheta^{(\tb,\tsigma)})^\top \vartheta^{(\tb,\tsigma)}, (\vartheta^{(\tb,\tsigma)})^\top\Big)
\end{equation*}
are controlled paths in $\crpW$, as, again, the product of controlled paths is itself a controlled path. By the same arguments as above, the estimates in \eqref{eq: estimate for b and sigma} and \eqref{eq: estimate for H} imply that
\begin{equation}\label{eq: estimate for theta}
\|\theta^{(b,\sigma)}\|_{\crpW} \leq C,
\end{equation}
where $C > 0$ depends only on $p$, $m$, $d$, $M$ and $s_0$, and the same holds for $\theta^{(\tb,\tsigma)}$, and the estimates in \eqref{eq: Lipschitz estimate for b and sigma} and \eqref{eq: Lipschitz estimate for H} imply moreover that
\begin{equation}\label{eq: Lipschitz estimate for theta}
\|\theta^{(b,\sigma)}; \theta^{(\tb,\tsigma)}\|_{\crpW} \lesssim \|b - \tb\|_{C^2_b} + \|\sigma - \tsigma\|_{C^2_b},
\end{equation}
where the implicit multiplicative constant also depends only on $p$, $m$, $d$, $M$ and $s_0$.

Define the rough integrals $U^{(b,\sigma)} := \int_0^\cdot \theta^{(b,\sigma)}_t \dd (\cdot,\bW)_t$ and $U^{(\tb,\tsigma)} := \int_0^\cdot \theta^{(\tb,\tsigma)}_t \dd (\cdot,\bW)_t$, which are controlled paths in $\crpW$. Using the estimate in \eqref{eq: estimate for rough integral} for the rough integral, and the bound in \eqref{eq: estimate for theta}, we obtain
\begin{equation}\label{eq: estimate for U}
\|U^{(b,\sigma)}\|_\infty \leq \|U^{(b,\sigma)}\|_{\crpW} \leq C_0,
\end{equation}
where $C_0 > 0$ depends only on $p$, $m$, $d$, $M$ and $s_0$, and the same holds for $U^{(\tb,\tsigma)}$. Furthermore, using the stability of rough integrals, e.g., \cite[Lemma~3.4]{Friz2018}, and the estimate~\eqref{eq: Lipschitz estimate for theta}, it follows immediately that
\begin{equation}\label{eq: Lipschitz estimate for U}
\|U^{(b,\sigma)}; U^{(\tb,\tsigma)}\|_{\crpW} \lesssim \|b - \tb\|_{C^2_b} + \|\sigma - \tsigma\|_{C^2_b},
\end{equation}
where the implicit multiplicative constant depends only on $p$, $m$, $d$, $M$ and $s_0$.

\emph{Step 6.}
Similarly to the proof of Lemma~\ref{lem: pathwise portfolio is optimal}, we let
\[Z^{(b,\sigma)} := \int_0^\cdot (H^{(b,\sigma)}_s)^\top \dd S^{(b,\sigma)}_s = \int_0^\cdot (H^{(b,\sigma)}_s)^\top b(s,S^{(b,\sigma)}_s) \dd s + \int_0^\cdot (H^{(b,\sigma)}_s)^\top \sigma(s,S^{(b,\sigma)}_s) \dd \bW_s,\]
and define $Z^{(\tb, \tsigma)}$ analogously, and then write $\bZ^{(b,\sigma)}$ for the canonical rough path lift of $Z^{(b,\sigma)}$, as in Lemma~\ref{lem: consistency of rough integrals}. We note in particular that
\[ [\bZ^{(b,\sigma)}] = \int_0^\cdot (\vartheta^{(b,\sigma)}_s)^\top \vartheta^{(b,\sigma)}_s \dd s = \int_0^\cdot (H^{(b,\sigma)}_s)^\top b(s,S^{(b,\sigma)}_s) \dd s. \]
Recalling the definition of the rough exponential $\cE$ from Lemma~\ref{lem: dynamics of rough exponential}, we then have that
\begin{align*}
\cE(Z^{(b,\sigma)})_t &= \exp \Big(Z^{(b,\sigma)}_t - \frac{1}{2}[\mathbf{Z}^{(b,\sigma)}]_t\Big)\\
&= \exp \bigg( \frac{1}{2} \int_0^t (\vartheta^{(b,\sigma)}_s)^\top \vartheta^{(b,\sigma)}_s \dd s + \int_0^t (\vartheta^{(b,\sigma)}_s)^\top \dd \bW_s \bigg) = \exp (U^{(b,\sigma)}_t)
\end{align*}
for each $t \in [0,T]$. By Lemma~\ref{lem: pathwise optimal portfolio}, we have that
\begin{equation*}
\kappa^{(b,\sigma)}_t := \frac{1}{K_T} \cE(Z^{(b, \sigma)})_t, \qquad \kappa^{(\tb, \tsigma)}_t := \frac{1}{K_T} \cE(Z^{(\tb, \tsigma)})_t
\end{equation*}
are the pathwise optimal consumption rates for the log-utility on the financial market modelled by $S^{(b,\sigma)}$ and $S^{(\tb, \tsigma)}$ respectively, and are controlled paths in $\crpW$. We deduce from the bound in \eqref{eq: estimate for U} that
\begin{equation}\label{eq: estimate for kappa}
\|\kappa^{(b,\sigma)}\|_{\crpW} \leq C,
\end{equation}
where $C > 0$ depends only on $p$, $m$, $d$, $M$, $s_0$, $\|\exp\|_{C^2_b(\{y : |y| \leq C_0\};\R)}$ and the consumption clock~$K$, and the same holds for $\kappa^{(\tb,\tilde \sigma)}$.

Because the composition of a controlled path with a regular function is locally Lipschitz continuous (see, e.g., \cite[Lemma~3.5]{Friz2018}), it follows from \eqref{eq: Lipschitz estimate for U} that
\begin{equation}\label{eq: Lipschitz estimate for kappa}
\|\kappa^{(b,\sigma)}; \kappa^{(\tb,\tsigma)}\|_{\crpW} \lesssim \|b - \tb\|_{C^2_b} + \|\sigma - \tsigma\|_{C^2_b},
\end{equation}
where the implicit multiplicative constant depends only on $p$, $m$, $d$, $M$, $s_0$ and $K$.

\emph{Step 7.}
Since $K_t$, $t \in [0,T]$, is a c{\`a}dl{\`a}g (deterministic) and increasing function (so of finite $1$-variation), we have by Lemma~\ref{lem: product of controlled paths} that the wealth process $V^{(b,\sigma)}_t := \kappa^{(b,\sigma)}_t (K_T - K_t)$, $t \in [0,T]$, (as the product of two controlled paths) is a controlled path in $\crpW$. We also have from \eqref{eq: estimate for kappa} that
\begin{equation}\label{eq: estimate for V}
\|V^{(b,\sigma)}\|_{\crpW} \leq C,
\end{equation}
where $C > 0$ depends only on $p$, $m$, $d$, $M$, $s_0$ and $K$. By Lemma~\ref{lem: estimate product of controlled paths} and \eqref{eq: Lipschitz estimate for kappa}, we then have that
\begin{equation}\label{eq: Lipschitz estimate for V}
\|V^{(b,\sigma)}; V^{(\tb,\tsigma)}\|_{\crpW} \lesssim \|\kappa^{(b,\sigma)}; \kappa^{(\tb,\tsigma)}\|_{\crpW} \lesssim \|b - \tb\|_{C^2_b} + \|\sigma - \tsigma\|_{C^2_b},
\end{equation}
where the implicit multiplicative constant depends only on $p$, $m$, $d$, $M$, $s_0$ and $K$.

\emph{Step 8.}
By Lemma~\ref{lem: product of controlled paths}, $\phi^{(b,\sigma),i} := H^{(b,\sigma),i} V^{(b,\sigma)}$ and $\phi^{(\tb, \tsigma),i} := H^{(\tb, \tsigma),i} V^{(\tb, \tsigma)}$ are controlled paths in $\crpW$ for each $i = 1, \ldots, m$, and, recalling \eqref{eq: estimate for H} and \eqref{eq: estimate for V}, we have that
\begin{equation*}
\|\phi^{(b,\sigma),i}\|_{\crpW} \leq C,
\end{equation*}
where $C > 0$ depends only on $p$, $m$, $d$, $M$, $s_0$ and $K$, and the same holds for $\phi^{(\tb, \tsigma),i}$. By Lemma~\ref{lem: estimate product of controlled paths}, combined with the estimates in \eqref{eq: Lipschitz estimate for H} and \eqref{eq: Lipschitz estimate for V}, we obtain
\begin{equation}\label{eq: Lipschitz estimate for phi_i}
\begin{split}
\|\phi^{(b,\sigma),i}; \phi^{(\tb,\tsigma),i}\|_{\crpW} &\lesssim \|H^{(b,\sigma),i}; H^{(\tb,\tsigma),i}\|_{\crpW} + \|V^{(b,\sigma),i}; V^{(\tb,\tsigma),i}\|_{\crpW}\\
&\lesssim \|b - \tb\|_{C^2_b} + \|\sigma - \tsigma\|_{C^2_b},
\end{split}
\end{equation}
where the implicit multiplicative constant depends only on $p$, $m$, $d$, $M$, $s_0$ and $K$. Finally, we consider
\begin{equation*}
\phi^{(b,\sigma),0}_t = \sum_{i=1}^m \bigg( \int_0^t \phi^{(b,\sigma),i}_s \dd S^{(b,\sigma),i}_s - \phi^{(b,\sigma),i}_t S^{(b,\sigma),i}_t \bigg)
\end{equation*}
for $t \in [0,T]$, with $\phi^{(\tb, \tsigma),0}_t$ defined similarly. By the associativity property of rough integrals, we can write
\begin{align*}
\int_0^\cdot \phi^{(b,\sigma),i}_s \dd S^{(b,\sigma),i}_s = \int_0^\cdot \psi^{(b,\sigma),i}_s \dd (\cdot, \bW)_s,
\end{align*}
where
\begin{equation*}
\psi^{(b,\sigma),i}_s := (\phi^{(b,\sigma),i}_s b(s,S^{(b,\sigma)}_s)^i, \phi^{(b,\sigma),i}_s \sigma(s,S^{(b,\sigma)}_s)^{i\cdot}),
\end{equation*}
with $\psi^{(\tb,\tsigma),i}$ defined similarly. It follows from \eqref{eq: Lipschitz estimate for b and sigma} and \eqref{eq: Lipschitz estimate for phi_i} that
\begin{equation*}
\|\psi^{(b,\sigma),i}; \psi^{(\tb,\tsigma),i}\|_{\crpW} \lesssim \|b - \tb\|_{C^2_b} + \|\sigma - \tsigma\|_{C^2_b},
\end{equation*}
for each $i = 1, \ldots, m$, and hence, by the stability of rough integration (e.g., \cite[Lemma~3.4]{Friz2018}) and the estimates in \eqref{eq: Lipschitz estimate for S} and \eqref{eq: Lipschitz estimate for phi_i}, that
\begin{equation*}
\|\phi^{(b,\sigma),0}; \phi^{(\tb,\tsigma),0}\|_{\crpW} \lesssim\|b - \tb\|_{C^2_b} + \|\sigma - \tsigma\|_{C^2_b},
\end{equation*}
where the implicit multiplicative constant depends only on $p$, $m$, $d$, $M$, $s_0$ and $K$.

Since we can bound the supremum norm by the controlled path norm (via the bound in \eqref{eq: bound on sup-norm}), the (local) Lipschitz continuity for optimal portfolios and wealth processes follows.
\end{proof}

In Theorem~\ref{thm: stability of optimal portfolios wrt model uncertainty} above we established stability of the optimal portfolio and the capital process with respect to perturbations in the model parameters, supposing that we always adopt the optimal portfolio according to the model considered.

However, from the perspective of stability with respect to \emph{model uncertainty}, we actually desire a slightly different statement. In particular, while we adopt the optimal portfolio according to the (possibly misspecified) model $(b,\sigma)$, we should consider the underlying price path $S$ as being governed by the ``true'' model $(\tb,\tsigma)$, and the capital process should be defined accordingly, accounting for this model discrepancy.

The path $\hat{V}^{(b,\sigma)}$ in \eqref{eq: defn hatV} below is a pathwise formulation of the wealth process $\bar{V}$ in \eqref{eq: defn wealth process}. In the case when $b$ and $\sigma$ are the true model parameters (i.e., the ones under which the price path $S$ is generated), one can verify that $\hat{V}^{(b,\sigma)}$ coincides with the path $V^{(b,\sigma)}$ in Lemma~\ref{lem: pathwise optimal portfolio}.

\begin{corollary}\label{corollary: model uncertainty for local vol model}
Recall the assumptions of Theorem~\ref{thm: stability of optimal portfolios wrt model uncertainty}, and let
\begin{equation}\label{eq: defn hatV}
\hat{V}^{(b,\sigma)}_t := 1 + \int_0^t (\phi^{(b,\sigma)}_s)^\top \dd S_s - \int_0^t \kappa^{(b,\sigma)}_s \dd K_s, \qquad t \in [0,T],
\end{equation}
with $\hat{V}^{(\tb,\tsigma)}$ defined analogously. Here, $\phi^{(b,\sigma)}$ and $\kappa^{(b,\sigma)}$ are defined as in Lemma~\ref{lem: pathwise optimal portfolio} above, except that we now take a common price path $S^{(b,\sigma)} = S^{(\tb,\tsigma)} =: S$ (which we may consider to be generated by the true model parameters). Then the stability estimate in \eqref{eq: stability est for optimal portfolio} holds, and we have moreover that
\begin{equation*}
\|\hat{V}^{(b,\sigma)}; \hat{V}^{(\tb,\tsigma)}\|_{\crpW} \lesssim \|b - \tb\|_{C^2_b} + \|\sigma - \tsigma\|_{C^2_b},
\end{equation*}
where the implicit multiplicative constant depends only on $p$, $m$, $d$, $M$, $s_0$ and $K$.
\end{corollary}

\begin{proof}
The estimate in \eqref{eq: stability est for optimal portfolio} holds by the same argument as in the proof of Theorem~\ref{thm: stability of optimal portfolios wrt model uncertainty}, except that one may essentially skip Step~1, since in this case $\|S^{(b,\sigma)}; S^{(\tb,\tsigma)}\|_{\crpW} = 0$.

Using the stability of rough and Young integration, we have that
\begin{equation*}
\|\hat{V}^{(b,\sigma)}; \hat{V}^{(\tb,\tsigma)}\|_{\crpW} \lesssim \|\phi^{(b,\sigma)}; \phi^{(\tb,\tsigma)}\|_{\crpW} + \|\kappa^{(b,\sigma)}; \kappa^{(\tb,\tsigma)}\|_p,
\end{equation*}
and we then simply recall the estimates in \eqref{eq: Lipschitz estimate for kappa} and \eqref{eq: Lipschitz estimate for phi_i}.
\end{proof}

\subsection{Discretization error of pathwise log-optimal portfolios}\label{subsec: discretization in local vola model}

To implement the pathwise log-optimal portfolio on a real financial market would require us to trade continuously in time. In reality, trading might be done at a very high frequency, but still on a discrete time grid and, thus, requires a discretization of any theoretically optimal portfolio.

We now introduce a discretized version of the pathwise log-optimal portfolio, as constructed in Lemma~\ref{lem: pathwise optimal portfolio}, and derive quantitative, pathwise error estimates resulting from this discretization for the portfolios as well as for their associated capital processes.

\smallskip

To define the discretized version of the pathwise log-optimal portfolio, we start by discretizing the underlying price paths. To this end, we recall that $W$ and the sequence of partitions $(\cP^n)_{n \in \N}$ satisfy Assumption~\ref{ass: partitions exhausting the jumps of K}, where $\cP^n = \{0 = t_0^n < t_1^n < \cdots < t_{N_n}^n = T\}$. For $n \in \N$, let $W^n \colon [0,T] \to \R^d$ be the piecewise constant approximation of $W$ along $\cP^n$, that is,
\begin{equation}\label{eq: defn W^n}
W^n_t := W_T \1_{\{T\}}(t) + \sum_{k=0}^{N_n-1} W_{t^n_k} \1_{[t^n_k,t^n_{k+1})}(t), \qquad t \in [0,T],
\end{equation}
and, setting $\gamma_t := t$, we define a time-discretization path along $\cP^n$ by
\begin{equation}\label{eq: defn gamma^n}
\gamma^n_t := T \1_{\{T\}}(t) + \sum_{k=0}^{N_n-1} t^n_k \1_{[t^n_k,t^n_{k+1})}(t), \qquad t \in [0,T].
\end{equation}

To discretize the price path $S$, we use the classical Euler approximation $S^n$ corresponding to the RDE \eqref{eq: dynamics of price process via RDE} along the partition $\cP^n$, which is given by
\begin{equation}\label{eq: dynamics of price process Euler scheme}
S^n_t = s_0 + \sum_{i \hspace{1pt} : \hspace{1pt} t^n_{i+1} \leq t} b(t^n_i,S^n_{t^n_i}) (t^n_{i+1} - t^n_i) + \sum_{i \hspace{1pt} : \hspace{1pt} t^n_{i+1} \leq t} \sigma(t^n_i,S^n_{t^n_i}) (W_{t^n_{i+1}} - W_{t^n_i})
\end{equation}
for $t \in [0,T]$. For each $n \in \N$, the discretized version of the pathwise log-optimal portfolio $(\phi^n, \kappa^n)$ is defined by
\begin{align*}
H^n_t &:= (\sigma(\gamma^n_t,S^n_t) \sigma(\gamma^n_t,S^n_t)^\top)^{-1} b(\gamma^n_t,S^n_t), \qquad \vartheta^n_t := \sigma(\gamma^n_t,S^n_t)^\top H^n_t,\\
\theta^n_t &:= \Big(\frac{1}{2} (\vartheta^n_t)^\top \vartheta^n_t, (\vartheta^n_t)^\top\Big),\\
\kappa^n_t &:= \frac{1}{K_T} \exp \bigg( \int_0^t \theta^n_s \dd (\gamma^n,W^n)_s \bigg), \qquad V^n_t := \kappa^n_t (K_T - K^n_t),\\
\phi^{n,i}_t &:= H^{n,i}_t V^n_t, \qquad i = 1, \ldots, m, \qquad \phi^{n,0}_t := \sum_{i=1}^m \int_0^t \phi^{n,i}_s \dd S^{n,i}_s - \phi^{n,i}_t S^{n,i}_t,
\end{align*}
for $t \in [0,T]$, where the integrals above are simply defined as left-point Riemann sums, and $K^n$ denotes the piecewise constant approximation of $K$ along $\cP^n$.

We also define the wealth processes associated with the log-optimal portfolio and with the discretized portfolio above, via
\begin{equation*}
\hat{V}_t := 1 + \int_0^t \phi_s^\top \dd S_s - \int_0^t \kappa_s \dd K_s, \qquad t \in [0,T],
\end{equation*}
and
\begin{equation}\label{eq: defn hat V^n}
\hat{V}^n_t := 1 + \int_0^t (\phi^n_s)^\top \dd S_s - \int_0^t \kappa^n_s \dd K_s, \qquad t \in [0,T].
\end{equation}

We note that the wealth process $\hat{V}^n$ is defined by considering the performance of the discretized portfolio $(\phi^n, \kappa^n)$ relative to the (non-discretized) price process $S$ and consumption clock $K$. Of course, since $(\phi, \kappa)$ is the (pathwise) log-optimal portfolio, one can verify that $\hat{V} = V$, where $V$ is the path defined in Lemma~\ref{lem: pathwise optimal portfolio}. On the other hand, it will typically be the case that $\hat{V}^n \neq V^n$. Indeed, $V^n$ is piecewise constant, whereas $\hat{V}^n$ represents the wealth generated by the discretized portfolio on the continuous (i.e., non-discretized) market.

For these discretized portfolios and their associated capital processes, we obtain the following convergence result with quantitative error estimates.

\begin{theorem}\label{thm: stability of optimal portfolios wrt discretization error}
For $(b, \sigma) \in C^3_b([0,T] \times \R^m; \R^m) \times \cA$, let $(\phi^{(b,\sigma)}, \kappa^{(b,\sigma)})$ be the pathwise log-optimal portfolio, as constructed in Lemma~\ref{lem: pathwise optimal portfolio}. Then
\begin{equation*}
\|(\phi^n,\kappa^n) - (\phi,\kappa)\|_{p'} \, \longrightarrow \, 0 \qquad \text{as} \quad n \, \longrightarrow \, \infty
\end{equation*}
and
\begin{equation*}
\|\hat{V}^n - \hat{V}\|_{p'} \, \longrightarrow \, 0 \qquad \text{as} \quad n \, \longrightarrow \, \infty,
\end{equation*}
for any $p' \in (p,3)$, with a rate of convergence given by
\begin{align*}
&\|(\phi^n,\kappa^n) - (\phi,\kappa)\|_{p'} + \|\hat{V}^n - \hat{V}\|_{p'}\\
&\lesssim |\cP^n|^{(1-\frac{1}{q}) (1-\frac{p}{p'})} + \|W^n - W\|_\infty^{1-\frac{p}{p'}} + \bigg\| \int_0^\cdot W^n_t \otimes \dd W_t - \int_0^\cdot W_t \otimes \dd W_t \bigg\|_\infty^{1-\frac{p}{p'}},
\end{align*}
for any $q \in (1,2)$ such that $\frac{1}{p'} + \frac{1}{q} > 1$, where the implicit multiplicative constant depends only on $p$, $p'$, $q$, $m$, $d$, $\|b\|_{C^3_b}$, $\|\sigma\|_{C^3_b}$, $\sup_{(t,x)} |\det(\sigma(t,x)\sigma(t,x)^\top)|^{-1}$, $s_0$, $|W_0|$ and $w(0,T)$, where $w$ is the control function for which \eqref{eq: RIE inequality} holds for $(\cdot,W)$. Moreover, the same convergence and estimates also hold when the $p'$-variation norm is replaced by the supremum norm.
\end{theorem}

Before we present the proof of Theorem~\ref{thm: stability of optimal portfolios wrt discretization error}, some preliminary steps are necessary. We start by noting that, since $W^n$ has finite $1$-variation, it admits a canonical rough path lift $\bW^n = (W^n,\W^n) \in \cD^p([0,T];\R^d)$, with $\W^n$ given by
\begin{equation*}
\W^n_{s,t} := \int_s^t W^n_{s,u} \otimes \d W^n_u, \qquad (s,t) \in \Delta_T,
\end{equation*}
where the integral is defined as a classical limit of left-point Riemann sums. Similarly, we can define a time-augmented rough path $(\cdot,\bW)^n$ of $(\cdot,W)^n := (\gamma^n,W^n)$.

\begin{lemma}\label{lem: rate of convergence of space-time rough paths}
There exists a constant $C > 0$, which depends only on $p$ and $w(0,T)$, where $w$ is the control function for which \eqref{eq: RIE inequality} holds for $(\cdot,W)$, such that
\begin{equation*}
\|(\cdot,\bW)^n\|_p + \|(\cdot,\bW)\|_p \leq C,
\end{equation*}
for every $n \in \N$. Moreover, for any $p' \in (p,3)$, we have that
\begin{equation*}
\|(\cdot,\bW)^n;(\cdot,\bW)\|_{p'} \, \longrightarrow \, 0 \qquad \text{as} \quad n \, \longrightarrow \, \infty,
\end{equation*}
with a rate of convergence given by
\begin{align*}
\|&(\cdot,\bW)^n;(\cdot,\bW)\|_{p'}\\
&\lesssim |\cP^n|^{(1-\frac{1}{q}) (1-\frac{p}{p'})} + \|W^n - W\|_\infty^{1-\frac{p}{p'}} + \bigg\| \int_0^\cdot W^n_t \otimes \dd W_t - \int_0^\cdot W_t \otimes \dd W_t \bigg\|_\infty^{1-\frac{p}{p'}},
\end{align*}
for any $q \in (1,2)$ such that $\frac{1}{p} + \frac{1}{q} > 1$, where the implicit multiplicative constant depends only on $p$, $p'$, $q$, $d$, $|W_0|$ and $w(0,T)$.
\end{lemma}

\begin{proof}
Since $(\cdot,W)$ satisfies Property \textup{(RIE)}, it is known (from, e.g., the proof of \cite[Lemma~2.13]{Allan2023b}) that its canonical rough path lift satisfies
\begin{equation*}
\|(\cdot,\bW)\|_p \leq C,
\end{equation*}
for some constant $C > 0$ which depends only on $p$ and $w(0,T)$, where $w$ is the control function for which \eqref{eq: RIE inequality} holds for $(\cdot,W)$. By \cite[Lemma~2.7]{Allan2023c}, we also have that
\begin{equation*}
\|(\cdot,\bW)^n\|_p \leq C
\end{equation*}
for every $n \in \N$, for some $C > 0$ which similarly depends only on $p$ and $w(0,T)$.

By \cite[Lemma~2.8]{Allan2023c}, we have that
\begin{align*}
&\|(\cdot,\bW)^n;(\cdot,\bW)\|_{p'}\\
&\lesssim \|(\gamma^n,W^n) - (\gamma,W)\|_\infty^{1-\frac{p}{p'}}\\
&\quad + \sup_{(s,t) \in \Delta_T} \bigg| \int_s^t (\gamma^n,W^n)_{s,u} \otimes \dd (\gamma,W)_u - \int_s^t (\gamma,W)_{s,u} \otimes \dd (\gamma,W)_u \bigg|^{1-\frac{p}{p'}}\\
&\lesssim \|(\gamma^n,W^n) - (\gamma,W)\|_\infty^{1-\frac{p}{p'}} + \bigg\| \int_0^\cdot (\gamma^n,W^n)_u \otimes \dd (\gamma,W)_u -\int_0^\cdot (\gamma,W)_u \otimes \dd (\gamma,W)_u \bigg\|_\infty^{1-\frac{p}{p'}},
\end{align*}
where the implicit multiplicative constants depend only on $p$, $p'$, $|W_0|$ and $w(0,T)$.

It is straightforward to see that
\begin{equation*}
\|(\gamma^n,W^n) - (\gamma,W)\|_\infty \lesssim |\cP^n| + \|W^n - W\|_\infty
\end{equation*}
for every $n \in \N$. Let $q \in (1,2)$ such that $\frac{1}{p} + \frac{1}{q} > 1$. We note that $\gamma^n$ and $\gamma$ have finite $1$-variation, with $\|\gamma^n\|_1 = \|\gamma\|_1 = T$, and it follows by interpolation that
\begin{equation*}
\|\gamma^n - \gamma\|_p \leq (2T)^{\frac{1}{q}} |\cP^n|^{1 - \frac{1}{q}}.
\end{equation*}
By the standard estimate for Young integrals---see, e.g., \cite[Proposition~2.4]{Friz2018}---we then have, for all $t \in [0,T]$, that
\begin{equation*}
\bigg| \int_0^t \gamma^n_u \dd \gamma_u - \int_0^t \gamma_u \dd \gamma_u \bigg| \lesssim \|\gamma^n - \gamma\|_p \|\gamma\|_1 \lesssim |\cP^n|^{1-\frac{1}{q}}.
\end{equation*}
Similarly, for each $j = 1, \ldots, d$ and $t \in [0,T]$, we have that
\begin{equation*}
\bigg| \int_0^t \gamma^n_u \dd W^j_u - \int_0^t \gamma_u \dd W^j_u \bigg| \lesssim \|\gamma^n - \gamma\|_p \|W\|_p \lesssim |\cP^n|^{1-\frac{1}{q}}
\end{equation*}
and, since $\gamma_u = u$, we can simply bound
\begin{equation*}
\bigg| \int_0^t W^{n,j}_u \dd \gamma_u - \int_0^t W^j_u \dd \gamma_u \bigg| \leq \|W^n - W\|_\infty T.
\end{equation*}
Combining the estimates above, we obtain the stated rate of convergence.
\end{proof}

It was established in the proof of \cite[Theorem~2.3]{Allan2023c} that the Euler scheme \eqref{eq: dynamics of price process Euler scheme} coincides with the solution to the RDE given by
\begin{equation*}
\widetilde{S}^n_t = s_0 + \int_0^t b(\gamma^n_s, \widetilde{S}^n_s) \dd \gamma^n_s + \int_0^t \sigma(\gamma^n_s, \widetilde{S}^n_s) \dd \bW^n_s = s_0 + \int_0^t (b,\sigma)(\gamma^n_s,\widetilde{S}^n_s) \dd (\cdot, \bW)^n_s
\end{equation*}
for $t \in [0,T]$. Furthermore, proceeding as in the proof of Lemma~\ref{lem: pathwise optimal portfolio}, one can show that $H^n$, $\theta^n$, $\kappa^n$, $\phi^n$ are controlled paths in $\cV^p_{(\cdot,W)^n}$, and that $\int_0^t \phi^{n,i}_s \dd S^{n,i}_s$ is a rough integral, defined in the sense of Lemma~\ref{lem: integration against controlled path}.

Let $\widetilde{\cP}^m = \{0 = r^m_0 < r^m_1 < \dots < r^m_{\widetilde{N}_m} = T\}$, $m \in \N$, be any sequence of partitions with vanishing mesh size, such that $\cP^n \subseteq \widetilde{\cP}^m$ for every $m \in \N$. Exactly as in the proof of \cite[Theorem~2.3]{Allan2023c}, one can show that the rough integral $\int_0^t \theta^n_s \dd (\cdot,\bW)^n_s$ is equal to a limit of left-point Riemann sums along the sequence $(\widetilde{\cP}^m)_{m \in \N}$, and thus coincides with the corresponding Riemann--Stieltjes integral. That is, for any $t \in [0,T]$, we have that
\begin{align*}
\int_0^t \theta^n_s \dd (\cdot,\bW)^n_s &= \lim_{m \to \infty} \sum_{k=0}^{\widetilde{N}_m-1} \theta^n_{r^m_k} (\gamma^n,W^n)_{r^m_k \wedge t, r^m_{k+1} \wedge t}\\
&= \sum_{k=0}^{N_n-1} \theta^n_{t^n_k} (\gamma^n,W^n)_{t^n_k \wedge t, t^n_{k+1} \wedge t} = \int_0^t \theta^n_s \dd (\gamma^n,W^n)_s,
\end{align*}
and, consequently, we have that
\begin{equation*}
\kappa^n_t := \frac{1}{K_T} \exp \bigg( \int_0^t \theta^n_s \dd (\gamma^n,W^n)_s \bigg)= \frac{1}{K_T} \exp \bigg( \int_0^t \theta^n_s \dd (\cdot,\bW)^n_s \bigg).
\end{equation*}
Similarly, by the associativity of rough integrals and Theorem~\ref{thm: rough int against controlled path under RIE}, the rough integral of the controlled path $\phi^{n,i}$ against the controlled path $S^{n,i}$ is given by a left-point Riemann sum.

\begin{proof}[Proof of Theorem~\ref{thm: stability of optimal portfolios wrt discretization error}]
By part~(i) of \cite[Corollary~3.3]{Kwossek2024} and Lemma~\ref{lem: rate of convergence of space-time rough paths}, we deduce that there exists a constant $L > 0$, which depends only on $p'$, $T$, $s_0$, $\|b\|_{C^2_b}$, $\|\sigma\|_{C^2_b}$, $|W_0|$ and $w(0,T)$, such that $\sup_{n \in \N} \|S^n\|_{\cV^{p'}_{(\cdot,W)^n}} \leq L$ and $\|S\|_{\cV^{p'}_{(\cdot,W)}} \leq L$. Further, \cite[Theorem~2.2]{Allan2023c} implies that
\begin{equation*}
\|S^n; S\|_{\cV^{p'}_{(\cdot,W)^n},\cV^{p'}_{(\cdot,W)}} \lesssim \|\gamma^n - \gamma\|_p + \|(\cdot,\bW)^n;(\cdot,\bW)\|_{p'} \lesssim |\cP^n|^{1-\frac{1}{q}} + \|(\cdot,\bW)^n;(\cdot,\bW)\|_{p'},
\end{equation*}
for any $q \in (1,2)$ such that $\frac{1}{p'} + \frac{1}{q} > 1$, where the implicit multiplicative constant depends only on $p'$, $q$, $T$, $\|b\|_{C^3_b}$, $\|\sigma\|_{C^3_b}$, $|W_0|$ and $w(0,T)$.

We note that $\gamma^n$ and $\gamma$ are controlled paths with respect to $(\cdot,W)^n$ and $(\cdot,W)$ respectively, with $\|\gamma^n\|_{\cV^{p'}_{(\cdot,W)^n}} = \|\gamma\|_{\cV^{p'}_{(\cdot,W)}} = 1$ and $\|\gamma^n;\gamma\|_{\cV^{p'}_{(\cdot,W)^n},\cV^{p'}_{(\cdot,W)}} = 0$. Since the composition of a controlled path with a regular function remains a controlled path, and such a composition is locally Lipschitz (see, e.g., \cite[Lemma~3.5]{Friz2018}), we have that
\begin{equation*}
\|(b,\sigma)(\gamma^n,S^n)\|_{\cV^{p'}_{(\cdot,W)^n}} + \|(b,\sigma)(\gamma,S)\|_{\cV^{p'}_{(\cdot,W)}} \leq C,
\end{equation*}
where $C > 0$ depends only on $p'$, $T$, $\|b\|_{C^2_b}$, $\|\sigma\|_{C^2_b}$, $L$ and $\|W\|_p$, and consequently
\begin{equation}\label{eq: gamma^n S^n gamma S bound}
\begin{split}
\|&(b,\sigma)(\gamma^n,S^n);(b,\sigma)(\gamma,S)\|_{\cV^{p'}_{(\cdot,W)^n},\cV^{p'}_{(\cdot,W)}}\\
&\lesssim \|\gamma^n;\gamma\|_{\cV^{p'}_{(\cdot,W)^n},\cV^{p'}_{(\cdot,W)}} +  \|S^n; S\|_{\cV^{p'}_{(\cdot,W)^n},\cV^{p'}_{(\cdot,W)}} + \|(\cdot,W)^n - (\cdot,W)\|_{p'}\\
&\lesssim |\cP^n|^{1-\frac{1}{q}} + \|(\cdot,\bW)^n; (\cdot,\bW)\|_{p'},
\end{split}
\end{equation}
where the implicit multiplicative constants depend only on $p'$, $q$, $T$, $\|b\|_{C^2_b}$, $\|\sigma\|_{C^2_b}$, $L$, $|W_0|$ and $w(0,T)$.

Similarly to the proof of Theorem~\ref{thm: stability of optimal portfolios wrt model uncertainty}, by standard estimates for controlled paths (e.g., \eqref{eq: estimate for rough integral} and Lemma~\ref{lem: product of controlled paths}), it follows from the estimates above that
\begin{equation*}
\|(\phi^n,\kappa^n)\|_{\cV^{p'}_{(\cdot,W)^n}} + \|(\phi,\kappa)\|_{\cV^{p'}_{(\cdot,W)}} \leq C,
\end{equation*}
where $C > 0$ depends only on $p'$, $m$, $d$, $T$, $\|b\|_{C^2_b}$, $\|\sigma\|_{C^2_b}$, $L$, $|W_0|$, $w(0,T)$ and the consumption clock $K$.

Since $\kappa^n$ and $\varphi^n$ depend on (the composition of regular functions with, or products of) controlled paths of the form $(\sigma(\gamma^n_t,S^n_t) \sigma(\gamma^n_t,S^n_t)^\top)^{-1} b(\gamma^n_t,S^n_t)$ and $\sigma(\gamma^n_t,S^n_t)$, and $\kappa$ and $\varphi$ depend similarly on controlled paths of the form $(\sigma(\gamma_t,S_t) \sigma(\gamma_t,S_t)^\top)^{-1} b(\gamma_t,S_t)$ and $\sigma(\gamma_t,S_t)$, standard estimates for rough integrals (e.g., \cite[Proposition~3.4]{Friz2018}) combined with the bound in \eqref{eq: gamma^n S^n gamma S bound} then imply that
\begin{equation*}
\|(\phi^n,\kappa^n),(\phi,\kappa)\|_{\cV^{p'}_{(\cdot,W)^n},\cV^{p'}_{(\cdot,W)}} \lesssim |\cP^n|^{1-\frac{1}{q}} + \|(\cdot,\bW)^n;(\cdot,\bW)\|_{p'},
\end{equation*}
where the implicit multiplicative constant depends only on $p'$, $q$, $m$, $d$, $T$, $\|b\|_{C^3_b}$, $\|\sigma\|_{C^3_b}$, $\sup_{(t,x)} |\det(\sigma(t,x) \sigma(t,x)^\top)|^{-1}$, $s_0$, $|W_0|$, $w(0,T)$ and $K$.

Letting $\widetilde{\W}^n_{s,t} = \int_s^t W^n_{s,u} \otimes \d W_u$ for $(s,t) \in \Delta_{[0,T]}$, the triple $\widetilde{\bW}^n := (W,W^n,\widetilde{\W}^n)$ is then a rough path in the sense of \cite[Definition~2.1]{Allan2023b}, and as usual we can define the natural time-augmentation of this rough path, which we denote by $(\cdot,\widetilde{\bW})^n$. Since $S$ is a controlled path with respect to $(\cdot,W)$, while $\phi^n$ is a controlled path with respect to $(\cdot,W)^n$, the integral $\int_0^t (\phi^n_s)^\top \dd S_s$ in \eqref{eq: defn hat V^n} is a well-defined rough integral (relative to $(\cdot,\widetilde{\bW})^n$), in the sense of \cite[Proposition~2.4]{Allan2023b}.

By the stability estimate for such rough integrals in part~(ii) of \cite[Proposition~2.7]{Allan2023b}, along with a standard estimate for Young integrals (e.g., \cite[Proposition~2.4]{Friz2018}), we obtain
\begin{align*}
\|\hat{V}^n - \hat{V}\|_{p'} &\lesssim \|\phi^n, \phi\|_{\cV^{p'}_{(\cdot,W)^n},\cV^{p'}_{(\cdot,W)}} + \|(\cdot,\widetilde{\bW})^n; (\cdot,\bW)\|_{p'} + \|\kappa^n - \kappa\|_{p'}\\
&\lesssim |\cP^n|^{1-\frac{1}{q}} + \|(\cdot,\bW)^n;(\cdot,\bW)\|_{p'} + \|(\cdot,\widetilde{\bW})^n; (\cdot,\bW)\|_{p'}.
\end{align*}
The term $\|(\cdot,\widetilde{\bW})^n; (\cdot,\bW)\|_{p'}$ can be bounded similarly to how $\|(\cdot,\bW)^n;(\cdot,\bW)\|_{p'}$ was bounded in Lemma~\ref{lem: rate of convergence of space-time rough paths}. In particular, we have by interpolation that
\begin{align*}
\|\widetilde{\W}^n - \W\|_{\frac{p'}{2}} &\lesssim \sup_{(s,t) \in \Delta_T} \bigg| \int_s^t W^n_{s,u} \otimes \dd W_u - \int_s^t W_{s,u} \otimes \dd W_u \bigg|^{1-\frac{p}{p'}}\\
&\lesssim \|W^n - W\|_\infty^{1-\frac{p}{p'}} + \bigg\| \int_0^\cdot W^n_u \otimes \dd W_u -\int_0^\cdot W_u \otimes \dd W_u \bigg\|_\infty^{1-\frac{p}{p'}}.
\end{align*}

Hence, since we can bound the $p$-variation norm by the controlled path norm (see \eqref{eq: bound on p var norm}), combining the estimates derived above with the rate of convergence given in Lemma~\ref{lem: rate of convergence of space-time rough paths}, we infer the desired convergence and estimate.
\end{proof}

If we assume stronger regularity properties of the ``driving noise'' path $W$ and the sequence of partitions, we can make the quantitative estimates provided in Theorem~\ref{thm: stability of optimal portfolios wrt discretization error} more explicit. In particular, considering the regularity properties of Brownian sample paths, we can derive the following two corollaries.

\begin{corollary}\label{cor: convergence rate for BM und uniform partitions}
Let $p \in (2,3)$ and let $(\cP^n_U)_{n \in \N}$ be the sequence of equidistant partitions of $[0,T]$ with width $\frac{T}{n}$. Let $W$ be a $\frac{1}{p}$-H{\"o}lder continuous path satisfying Assumption~\ref{ass: partitions exhausting the jumps of K} relative to $p$ and $(\cP^n_U)_{n \in \N}$, and suppose that
\begin{equation}\label{eq: convergence rate BM and uniform partitions}
\bigg\| \int_0^\cdot W^n_t \otimes \d W_t - \int_0^\cdot W_t \otimes \d W_t \bigg \|_{\infty} \lesssim n^{-(\frac{2}{p} - \beta)}, \qquad n \in \N,
\end{equation}
for some $\beta \in (1-\frac{1}{p},\frac{2}{p})$. Then, for any $p' \in (p,3)$ and $q \in (1,2)$ such that $\frac{1}{p'} + \frac{1}{q} > 1$, there exists a constant $C > 0$, which does not depend on $n$, such that, for every $n \in \N$,
\begin{equation*}
\|(\phi^n,\kappa^n) - (\phi,\kappa)\|_{p'} + \|\hat{V}^n - \hat{V}\|_{p'} \leq C(n^{-(1-\frac{1}{q})(1-\frac{p}{p'})} + n^{-(\frac{2}{p}-\beta)(1-\frac{p}{p'})}).
\end{equation*}
\end{corollary}

\begin{proof}
Since $W$ is assumed to be $\frac{1}{p}$-H{\"o}lder continuous, we have that
\begin{equation*}
\|W^n - W\|_\infty \lesssim n^{-\frac{1}{p}}, \qquad n \in \N.
\end{equation*}
Combining this with \eqref{eq: convergence rate BM and uniform partitions} and Theorem~\ref{thm: stability of optimal portfolios wrt discretization error}, and noting that $\frac{1}{p} < 1 - \frac{1}{p} < \beta$ so that in particular $\frac{2}{p} - \beta < \frac{1}{p}$, we obtain the claimed rate of convergence.
\end{proof}

\begin{remark}
Almost all sample paths of a $d$-dimensional Brownian motion satisfy Property \textup{(RIE)} relative to $p$ and $(\cP^n_U)_{n \in \N}$, as shown in \cite[Proposition~3.2]{Allan2023c}, and, thus, Assumption~\ref{ass: partitions exhausting the jumps of K} is satisfied if the sequence $(\cP^n_U)_{n \in \N}$ of partitions exhausts the jumps of the consumption clock~$K$. Moreover, by \cite[Appendix~B]{Perkowski2016}, the bound in \eqref{eq: convergence rate BM and uniform partitions} holds almost surely for sample paths of Brownian motion. Hence, the sample paths of Brownian motion fulfill the conditions of Corollary~\ref{cor: convergence rate for BM und uniform partitions} almost surely.
\end{remark}

\begin{corollary}\label{cor: convergence rate for BM und dyadic partitions}
Let $p \in (2,3)$ and let $(\cP^n_D)_{n \in \N}$ be the sequence of dyadic partitions of $[0,T]$, given by
\begin{equation*}
\cP^n_D := \{0 = t_0^n < t_1^n < \dots < t^n_{2^n} = T \} \quad \text{with} \quad t_k^n := k2^{-n}T \quad \text{for} \quad k = 0, 1, \ldots, 2^n.
\end{equation*}
For some $\epsilon \in (0,1)$, let $W$ be a $\frac{1}{p}$-H{\"o}lder continuous path which satisfies Assumption~\ref{ass: partitions exhausting the jumps of K} relative to $p$ and $(\cP^n_D)_{n \in \N}$, and suppose also that
\begin{equation}\label{eq: convergence rate BM and dyadic partitions}
\bigg\| \int_0^\cdot W^n_t \otimes \d W_t - \int_0^\cdot W_t \otimes \d W_t \bigg\|_\infty \lesssim 2^{-\frac{n}{2}(1-\epsilon)}, \qquad n \in \N.
\end{equation}
Then, for any $p' \in (p,3)$ and $q \in (1,2)$ such that $\frac{1}{p'} + \frac{1}{q} > 1$, there exists a constant $C > 0$, which does not depend on $n$, such that, for every $n \in \N$,
\begin{equation*}
\|(\phi^n,\kappa^n) - (\phi,\kappa)\|_{p'} + \|\hat{V}^n - \hat{V}\|_{p'} \leq C(2^{-n(1-\frac{1}{q})(1-\frac{p}{p'})} + 2^{-n(\frac{1}{p}-\frac{1}{p'})} + 2^{-\frac{n}{2}(1-\epsilon)(1-\frac{p}{p'})}).
\end{equation*}
\end{corollary}

\begin{proof}
Since $W$ is assumed to be $\frac{1}{p}$-H{\"o}lder continuous, we have that
\begin{equation*}
\|W^n - W\|_\infty \lesssim 2^{-\frac{n}{p}}, \qquad n \in \N.
\end{equation*}
Combining this with the bound in \eqref{eq: convergence rate BM and dyadic partitions} and Theorem~\ref{thm: stability of optimal portfolios wrt discretization error}, we deduce the claimed rate of convergence.
\end{proof}

\begin{remark}
  Almost all sample paths of a $d$-dimensional Brownian motion satisfy Property \textup{(RIE)} relative to $p$ and $(\cP^n_D)_{n \in \N}$ almost surely, as shown in \cite[Proposition~3.6]{Allan2023c}, and, thus, Assumption~\ref{ass: partitions exhausting the jumps of K} is satisfied if the sequence of partitions $(\cP^n_U)_{n \in \N}$ exhausts the jumps of the consumption clock~$K$. Moreover, as shown in the proof of part~(ii) of \cite[Proposition~3.6]{Allan2023c}, for any $\epsilon \in (0,1)$, almost all sample paths of Brownian motion satisfy
  \begin{equation*}
    \bigg\| \int_0^\cdot W^n_t \otimes \d W_t - \int_0^\cdot W_t \otimes \d W_t \bigg\|_\infty < 2^{-\frac{n}{2}(1-\epsilon)},
  \end{equation*}
  for all sufficiently large $n$. Hence, the sample paths of Brownian motion fulfill the conditions of Corollary~\ref{cor: convergence rate for BM und dyadic partitions} almost surely.
\end{remark}

\section{Black--Scholes-type models: pathwise analysis of log-optimal portfolios}\label{sec: BS model}

In this section we shall study the log-optimal portfolio for the investment-consumption problem, acting on deterministic price paths generated by Black--Scholes-type models, which are defined in a pathwise manner.

For readability, this section is structured similarly to Section~\ref{sec: local vol models}. We point out that similar arguments apply and the method of proof carries over, but due to the unboundedness of the coefficients, this case needs to be treated separately.

\medskip

To this end, we shall work again under the standing Assumption~\ref{ass: partitions exhausting the jumps of K}. For convenience, in this section we shall also assume that the path $W$ is continuous (but see Remark~\ref{remark: W doesn't need to be continuous} below). For clarity, these assumptions are stated precisely below.

\begin{assumption}\label{assumption: partitions and W for Black-Scholes model}
Let $p \in (2,3)$ and let $\cP^n = \{0 = t_0^n < t_1^n < \cdots < t_{N_n}^n = T\}$, $n \in \N$, be a sequence of equidistant partitions of the interval $[0,T]$, such that,
   \begin{itemize}
     \item[$\vcenter{\hbox{\scriptsize$\bullet$}}$] for each $n \in \N$, there exists some $\pi_n > 0$ such that $t^n_{i+1} - t^n_i = \pi_n$ for each $0 \leq i < N_n$,
     \item[$\vcenter{\hbox{\scriptsize$\bullet$}}$] $\pi_n^{2 - \frac{4}{p}} \log(n) \to 0$ as $n \to \infty$,
     \item[$\vcenter{\hbox{\scriptsize$\bullet$}}$] $J_K \subseteq \liminf_{n \to \infty} \cP^n$ with $J_K := \{t \in (0,T] : K_{t-} \neq K_t\}$,
   \end{itemize}
where the consumption clock $K \colon [0,T] \to \R$ is fixed as in Section~\ref{subsec: setting}. Moreover, the deterministic path $W \colon [0,T] \to \R^d$ is continuous, and satisfies Property \textup{(RIE)} relative to $p$ and $(\mathcal{P}^n)_{n \in \N}$.
\end{assumption}

We also adopt the following assumption.

\begin{assumption}\label{assumption: b and sigma in BS model}
For each $i = 1, \ldots, m$, we assume that $b^i \in \cV^p_W([0,T];\R)$ and $\sigma^{i\cdot} \in \cV^p_W([0,T];\cL(\R^d;\R))$ are controlled paths with respect to $W$, such that $\sigma_t \sigma_t^\top \in \mathrm{GL}(\R^{m \times m})$ for all $t \in [0,T]$, with the determinant being uniformly bounded away from zero, and such that the jump times of $b^i$ and $\sigma^{i\cdot}$ belong to $\liminf_{n \to \infty} \cP^n$.
\end{assumption}

The components of the discounted price path $S = (S_t)_{t \in [0,T]}$ are now assumed to satisfy the linear rough differential equations
\begin{equation}\label{eq: pathwise Black Scholes model}
  S^i_t = s^i_0 + \int_0^t S^i_s b^i_s \dd s + \int_0^t S^i_s \sigma^{i \cdot}_s \dd \bW_s, \qquad t \in [0,T], \qquad i = 1, \ldots, m,
\end{equation}
where $s_0^i > 0$, and $\bW = (W, \mathbb{W})$ is the canonical rough path lift of $W$, as defined in \eqref{eq: RIE rough path}.

By Lemma~\ref{lem: consistency of rough integrals} and Proposition~\ref{prop: associativity of rough integration}, this is equivalent to saying that $S$ satisfies the linear rough differential equations
\begin{equation}\label{eq: linear RDE}
  S^i_t = s_0^i + \int_0^t S^i_s \dd \mathbf{\Xi}^i_s, \qquad t \in [0,T], \qquad i = 1, \ldots, m,
\end{equation}
where $\Xi^i := \int_0^\cdot b^i_t \dd t + \int_0^\cdot \sigma^{i\cdot}_t \dd \bW_t$, which is a controlled path in $\cV^p_W$, and thus admits a canonical rough path lift $\mathbf{\Xi}^i$ by Lemma~\ref{lem: consistency of rough integrals}. By Theorem~\ref{thm: controlled path satisfies RIE}, $\Xi^i$ itself satisfies Property \textup{(RIE)} relative to $p$ and $(\cP^n)_{n \in \N}$. Hence, by Lemma~\ref{lem: dynamics of rough exponential}, we have that the solution to the linear RDE above is given by the rough exponential $S^i = s^i_0 \cE(\Xi^i)$.

\begin{remark}\label{remark: W doesn't need to be continuous}
We supposed in Assumption~\ref{assumption: partitions and W for Black-Scholes model} that the path $W$ is continuous. However, the analysis in this section also works for c\`adl\`ag paths $W$ which satisfy Property \textup{(RIE)} (noting in particular Remark~\ref{remark: RIE paths satisfy conditions for rough exp}). The only restriction is that one must ensure that, for every $i = 1, \ldots, m$, the price path $S^i$ remains positive for all $t \in [0,T]$, which, by Lemma~\ref{lem: dynamics of rough exponential}, is equivalent to ensuring that $\Delta \Xi^i_t \neq -1$ for all $t \in (0,T]$.
\end{remark}

\begin{remark}
  If $W$ is a realization of a Brownian motion, the dynamics of the RDE \eqref{eq: pathwise Black Scholes model} can be seen as a fixed realization of a Black--Scholes-type model for a financial market.

Indeed, let us assume that $\bar{W} = (\bar{W}_t)_{t\in [0,T]}$ is a $d$-dimensional Brownian motion on a probability space $(\Omega,\mathcal{F},\P)$ with respect to an underlying filtration $(\cF_t)_{t \in [0,T]}$. Let us suppose that $\bar{b}^i$ and $\bar{\sigma}^{ij}$, $i = 1, \ldots, m$, $j = 1, \ldots, d$, are predictable processes such that, for almost every $\omega \in \Omega$, the paths $b = \bar{b}(\omega)$ and $\sigma = \bar{\sigma}(\omega)$ satisfy Assumption~\ref{assumption: b and sigma in BS model}. It is well-known that the linear stochastic differential equations
  \begin{equation}\label{eq: Black Scholes model}
    \bar{S}^i_t = s^i_0 + \int_0^t \bar{S}^i_s \bar{b}^i_s \dd s + \int_0^t \bar{S}^i_s \bar{\sigma}^{i\cdot}_s \dd \bar{W}_s, \qquad t \in [0,T], \qquad i = 1, \ldots, m,
  \end{equation}
  admit a unique strong solution, where $\int_0^t \bar{S}^i_s \bar{\sigma}^{i\cdot}_s \dd \bar{W}_s$ denotes a stochastic It{\^o} integral; see, e.g., \cite[Ch.~V, Theorem~6]{Protter2005}. We note that this model $(\bar{S}_t)_{t \in [0,T]}$ includes the classical Black--Scholes model, and some stochastic volatility models, where the volatility is modelled by an SDE driven by $\bar{W}$.

  The solution to the SDE \eqref{eq: Black Scholes model} is given by the stochastic exponential $\bar{S}^i = s_0^i \bar{\cE}(\bar{\Xi}^i)$, where $\bar{\Xi}^i := \int_0^\cdot \bar{b}^i_t \dd t + \int_0^\cdot \bar{\sigma}^{i\cdot}_t \dd \bar{W}_t$, that is,
  \begin{equation*}
    \bar{S}^i_t = s_0^i \exp \Big(\bar{\Xi}^i_t - \frac{1}{2} [\bar{\Xi}^i]_t\Big), \qquad t \in [0,T], \qquad i = 1, \ldots, m,
  \end{equation*}
  where $[\bar{\Xi}^i]$ denotes the quadratic variation of $\bar{\Xi}^i$; see, e.g., \cite[Ch.~II, Theorem~37]{Protter2005}.

  Recall that, for almost every $\omega \in \Omega$, the sample path $\bar{W}(\omega)$ of the Brownian motion satisfies Property \textup{(RIE)} relative to $p$ and $(\cP^n)_{n \in \N}$; see Remark~\ref{remark: BM satisfies RIE}. Hence, for almost every $\omega \in \Omega$, $\Xi^i = \bar{\Xi}^i(\omega)$, for $\Xi^i$ defined as in \eqref{eq: linear RDE} for the (random) controlled paths $b = \bar{b}(\omega)$, $\sigma = \bar{\sigma}(\omega)$ and the (random) rough path $\bW := \bar{\bW}(\omega)$, since the (random) rough integral and the It{\^o} integral coincide almost surely by Lemma~\ref{lemma: rough integral is equal to stochastic integral}. Moreover, since $\Xi^i$ is a controlled path in $\cV^p_W$, it admits a canonical rough path lift $\mathbf{\Xi}^i$.

  Since, for almost every $\omega \in \Omega$, $\Xi^i = \bar{\Xi}^i(\omega)$ satisfies Property \textup{(RIE)} by Theorem~\ref{thm: controlled path satisfies RIE}, it follows from \cite[Proposition~2.18]{Allan2023b} that, for almost every $\omega \in \Omega$, the quadratic variation $[\bar{\Xi}^i](\omega)$ and the rough path bracket $[\mathbf{\Xi}^i]$ coincide.

  Finally, by Lemma~\ref{lem: dynamics of rough exponential}, for almost every $\omega \in \Omega$, the stochastic exponential $\bar{S}^i(\omega)$ of $\bar{\Xi}^i$ and the (random) rough exponential $S = \bar{S}(\omega)$ of $\Xi^i = \bar{\Xi}^i(\omega)$, which solves the linear rough differential equation \eqref{eq: linear RDE}, coincide. Thus, $(S_t)_{t \in [0,T]}$ can be understood as a fixed realization of the probabilistic model $(\bar{S}_t)_{t \in [0,T]}$.
\end{remark}

In the present pathwise setting, it turns out to be more convenient to equivalently reformulate the RDEs \eqref{eq: pathwise Black Scholes model} into the RDEs
\begin{equation}\label{eq: Black Scholes model via RDE}
S^i_t = s_0^i + \int_0^t (S^i_s b^i_s, S^i_s \sigma^{i\cdot}_s) \dd (\cdot,\bW)_s, \qquad t \in [0,T], \qquad i = 1, \ldots, m,
\end{equation}
where $(\cdot, \bW)$ denotes the time-extended rough path of $\bW$, i.e., the path-level of $(\cdot, \bW)$ is given by $(t,W_t)_{t \in [0,T]}$, and the missing integrals $\int \bar{W}^j_t(\omega) \dd t$, $\int t \dd \bar{W}^j_t(\omega)$, $j = 1, \ldots, d$, are canonically defined as Riemann--Stieltjes integrals. Using Lemma~\ref{lem: consistency of rough integrals}, Proposition~\ref{prop: associativity of rough integration} and Lemma~\ref{lem: dynamics of rough exponential}, there exists a unique solution $(S^i,(S^i)') \in \cV^p_{(\cdot,W)}$ to the above RDE, where $(S^i)' = (S^i b^i, S^i \sigma^{i\cdot})$, $i = 1, \ldots, m$. In particular, $S$, $b$ and $\sigma$ are all controlled paths with respect to $(\cdot,W)$. Moreover, $(S_t)_{t \in [0,T]}$ satisfies the RDE \eqref{eq: Black Scholes model via RDE} if and only if it satisfies the RDE \eqref{eq: pathwise Black Scholes model}. For later reference, let us also remark (again) that $(\cdot,W)$ satisfies Property \textup{(RIE)} relative to $p$ and $(\cP^n)_{n \in \N}$.

\subsection{Pathwise construction of log-optimal portfolios}

As a first step towards a pathwise analysis of optimal portfolios, we give a pathwise construction of the log-optimal portfolio provided that the underlying price dynamics of the financial market are given by a Black--Scholes-type model. We recall that in the probabilistic setting the log-optimal portfolio is well-known and was presented in Theorem~\ref{thm: Goll-Kallsen's main theorem Ito processes}, which is due to \cite{Goll2000}.

We recall that the time-extension of the path $W$ is denoted by $(\cdot,W) = (t,W)_{t \in [0,T]}$.

\begin{lemma}\label{lem: pathwise optimal portfolio for Black Scholes}
Let $b$ and $\sigma$ be as in Assumption~\ref{assumption: b and sigma in BS model}. Let
\begin{equation*}
H_t^i := H^{(b,\sigma),i}_t := \frac{h_t^i}{S_t^i} \qquad \text{with} \qquad h_t := h^{(b,\sigma)}_t := (\sigma_t \sigma^\top_t)^{-1} b_t, \qquad t \in [0,T],
\end{equation*}
and set $(\phi, \kappa) := (\phi^{(b,\sigma)}, \kappa^{(b,\sigma)}) := (\phi^{(b,\sigma),0}, \dots, \phi^{(b,\sigma),m}, \kappa^{(b,\sigma)})$ with
\begin{align*}
Z_t &:= Z^{(b,\sigma)}_t := \int_0^t h_s^\top b_s \dd s + \int_0^t h_s^\top \sigma_s \dd \bW_s,\\
\kappa_t &:= \kappa_t^{(b,\sigma)} := \frac{1}{K_T} \cE(Z)_t, \qquad V_t := V^{(b,\sigma)}_t := \kappa_t (K_T - K_t),\\
\phi_t^i &:= \phi^{(b,\sigma),i}_t := H_t^i V_t, \qquad i = 1, \ldots, m, \qquad \phi^0_t := \phi^{(b,\sigma),0}_t := \sum_{i=1}^m \int_0^t \phi^i_s \dd S^i_s - \phi_t^i S_t^i,
\end{align*}
for $t \in [0,T]$, where $\int_0^t \phi_s^i \dd S_s^i$ is a rough integral, and $\cE$ is the rough exponential as defined in Lemma~\ref{lem: dynamics of rough exponential}. Then $(\phi, \kappa)$ is well-defined and is a controlled path with respect to $W$ and, in particular, with respect to $(\cdot,W)$.
\end{lemma}

\begin{proof}
We have that $b$ and $\sigma$ are controlled paths in $\cV^p_W$, and that $\det(\sigma \sigma^\top)$ is bounded away from zero. Recalling that sums and products of (real-valued) controlled paths are again controlled paths (Lemma~\ref{lem: product of controlled paths}), as well as the inverse of a controlled path that is bounded away from zero (as a composition with the regular function $x \mapsto \frac{1}{x}$), we have that $\sigma \sigma^\top$, $h$, $h^\top b$ and $h^\top \sigma$ are all controlled paths in $\cV^p_W$.

Similarly, since each component of $S$ is bounded away from zero (due to being a positive c\`adl\`ag path on a compact time interval), we have that $H$ is a controlled path in $\cV^p_W$.

It is clear that $Z$ is a controlled path, and Lemma~\ref{lem: consistency of rough integrals} and Lemma~\ref{lem: dynamics of rough exponential} then imply that the rough exponential $\cE(Z)$ and $\kappa$ are controlled paths. Since the consumption clock $K$ is a c{\`a}dl{\`a}g (deterministic) and increasing function (so of finite $1$-variation), by Lemma~\ref{lem: product of controlled paths}, $V$ is also a controlled path. It follows that $\phi^0, \phi^i$, $i = 1, \ldots, m$, and thus $(\phi, \kappa)$ are controlled paths with respect to $W$, and hence also with respect to $(\cdot,W)$.
\end{proof}

The portfolio, as constructed in Lemma~\ref{lem: pathwise optimal portfolio for Black Scholes} in a pathwise manner, agrees, indeed, with the log-optimal portfolio for the investment-consumption problem, as considered in Section~\ref{sec: portfolio in probabilistic setting}, if the underlying frictionless financial market is generated by a Black--Scholes-type model, such as the stochastic differential equation~\eqref{eq: Black Scholes model}. Hence, in the following we shall call the portfolio $(\phi, \kappa)=(\phi^{(b,\sigma)}, \kappa^{(b,\sigma)})$ from Lemma~\ref{lem: pathwise optimal portfolio for Black Scholes} the \emph{pathwise log-optimal portfolio}.

\begin{lemma}
Suppose that the price process $(\bar{S}_t)_{t \in [0,T]}$ is modelled by the SDE \eqref{eq: Black Scholes model} driven by a Brownian motion $\bar{W}$ on a probability space $(\Omega,\mathcal{F},\P)$, where $\bar{b}$ and $\bar{\sigma}$ are predictable processes such that, almost surely, $\bar{b}, \bar{\sigma}$ are controlled paths with respect to $\bar{W}$ which satisfy Assumption~\ref{assumption: b and sigma in BS model}. Then the log-optimal portfolio $(\bar{\phi}, \bar{\kappa})$, as provided in Remark~\ref{remark: Goll Kallsen}, and the pathwise log-optimal portfolio $(\phi, \kappa)$, as provided in Lemma~\ref{lem: pathwise optimal portfolio for Black Scholes}, coincide $\P$-almost surely, where $(\phi, \kappa)$ is constructed given a realization $W := \bar{W}(\omega)$ of the Brownian motion $\bar{W}$, and with $b = \bar{b}(\omega)$ and $\sigma = \bar{\sigma}(\omega)$, for almost every $\omega \in \Omega$.
\end{lemma}

\begin{proof}
In this proof we consider $S$, $(\phi,\kappa)$, etc., as random controlled paths, in the sense that, given a realization $W = \bar{W}(\omega)$ of the Brownian motion, $S$ is defined as the solution to the RDE \eqref{eq: pathwise Black Scholes model} driven by the canonical rough path lift $\bW$ of $W$, and $(\phi, \kappa)$ is defined pathwise via Lemma~\ref{lem: pathwise optimal portfolio for Black Scholes}. We note that
\begin{equation*}
\bar{\kappa}_t = \frac{1}{K_T} \bar{\cE} (\bar{Z})_t, \qquad t \in [0,T],
\end{equation*}
for $\bar{Z} = \int_0^\cdot \bar{h}_t^\top \bar{b}_t \dd t + \int_0^\cdot \bar{h}_t^\top \bar{\sigma}_t \dd \bar{W}_t$, where $\bar{\cE}$ denotes the stochastic exponential. By Lemma~\ref{lemma: rough integral is equal to stochastic integral}, we have that
\begin{equation*}
\int_0^t h_s^\top \sigma_s \dd \bW_s = \int_0^t \bar{h}_s^\top \bar{\sigma}_s \dd \bar{W}_s,
\end{equation*}
almost surely, which implies that $Z = \bar{Z}$ almost surely.

By Lemma~\ref{lem: consistency of rough integrals}, $Z$ admits a canonical rough path lift $\bZ = (Z,\Z) \in \cD^p$, and since $W$ satisfies Property \textup{(RIE)} relative to $p$ and $(\cP^n)_{n \in \N}$, the same is also true of $Z$ by Theorem~\ref{thm: controlled path satisfies RIE}. Analogously to the proof of Lemma~\ref{lem: pathwise portfolio is optimal}, $[\bZ]$ and $[\bar{Z}]$ coincide almost surely. By Lemma~\ref{lem: dynamics of rough exponential}, it then follows that
\begin{equation*}
\cE(Z)_t = \exp \Big(Z_t - \frac{1}{2} [\bZ]_t\Big) = \exp \Big(\bar{Z}_t - \frac{1}{2} [\bar{Z}]_t\Big) = \bar{\cE}(\bar{Z})_t,
\end{equation*}
almost surely, and it is then clear that $\kappa = \bar{\kappa}$ and $V = \bar{V}$ almost surely.

We also note that $H^i = \frac{h^i}{S^i} = \frac{\bar{h}^i}{\bar{S}^i} = \bar{H}^i$ almost surely for each $i = 1, \ldots, m$. Finally, by the associativity of rough (Proposition~\ref{prop: associativity of rough integration}) and stochastic integration, we have that
\begin{align*}
\int_0^t \phi^i_s \dd S^i_s &= \int_0^t \phi^i_s S^i_s b^i_s \dd s + \int_0^t \phi^i_s S^i_s \sigma^{i \cdot}_s \dd \bW_s\\
&= \int_0^t \bar{\phi}^i_s \bar{S}^i_s \bar{b}^i_s \dd s + \int_0^t \bar{\phi}^i_s \bar{S}^i_s \bar{\sigma}^{i \cdot}_s \dd \bar{W}_t = \int_0^t \bar{\phi}^i_s \dd \bar{S}^i_s
\end{align*}
almost surely, which in particular then also implies that $\phi^0 = \bar{\phi}^0$ almost surely. Thus, the log-optimal portfolio $(\bar{\phi}, \bar{\kappa})$, as provided in Remark~\ref{remark: Goll Kallsen}, and the pathwise log-optimal portfolio $(\phi, \kappa)$, as provided in Lemma~\ref{lem: pathwise optimal portfolio for Black Scholes}, coincide almost surely.
\end{proof}

\begin{remark}
We consider $\bar{W}$ to be a Brownian motion to ensure that the pathwise log-optimal portfolio $(\phi,\kappa)$, as constructed in Lemma~\ref{lem: pathwise optimal portfolio for Black Scholes}, is, indeed, a log-optimal portfolio for the investment-consumption problem in the setting of local volatility models. However, again, we emphasize that the construction of the pathwise portfolio $(\phi,\kappa)$, as well as its pathwise analysis as developed in Sections~\ref{subsec: stability in Black Scholes model} and \ref{subsec: discretization in Black Scholes model}, works for any path $W$ satisfying Assumption~\ref{assumption: partitions and W for Black-Scholes model}.
\end{remark}

\subsection{Stability of pathwise log-optimal portfolios with respect to drift and volatility}\label{subsec: stability in Black Scholes model}

Having at hand a pathwise construction of log-optimal portfolios, we are in a position to study its pathwise stability properties. In this subsection, we analyze the change of the log-optimal portfolio and the associated capital processes with respect to the model parameters $b$ and $\sigma$.

The next theorem states that the pathwise log-optimal portfolio $(\phi, \kappa) = (\phi^{(b, \sigma)}, \kappa^{(b, \sigma)})$ and the capital processes $V = V^{(b,\sigma)}$ are locally Lipschitz continuous with respect to the model parameters $b$ and $\sigma$.

\begin{theorem}\label{thm: stability of optimal portfolios wrt model uncertainty for Black Scholes}
Let $b, \sigma$ and $\tb, \tsigma$ satisfy Assumption~\ref{assumption: b and sigma in BS model}. Let $(\phi^{(b,\sigma)}, \kappa^{(b,\sigma)})$ and $(\phi^{(\tb,\tsigma)}, \kappa^{(\tb, \tsigma)})$ be the corresponding pathwise log-optimal portfolios, as constructed in Lemma~\ref{lem: pathwise optimal portfolio for Black Scholes}, and let $M > 0$ be an upper bound for
\begin{equation*}
\|b\|_{\crpW}, \|\tb\|_{\crpW}, \|\sigma\|_{\crpW}, \|\tsigma\|_{\crpW}, \sup_t |\det(\sigma_t \sigma_t^\top)|^{-1}, \sup_t |\det(\tsigma_t \tsigma_t^\top)|^{-1} \text{ and } \|(\cdot,\bW)\|_p.
\end{equation*}
Then
\begin{equation*}
\|(\phi^{(b,\sigma)}, \kappa^{(b,\sigma)}); (\phi^{(\tilde b,\tilde \sigma)}, \kappa^{(\tilde b,\tilde \sigma)})\|_{\crpW} \lesssim \|b; \tb\|_{\crpW} + \|\sigma; \tsigma\|_{\crpW}
\end{equation*}
and
\begin{equation*}
\|V^{(b,\sigma)}; V^{(\tb,\tsigma)}\|_{\crpW} \lesssim \|b; \tb\|_{\crpW} + \|\sigma; \tsigma\|_{\crpW},
\end{equation*}
and in particular that
\begin{equation*}
\|(\phi^{(b,\sigma)}, \kappa^{(b,\sigma)}) - (\phi^{(\tb,\tsigma)}, \kappa^{(\tb,\tsigma)})\|_\infty \lesssim \|b; \tb\|_{\crpW} + \|\sigma; \tsigma\|_{\crpW}
\end{equation*}
and
\begin{equation*}
\|V^{(b,\sigma)} - V^{(\tb,\tsigma)}\|_\infty \lesssim \|b; \tb\|_{\crpW} + \|\sigma; \tsigma\|_{\crpW},
\end{equation*}
where the implicit multiplicative constants depend only on $p$, $m$, $d$, $M$, $s_0$, $\sup_t |S^{(b,\sigma),i}_t|^{-1}$, $\sup_t |S^{(\tb,\tsigma),i}|^{-1}$, $i = 1, \ldots, m$, and the consumption clock $K$.
\end{theorem}

\begin{proof}
\emph{Step 1.}
Let $c^{(b,\sigma)} := \sigma \sigma^\top$ and $c^{(\tb,\tsigma)} = \tsigma \tsigma^\top$. As seen in the proof of Lemma~\ref{lem: pathwise optimal portfolio for Black Scholes}, $c^{(b,\sigma)}$ is a controlled path in $\cV^p_W$, thus in $\crpW$. Lemma~\ref{lem: product of controlled paths} then yields for each component that
\begin{equation}\label{eq: estimate for c for Black Scholes}
\|((c^{(b,\sigma)})^{-1})^{ij}\|_{\crpW} \leq C,
\end{equation}
where $C > 0$ depends only on $p$, $m$, $d$ and $M$, and the same holds for $((c^{(\tb,\tsigma)})^{-1})^{ij}$. By Lemma~\ref{lem: estimate product of controlled paths}, and the fact that the composition of a controlled path that is bounded away from zero with the inverse function $x \mapsto \frac{1}{x}$ is locally Lipschitz continuous, we infer that
\begin{equation}\label{eq: Lipschitz estimate for c for Black Scholes}
\|((c^{(b,\sigma)})^{-1})^{ij}; ((c^{(\tb,\tsigma)})^{-1})^{ij}\|_{\crpW} \lesssim \|b; \tb\|_{\crpW} + \|\sigma; \tsigma\|_{\crpW},
\end{equation}
where the implicit multiplicative constant depends only on $p$, $m$, $d$ and $M$.

Let $\vartheta^{(b,\sigma)} := ((h^{(b,\sigma)})^\top \sigma)^\top = \sigma^\top (c^{(b,\sigma)})^{-1} b$ and $\vartheta^{(\tb,\tsigma)} := ((h^{(\tb,\tsigma)})^\top \tsigma)^\top = \tsigma^\top (c^{(\tb,\tsigma)})^{-1} \tb$. Then, $\theta^{(b,\sigma)} = (\frac{1}{2} (\vartheta^{(b,\sigma)})^\top \vartheta^{(b,\sigma)}, (\vartheta^{(b,\sigma)})^\top)$ are controlled paths in $\crpW$ as, again, the sum and product of controlled paths remains a controlled path. By the same arguments as before, combined with the estimates in \eqref{eq: estimate for c for Black Scholes} and \eqref{eq: Lipschitz estimate for c for Black Scholes}, we obtain
\begin{equation}\label{eq: estimate for theta for Black Scholes}
\|\theta^{(b,\sigma)}\|_{\crpW} \leq C,
\end{equation}
where $C > 0$ depends only on $p$, $m$, $d$ and $M$ (and the same estimate for $\theta^{(\tb,\tsigma)}$), and
\begin{equation}\label{eq: Lipschitz estimate for theta for Black Scholes}
\|\theta^{(b,\sigma)}; \theta^{(\tb,\tsigma)}\|_{\crpW} \lesssim \|b; \tb\|_{\crpW} + \|\sigma; \tsigma\|_{\crpW},
\end{equation}
where the implicit multiplicative constant also depends only on $p$, $m$, $d$ and $M$.

\emph{Step 2.}
Define the rough integrals $U^{(b,\sigma)} := \int_0^\cdot \theta_t^{(b,\sigma)} \dd (\cdot,\bW)_t$ and $U^{(\tb,\tsigma)} := \int_0^\cdot \theta_t^{(\tb,\tsigma)} \dd (\cdot,\bW)_t$, which are controlled paths in $\crpW$. Using the estimate in \eqref{eq: estimate for rough integral} for the rough integral, and the bound in \eqref{eq: estimate for theta for Black Scholes}, we have that
\begin{equation}\label{eq: estimate for U for Black Scholes}
\|U^{(b,\sigma)}\|_{\crpW} \leq C,
\end{equation}
where $C > 0$ depends only on $p$, $m$, $d$ and $M$, and the same holds for $U^{(\tb,\tsigma)}$. In particular,
\begin{equation}\label{eq: sup-norm estimate for U}
\|U^{(b,\sigma)}\|_\infty \leq C_0,
\end{equation}
where $C_0 > 0$ depends only on $p$, $m$, $d$ and $M$. Further, using the stability of rough integrals (e.g., \cite[Lemma~3.4]{Friz2018}), and the estimate in \eqref{eq: Lipschitz estimate for theta for Black Scholes}, it follows immediately that
\begin{equation}\label{eq: Lipschitz estimate for U for Black Scholes}
\|U^{(b,\sigma)}; U^{(\tb,\tsigma)}\|_{\crpW} \lesssim \|b; \tb\|_{\crpW} + \|\sigma; \tsigma\|_{\crpW},
\end{equation}
where the implicit multiplicative constant depends only on $p$, $m$, $d$ and $M$.

\emph{Step 3.}
We have from Lemma~\ref{lem: dynamics of rough exponential} that
\begin{align*}
\cE(Z^{(b,\sigma)})_t &= \exp \Big(Z_t^{(b,\sigma)} - \frac{1}{2}[\mathbf{Z}^{(b,\sigma)}]_t\Big)\\
&= \exp \bigg( \frac{1}{2} \int_0^t (\vartheta^{(b,\sigma)}_s)^\top \vartheta^{(b,\sigma)}_s \dd s + \int_0^t (\vartheta^{(b,\sigma)}_s)^\top \dd \bW_s \bigg) = \exp(U^{(b,\sigma)}_t).
\end{align*}
By Lemma~\ref{lem: pathwise optimal portfolio for Black Scholes}, we have that
\begin{equation*}
\kappa^{(b,\sigma)}_t := \frac{1}{K_T} \cE(Z^{(b,\sigma)})_t, \qquad \kappa^{(\tb,\tsigma)}_t := \frac{1}{K_T} \cE(Z^{(\tb,\tsigma)})_t,
\end{equation*}
for $t \in [0,T]$, are the pathwise optimal consumption rates for the log-utility on the financial market modelled by $S^{(b,\sigma),i}$ and $S^{(\tb,\tsigma),i}$, $i = 1, \ldots, m$, respectively, and are controlled paths in $\crpW$. We also infer from the bounds in \eqref{eq: estimate for U for Black Scholes} and \eqref{eq: sup-norm estimate for U} that
\begin{equation*}
\|\kappa^{(b,\sigma)}\|_{\crpW} \leq C,
\end{equation*}
where $C > 0$ depends only on $p$, $m$, $d$, $M$, $\|\exp\|_{C^2_b(\{y: |y| \leq C_0\}; \R)}$ and $K$, as it is a composition of a controlled path with a regular function, and the same holds for $\kappa^{(\tb,\tsigma)}$.

Because the composition of a controlled path with a regular function is locally Lipschitz continuous (see, e.g., \cite[Lemma~3.5]{Friz2018}), it follows from \eqref{eq: Lipschitz estimate for U for Black Scholes} that
\begin{equation}\label{eq: Lipschitz estimate for kappa for Black Scholes}
\|\kappa^{(b,\sigma)}; \kappa^{(\tb,\tsigma)}\|_{\crpW} \lesssim \|b; \tb\|_{\crpW} + \|\sigma; \tsigma\|_{\crpW},
\end{equation}
where the implicit multiplicative constant depends only on $p$, $m$, $d$, $M$ and $K$.

\emph{Step 4.}
Since $(K_t)_{t \in [0,T]}$ is a c{\`a}dl{\`a}g (deterministic) and increasing function (so of finite $1$-variation), we recall that, by Lemma~\ref{lem: product of controlled paths}, the wealth process $V^{(b,\sigma)}_t = \kappa^{(b,\sigma)}_t (K_T - K_t)$, $t \in [0,T]$, (as the product of two controlled paths) is itself a controlled path in $\crpW$, and we can bound
\begin{equation}\label{eq: estimate for V for Black Scholes}
\|V^{(b,\sigma)}\|_{\crpW} \leq C,
\end{equation}
where $C > 0$ depends only on $p$, $m$, $d$, $M$ and $K$. Lemma~\ref{lem: estimate product of controlled paths} then implies that
\begin{equation*}
\|V^{(b,\sigma)}; V^{(\tb,\tsigma)}\|_{\crpW} \lesssim \|\kappa^{(b,\sigma)}; \kappa^{(\tb,\tsigma)}\|_{\crpW},
\end{equation*}
and, combining this with \eqref{eq: Lipschitz estimate for kappa for Black Scholes}, we have that
\begin{equation}\label{eq: Lipschitz estimate for V for Black Scholes}
\|V^{(b,\sigma)}; V^{(\tb,\tsigma)}\|_{\crpW} \lesssim \|b; \tb\|_{\crpW} + \|\sigma; \tsigma\|_{\crpW},
\end{equation}
where the implicit multiplicative constant depends only on $p$, $m$, $d$, $M$ and $K$.

\emph{Step 5.}
Define the rough integrals $A^{(b,\sigma),i} := \int_0^\cdot (b^i_t - \frac{1}{2} \sigma^{i \cdot}_t (\sigma^{i \cdot}_t)^\top, \sigma^{i, \cdot}_t) \dd (\cdot,\bW)_t$ and $A^{(\tb,\tsigma),i} := \int_0^\cdot (\tb^i_t - \frac{1}{2} \tsigma^{i \cdot}_t (\tsigma^{i \cdot}_t)^\top, \tsigma^{i, \cdot}_t) \dd (\cdot,\bW)_t$, $i = 1, \ldots, m$. By the bound \eqref{eq: estimate for rough integral}, we have that
\begin{equation}\label{eq: estimate for A for Black Scholes}
\|A^{(b,\sigma),i}\|_{\crpW} \leq C,
\end{equation}
where $C > 0$ depends only on $p$, $d$ and $M$, and the same holds for $A^{(\tb,\tsigma),i}$. In particular,
\begin{equation}\label{eq: sup-norm estimate for A for Black Scholes}
\|A^{(b,\sigma),i}\|_\infty \leq C_1,
\end{equation}
where $C_1 > 0$ depends only on $p$, $d$ and $M$. Further, by the stability of rough integration (e.g., \cite[Lemma~3.4]{Friz2018}) and Lemma~\ref{lem: estimate product of controlled paths}, it follows that
\begin{equation}\label{eq: Lipschitz estimate for A for Black Scholes}
\|A^{(b,\sigma),i}; A^{(\tb,\tsigma),i}\|_{\crpW} \lesssim \|b; \tb\|_{\crpW} + \|\sigma; \tsigma\|_{\crpW},
\end{equation}
where the implicit multiplicative constant depends only on $p$, $d$ and $M$.

\emph{Step 6.}
At the beginning of this section, we noted that the solution $S = (S_t)_{t \in [0,T]}$ to the linear rough differential equation \eqref{eq: pathwise Black Scholes model} is given by a rough exponential. More precisely, for $i = 1, \ldots, m$, we have that
\begin{equation}\label{eq: represent S in terms of exp(A)}
S_t^{(b,\sigma)i} = s_0^i \exp \bigg( \int_0^t \Big(b^i_s - \frac{1}{2} \sigma^{i \cdot}_s (\sigma^{i \cdot}_s)^\top\Big) \dd s + \int_0^t \sigma^{i \cdot}_s \dd \bW_s \bigg) = s_0^i \exp(A^{(b,\sigma),i}_t).
\end{equation}
We therefore have from \eqref{eq: estimate for A for Black Scholes} and \eqref{eq: sup-norm estimate for A for Black Scholes} that
\begin{equation}\label{eq: estimate for S for Black Scholes}
\|S^{(b,\sigma),i}\|_{\crpW} \leq C,
\end{equation}
where $C > 0$ depends only on $p$, $d$, $M$, $s_0$ and $\|\exp\|_{C^2_b(\{y: |y|\leq C_1\}; \R)}$, as it is the composition of a controlled path with a regular function, and the same holds for $S^{(\tb,\tsigma),i}$.

Then, by the stability of regular functions of controlled paths (see, e.g., \cite[Lemma~3.5]{Friz2018}), combined with the estimate in \eqref{eq: Lipschitz estimate for A for Black Scholes}, we have that
\begin{equation}\label{eq: Lipschitz estimate for S for Black Scholes}
\|S^{(b,\sigma),i}; S^{(\tb, \tsigma),i}\|_{\crpW} \lesssim \|b; \tb\|_{\crpW} + \|\sigma; \tsigma\|_{\crpW},
\end{equation}
where the implicit multiplicative constant depends only on $p$, $d$, $M$ and $s_0$.

\emph{Step 7.}
For $h^{(b,\sigma)} = (c^{(b,\sigma)})^{-1} b$ and $h^{(\tb,\tsigma)} = (c^{(\tb,\tsigma)})^{-1} \tb$, we have from Lemma~\ref{lem: product of controlled paths} and the bound in \eqref{eq: estimate for c for Black Scholes} that
\begin{equation}\label{eq: estimate for gamma for Black Scholes}
\|h^{(b,\sigma),i}\|_{\crpW} \leq C,
\end{equation}
where $C > 0$ depends only on $p$, $m$, $d$ and $M$, and the same holds for $h^{(\tb,\tsigma),i}$, and then, by \eqref{eq: Lipschitz estimate for c for Black Scholes} and Lemma~\ref{lem: estimate product of controlled paths}, it follows that
\begin{equation}\label{eq: Lipschitz estimate for gamma for Black Scholes}
\|h^{(b,\sigma),i}; h^{(\tb, \tsigma),i}\|_{\crpW} \lesssim \|b;\tb\|_{\crpW} + \|\sigma;\tsigma\|_{\crpW},
\end{equation}
where the implicit multiplicative constant depends only on $p$, $m$, $d$ and $M$. We now consider
\begin{equation*}
H_t^{(b,\sigma),i} = \frac{h^{(b,\sigma),i}_t}{S_t^{(b,\sigma),i}}.
\end{equation*}
Using Lemma~\ref{lem: product of controlled paths}, Lemma~\ref{lem: estimate product of controlled paths} and the estimates in \eqref{eq: estimate for S for Black Scholes}, \eqref{eq: Lipschitz estimate for S for Black Scholes}, \eqref{eq: estimate for gamma for Black Scholes} and \eqref{eq: Lipschitz estimate for gamma for Black Scholes}, we obtain
\begin{equation}\label{eq: estimate for H for Black Scholes}
\|H^{(b,\sigma),i}\|_{\crpW} \leq C,
\end{equation}
and the same bound for $H^{(\tb,\tsigma),i}$, and
\begin{equation}\label{eq: Lipschitz estimate for H for Black Scholes}
\|H^{(b,\sigma),i}; H^{(\tb, \tsigma),i}\|_{\crpW} \lesssim \|b; \tb\|_{\crpW} + \|\sigma; \tsigma\|_{\crpW},
\end{equation}
where the constant $C$ and the implicit multiplicative constant above depend only on $p$, $m$, $d$, $M$, $s_0$, $\sup_t |S^{(b,\sigma),i}_t|^{-1}$ and $\sup_t |S^{(\tb,\tsigma),i}_t|^{-1}$.

\emph{Step 8.}
By Lemma~\ref{lem: product of controlled paths}, $\phi^{(b,\sigma),i}_t = H^{(b,\sigma),i}_t V^{(b,\sigma)}_t$, $\phi^{(\tb, \tsigma),i}_t = H^{(\tb, \tsigma),i}_t V^{(\tb, \tsigma)}_t$, $i = 1, \ldots, m$, are controlled paths, and, by the bounds in \eqref{eq: estimate for H for Black Scholes} and \eqref{eq: estimate for V for Black Scholes},
\begin{equation}\label{eq: estimate for phi_i for Black Scholes}
\|\phi^{(b,\sigma),i}\|_{\crpW} \leq C,
\end{equation}
where $C > 0$ depends only on $p$, $m$, $d$, $M$, $s_0$, $\sup_t |S^{(b,\sigma),i}_t|^{-1}$ and $K$, and the same holds for $\phi^{(\tb, \tsigma),i}$. By Lemma~\ref{lem: estimate product of controlled paths}, we have that
\begin{equation*}
\|\phi^{(b,\sigma),i}; \phi^{(\tb,\tsigma),i}\|_{\crpW} \lesssim \|H^{(b,\sigma),i}; H^{(\tb,\tsigma),i}\|_{\crpW} + \|V^{(b,\sigma)}; V^{(\tb,\tsigma)}\|_{\crpW},
\end{equation*}
which, together with the estimates in \eqref{eq: estimate for V for Black Scholes}, \eqref{eq: Lipschitz estimate for V for Black Scholes}, \eqref{eq: estimate for H for Black Scholes} and \eqref{eq: Lipschitz estimate for H for Black Scholes}, gives
\begin{equation}\label{eq: Lipschitz estimate for phi_i for Black Scholes}
\|\phi^{(b,\sigma),i}; \phi^{(\tb,\tsigma),i}\|_{\crpW} \lesssim \|b; \tb\|_{\crpW} + \|\sigma; \tsigma\|_{\crpW},
\end{equation}
where the implicit multiplicative constant depends only on $p$, $m$, $d$, $M$, $s_0$, $\sup_t |S^{(b,\sigma),i}_t|^{-1}$, $\sup_t |S^{(\tb,\tsigma),i}_t|^{-1}$ and $K$. Finally, we consider
\begin{equation*}
\phi^{(b,\sigma),0}_t = \sum_{i=1}^m \int_0^t \phi^{(b,\sigma),i}_s \dd S^{(b,\sigma),i}_s - \phi^{(b,\sigma),i}_t S^{(b,\sigma),i}_t
\end{equation*}
and
\begin{equation*}
\phi^{(\tb, \tsigma),0}_t = \sum_{i=1}^m \int_0^t \phi^{(\tb,\tsigma),i}_s \dd S^{(\tb, \tsigma),i}_s - \phi^{(\tb,\tsigma),i}_t S^{(\tb, \tsigma),i}_t
\end{equation*}
for $t \in [0,T]$. By the associativity of rough integrals, we have that
\begin{equation*}
\int_0^\cdot \phi^{(b,\sigma),i}_t \dd S^{(b,\sigma),i}_t = \int_0^\cdot \psi^{(b,\sigma),i}_t \dd (\cdot, \bW)_t,
\end{equation*}
where $\psi^{(b,\sigma),i}_t := (\phi^{(b,\sigma),i}_t S^i_t b^i_t, \phi^{(b,\sigma),i}_t S^i_t \sigma^{i\cdot}_t)$, and similarly for $\int_0^{\cdot} \phi^{(\tb,\tsigma),i}_t \dd S^{(\tb,\tsigma),i}_t$. By the estimates in \eqref{eq: estimate for S for Black Scholes}, \eqref{eq: Lipschitz estimate for S for Black Scholes}, \eqref{eq: estimate for phi_i for Black Scholes} and \eqref{eq: Lipschitz estimate for phi_i for Black Scholes}, we deduce that
\begin{equation*}
\|\psi^{(b,\sigma)}; \psi^{(\tb,\tsigma)}\|_{\crpW} \lesssim \|b; \tb\|_{\crpW} + \|\sigma; \tsigma\|_{\crpW},
\end{equation*}
where the implicit multiplicative constant depends only on $p$, $m$, $d$, $M$, $\|(\cdot,\bW)\|_p$, $s_0$, $K$ and on $\sup_t |S^{(b,\sigma),i}_t|^{-1}$ and $\sup_t |S^{(\tb,\tsigma),i}_t|^{-1}$ for $i = 1, \ldots, m$.

Combining this with the stability of rough integration, Lemma~\ref{lem: estimate product of controlled paths}, and the estimates in \eqref{eq: Lipschitz estimate for S for Black Scholes} and \eqref{eq: Lipschitz estimate for phi_i for Black Scholes}, we obtain
\begin{equation*}
\|\phi^{(b,\sigma),0}; \phi^{(\tb,\tsigma),0}\|_{\crpW} \lesssim \|b; \tb\|_{\crpW} + \|\sigma; \tsigma\|_{\crpW},
\end{equation*}
where the implicit multiplicative constant depends only on $p$, $m$, $d$, $M$, $\|(\cdot,\bW)\|_p$, $s_0$, $K$ and on $\sup_t |S^{(b,\sigma),i}_t|^{-1}$ and $\sup_t |S^{(\tb,\tsigma),i}_t|^{-1}$ for $i = 1, \ldots, m$.

Hence, since we can bound the supremum norm by the controlled path norm, we have established the (local) Lipschitz continuity for optimal portfolios and wealth processes.
\end{proof}

\begin{remark}
As noted in Section~\ref{subsec: stability in local vola model} above, if one is interested in the stability of optimal portfolios with respect to model uncertainty, then it is more natural to consider a common price path $S^{(b,\sigma)} = S^{(\tb,\tsigma)} =: S$, which we interpret as the ``true'' price path, and to then consider the wealth process corresponding to the log-optimal portfolio under (possibly misspecified) parameters $b$ and $\sigma$. Indeed, in the current setting, a result analogous to that in Corollary~\ref{corollary: model uncertainty for local vol model} follows as an easy corollary of Theorem~\ref{thm: stability of optimal portfolios wrt model uncertainty for Black Scholes}.
\end{remark}

\subsection{Discretization error of pathwise log-optimal portfolios}\label{subsec: discretization in Black Scholes model}

We now introduce a discretized version of the pathwise log-optimal portfolio, as constructed in Lemma~\ref{lem: pathwise optimal portfolio for Black Scholes}, and derive quantitative pathwise error estimates resulting from this discretization for the portfolios, as well as for their associated capital processes.

\smallskip

To define the discretized version of the pathwise log-optimal portfolio, we start by discretizing the underlying price paths. To this end, we recall that $W$ and the sequence of partitions $(\cP^n)_{n \in \N}$ satisfy Assumption~\ref{assumption: partitions and W for Black-Scholes model}, where $\cP^n = \{0 = t_0^n < t_1^n < \cdots < t_{N_n}^n = T\}$. For $n \in \N$, we let $W^n \colon [0,T] \to \R^d$ be the piecewise constant approximation of $W$ along $\cP^n$, as defined in \eqref{eq: defn W^n}. We define $b^n$ and $\sigma^n$ in the same way, and let $\gamma^n \colon [0,T] \to \R$ be the time-discretization path along $\cP^n$, as defined in \eqref{eq: defn gamma^n}.

To discretize the price path $S$, we discretize its rough exponential representation along $\cP^n$. That is, recalling \eqref{eq: represent S in terms of exp(A)}, we let
\begin{equation*}
S^{n,i}_t := s_0^i \exp(A^{n,i}_t) \qquad \text{with} \qquad A^{n,i}_t := \int_0^t \Big(b^{n,i}_s - \frac{1}{2} \sigma^{n,i\cdot}_s (\sigma^{n,i\cdot}_s)^\top\Big) \d \gamma^n_s + \int_0^t \sigma^{n,i\cdot}_s \dd W^n_s,
\end{equation*}
for $t \in [0,T]$, $i = 1, \ldots, m$. The discretized version $(\phi^n, \kappa^n)$ of the pathwise log-optimal portfolio is then given by
\begin{align*}
H^{n,i}_t &:= \frac{h_t^{n,i}}{S^{n,i}_t} \qquad \text{with} \qquad h^n_t := (\sigma^n_t (\sigma^n_t)^\top)^{-1} b^n_t,\\
\kappa^n_t &:= \frac{1}{K_T} \exp \bigg( \frac{1}{2} \int_0^t (h^n_s)^\top b^n_s \dd \gamma^n_s + \int_0^t (h^n_s)^\top \sigma^n_s \dd W^n_s \bigg), \qquad V^n_t := \kappa_t^n (K_T - K^n_t),\\
\phi^{n,i}_t &:= H^{n,i}_t V^n_t, \quad i = 1, \ldots, m, \qquad \phi^{n,0}_t := \sum_{i=1}^m \int_0^t \phi^{n,i}_s \dd S^{n,i}_s - \phi^{n,i}_t S^{n,i}_t,\\
\hat{V}^n_t &:= 1 + \int_0^t (\phi^n_s)^\top \dd S_s - \int_0^t \kappa^n_s \dd K_s, \qquad t \in [0,T],
\end{align*}
where $K^n$ denotes the piecewise constant approximation of $K$ along $\cP^n$. Similarly to Section~\ref{subsec: discretization in local vola model}, we define the wealth process $\hat{V}^n$ by considering the performance of the discretized portfolio $(\phi^n, \kappa^n)$ relative to the continuous (i.e., non-discretized) price process $S$ and consumption clock $K$. We can define the wealth process corresponding to the log-optimal portfolio on the continuous market by
\begin{equation*}
\hat{V}_t := 1 + \int_0^t \phi_s^\top \dd S_s - \int_0^t \kappa_s \dd K_s, \qquad t \in [0,T],
\end{equation*}
which, of course, coincides with the process $V$ in Lemma~\ref{lem: pathwise optimal portfolio for Black Scholes}.

For these discretized portfolios and their associated capital processes, we obtain the following convergence result with quantitative error estimates.

\begin{theorem}\label{thm: stability of optimal portfolios wrt discretization error for Black Scholes}
Let $b$ and $\sigma$ be as in Assumption~\ref{assumption: b and sigma in BS model}. Let $(\phi, \kappa)$ be the pathwise log-optimal portfolio as constructed in Lemma~\ref{lem: pathwise optimal portfolio for Black Scholes}, let $(\phi^n, \kappa^n)$ be the discretized portfolio as defined above, and let $\hat{V}$ and $\hat{V}^n$ be the corresponding wealth processes. Then,
\begin{equation*}
\|(\phi^n, \kappa^n) - (\phi,\kappa)\|_{p'} \, \longrightarrow \, 0 \qquad \text{as} \qquad n \, \longrightarrow \, \infty
\end{equation*}
and
\begin{equation*}
\|\hat{V}^n - \hat{V}\|_{p'} \, \longrightarrow \, 0 \qquad \text{as} \qquad n \, \longrightarrow \, \infty,
\end{equation*}
for any $p' \in (p,3)$, with a rate of convergence given by
\begin{align*}
&\|(\phi^n,\kappa^n) - (\phi,\kappa)\|_{p'} + \|\hat{V}^n - \hat{V}\|_{p'}\\
&\lesssim \Big(\|(b^n)' - b'\|_\infty + \|R^{b^n} - R^b\|_\infty + \|(\sigma^n)' - \sigma'\|_\infty + \|R^{\sigma^n} - R^\sigma\|_\infty\Big)^{1-\frac{p}{p'}}\\
&\quad + |\cP^n|^{(1-\frac{1}{q}) (1-\frac{p}{p'})} + \|W^n - W\|_\infty^{1-\frac{p}{p'}} + \bigg\| \int_0^\cdot W^n_t \otimes \dd W_t - \int_0^\cdot W_t \otimes \dd W_t \bigg\|_\infty^{1-\frac{p}{p'}},
\end{align*}
for any $q \in (1,2)$ such that $\frac{1}{p'} + \frac{1}{q} > 1$, where the implicit multiplicative constant depends only on $p$, $p'$, $q$, $m$, $d$, $\|b\|_{\cV^{p'}_{(\cdot,W)}}$, $\|\sigma\|_{\cV^{p'}_{(\cdot,W)}}$, $\sup_t |\det(\sigma_t \sigma_t^\top)|^{-1}$, $K$ and $w(0,T)$, where $w$ is the control function for which \eqref{eq: RIE inequality} holds for $(\cdot,W)$.
\end{theorem}

Of course, the convergence and quantitative estimates in Theorem~\ref{thm: stability of optimal portfolios wrt discretization error for Black Scholes} also hold when the $p'$-variation norm $\|\cdot\|_{p'}$ is replaced by the supremum norm $\|\cdot\|_\infty$.

Before we present the proof, some preliminary steps are necessary. We start by recalling that, since $W^n$ has finite $1$-variation, $W^n$ admits a canonical rough path lift $\bW^n = (W^n,\W^n) \in \cD^p([0,T];\R^d)$, with $\W^n$ given by
\begin{equation*}
\W^n_{s,t} := \int_s^t W^n_{s,u} \otimes \d W^n_u, \qquad (s,t) \in \Delta_T,
\end{equation*}
where the integral is defined as a classical limit of left-point Riemann sums. Similarly, we define the time-extended rough path $(\cdot,\bW)^n$ above $(\cdot,W)^n := (\gamma^n,W^n)$.

Since $b^n$ is the piecewise constant approximation of $b$ along $\cP^n$, it is a controlled path with respect to $(\cdot,W)^n$. If $t^n_k \leq s \leq t < t^n_{k+1}$ for some $k$, then $b^n_{s,t} = b_{t^n_k,t^n_k} = 0$. Otherwise, let $k_0$ and $k_1$ be such that $s \in [t^n_{k_0},t^n_{k_0+1})$ and $t \in [t^n_{k_1},t^n_{k_1+1})$. Then
\begin{equation*}
  b^n_{s,t} = b_{t^n_{k_0},t^n_{k_1}} = b'_{t^n_{k_0}} (\cdot,W)_{t^n_{k_0},t^n_{k_1}} + R^b_{t^n_{k_0},t^n_{k_1}} = (b')^n_s (\cdot,W)^n_{s,t} + (R^b)^n_{s,t},
\end{equation*}
where $(b')^n$ and $(R^b)^n$ be the piecewise constant approximations of $b'$ and $R^b$ along $\cP^n$, respectively. It follows that $(b^n)' = (b')^n$ and $R^{b^n} = (R^b)^n$, so that $\sup_{n \in \N} \|b^n\|_{\cV^p_{(\cdot,W)^n}} \leq \|b\|_{\crpW}$, and similarly for $\sigma$.

Proceeding as in the proof of Lemma~\ref{lem: pathwise optimal portfolio for Black Scholes}, it is straightforward to see that $H^n$, $\kappa^n$, $V^n$ and $\phi^n$ are controlled paths in $\cV^p_{(\cdot,W)^n}$, and that $\int_0^t \phi^{n,i}_s \dd S^{n,i}_s$ is a rough integral in the sense of Lemma~\ref{lem: integration against controlled path}.

Similarly to in Section~\ref{subsec: discretization in local vola model}, one can show that the rough integral of a controlled path in $\cV^p_{(\cdot,W)^n}$ against $(\cdot,\bW)^n$ coincides with the corresponding Riemann--Stieltjes integral against $(\cdot,W)^n$. In particular, for every $t \in [0,T]$ and $i = 1, \ldots, m$, we have that
\begin{equation*}
A^{n,i}_t = \int_0^t \Big(b^{n,i}_s - \frac{1}{2} \sigma^{n,i\cdot}_s (\sigma^{n,i\cdot}_s)^\top, \sigma^{n,i\cdot}_s\Big) \d (\cdot,\bW)^n_s,
\end{equation*}
and that
\begin{equation*}
\kappa^n_t = \frac{1}{K_T} \exp \bigg( \int_0^t \Big(\frac{1}{2} (h^n_s)^\top b^n_s, (h^n_s)^\top \sigma^n_s\Big) \d (\cdot,\bW)^n_s \bigg).
\end{equation*}
Similarly, by Theorem~\ref{thm: rough int against controlled path under RIE}, the rough integral of $\phi^{n,i}$ against $S^{n,i}$ is simply given by a left-point Riemann sum.

\begin{proof}[Proof of Theorem~\ref{thm: stability of optimal portfolios wrt discretization error for Black Scholes}]
Let $p' \in (p,3)$. It follows by interpolation (see, e.g., \cite[Proposition~5.5]{Friz2010}) that
\begin{equation*}
\|b^n;b\|_{\cV^{p'}_{(\cdot,W)^n},\cV^{p'}_{(\cdot,W)}} = \|(b^n)' - b'\|_{p'} + \|R^{b^n} - R^b\|_{\frac{p'}{2}} \lesssim \|(b^n)' - b'\|_\infty^{1-\frac{p}{p'}} + \|R^{b^n} - R^b\|_\infty^{1-\frac{p}{p'}},
\end{equation*}
where the implicit multiplicative constant depends only on $p$, $p'$ and $\|b\|_{\crpW}$, and similarly for $\sigma$. Since $b^n$ converges uniformly to $b$ as $n \to \infty$ (by the assumption that the jumps of $b$ belong to $\liminf_{n \to \infty} \cP^n$), it follows that
\begin{equation*}
\|b^n;b\|_{\cV^{p'}_{(\cdot,W)^n},\cV^{p'}_{(\cdot,W)}} \, \longrightarrow \, 0 \qquad \text{as} \qquad n \, \longrightarrow \, \infty,
\end{equation*}
and similarly for $\sigma$.

Similarly to the arguments in the proof of Theorem~\ref{thm: stability of optimal portfolios wrt model uncertainty for Black Scholes}, using the estimate in Lemma~\ref{lem: rate of convergence of space-time rough paths}, one can show that
\begin{equation*}
\|\kappa^n\|_{\cV^{p'}_{(\cdot,W)^n}} + \|\kappa\|_{\cV^{p'}_{(\cdot,W)}} \leq C
\end{equation*}
where $C > 0$ depends only on $p'$, $m$, $d$, $\|b\|_{\crpW}$, $\|\sigma\|_{\crpW}$, $\sup_t |\det(\sigma_t \sigma_t^\top)|^{-1}$, $K$ and $w(0,T)$, where $w$ is the control function for which \eqref{eq: RIE inequality} holds for $(\cdot,W)$, and that
\begin{equation*}
\|\kappa^n;\kappa\|_{\cV^{p'}_{(\cdot,W)^n}, \cV^{p'}_{(\cdot,W)}} \lesssim \|b^n;b\|_{\cV^{p'}_{(\cdot,W)^n}, \cV^{p'}_{(\cdot,W)}} + \|\sigma^n;\sigma\|_{\cV^{p'}_{(\cdot,W)^n}, \cV^{p'}_{(\cdot,W)}} + \|(\cdot,\bW)^n;(\cdot,\bW)\|_{p'},
\end{equation*}
where the implicit multiplicative constant depends only on $p'$, $m$, $d$, $\|b\|_{\crpW}$, $\|\sigma\|_{\crpW}$, $\sup_t |\det(\sigma_t \sigma_t^\top)|^{-1}$, $K$ and $w(0,T)$.

The same stability estimates also hold for $S^{n,i}$ and $S^i$, and for $\phi^{n,i}$ and $\phi^i$, for each $i = 1, \ldots, m$, where the respective constants also depend on $s_0$.

We can further apply Lemma~\ref{lem: stability of rough integrals against controlled paths} and Lemma~\ref{lem: estimate product of controlled paths} to obtain
\begin{equation*}
\|\phi^{n,0} - \phi^0\|_{p'} \lesssim \|\phi^{n,i}; \phi^i\|_{\cV^{p'}_{(\cdot,W)^n}, \cV^{p'}_{(\cdot,W)}} + \|S^{n,i};S^i\|_{\cV^{p'}_{(\cdot,W)^n}, \cV^{p'}_{(\cdot,W)}} + \|(\cdot,\bW)^n;(\cdot,\bW)\|_{p'},
\end{equation*}
and hence that
\begin{equation*}
\|\phi^{n,0} - \phi^0\|_{p'} \lesssim \|b^n;b\|_{\cV^{p'}_{(\cdot,W)^n}, \cV^{p'}_{(\cdot,W)}} + \|\sigma^n;\sigma\|_{\cV^{p'}_{(\cdot,W)^n}, \cV^{p'}_{(\cdot,W)}} + \|(\cdot,\bW)^n;(\cdot,\bW)\|_{p'},
\end{equation*}
where the implicit multiplicative constant depends only on $p'$, $m$, $d$, $\|b\|_{\crpW}$, $\|\sigma\|_{\crpW}$, $\sup_t |\det(\sigma_t \sigma_t^\top)|^{-1}$, $s_0$, $K$ and $w(0,T)$.

Combining the estimate derived above with the rate of convergence given in Lemma~\ref{lem: rate of convergence of space-time rough paths}, we obtain the desired convergence and estimate for the discretized portfolios. The convergence rate for the discretized wealth process $\hat{V}^n$ may then be derived exactly as in the proof of Theorem~\ref{thm: stability of optimal portfolios wrt discretization error}.
\end{proof}

\begin{remark}
If we assume additional regularity of the driving path $W$ and the sequence of partitions $(\cP^n)_{n \in \N}$, we can make the quantitative estimates provided in Theorem~\ref{thm: stability of optimal portfolios wrt discretization error for Black Scholes} more explicit, for instance by considering the regularity properties of Brownian sample paths; recall Corollaries~\ref{cor: convergence rate for BM und uniform partitions} and \ref{cor: convergence rate for BM und dyadic partitions} and the respective remarks for the local volatility models.

Furthermore, if we assume that $b$ and $\sigma$ have finite $\frac{p}{2}$-variation, so that $b' = 0$, $R^b = b$, $\sigma' = 0$ and $R^\sigma = \sigma$ for each $n \in \N$, the rate of convergence in Theorem~\ref{thm: stability of optimal portfolios wrt discretization error for Black Scholes} becomes more tractable, namely,
\begin{align*}
&\|(\phi^n,\kappa^n) - (\phi,\kappa)\|_{p'} + \|\hat{V}^n - \hat{V}\|_{p'}\\
&\lesssim \Big(\|b^n - b\|_\infty + \|\sigma^n - \sigma\|_\infty\Big)^{1-\frac{p}{p'}}\\
&\quad + |\cP^n|^{(1-\frac{1}{q}) (1-\frac{p}{p'})} + \|W^n - W\|_\infty^{1-\frac{p}{p'}} + \bigg\| \int_0^\cdot W^n_t \otimes \dd W_t - \int_0^\cdot W_t \otimes \dd W_t \bigg\|_\infty^{1-\frac{p}{p'}}.
\end{align*}
\end{remark}

\appendix
\section{Some essential results in rough path theory}\label{sec: appendix}

In this appendix, we collect some fundamental results in the theory of c{\`a}dl{\`a}g rough paths. While the analogous results are standard for stochastic It{\^o} integration, they are less well-known and in some cases novel in the context of rough integration.

In the following, we consider a general c{\`a}dl{\`a}g rough path $\bX = (X,\X) \in \cD^p([0,T];\R^d)$ for some $p \in [2,3)$, as introduced in Section~\ref{subsec: rough path theory}, and do not impose Property \textup{(RIE)} on $X$.

\subsection{Rough integration with respect to controlled paths}

This subsection contains slight modifications of results in \cite{Allan2023b} on rough integration with respect to controlled paths in $\crpX([0,T];\R^m)$.

\begin{lemma}[Proposition~2.4 in \cite{Allan2023b}]\label{lem: integration against controlled path}
  Let $\bX = (X, \X) \in \cD^p$ be a c{\`a}dl{\`a}g rough path and let $(F,F'), (G,G') \in \crpX$ be controlled paths with remainders $R^F$ and $R^G$, respectively. Then the limit
  \begin{equation}\label{eq: integration against controlled path}
    \int_0^T F_u \dd G_u := \lim_{|\cP| \to 0} \sum_{[s,t] \in \mathcal{P}} F_s G_{s,t} + F'_s G'_s \X_{s,t}
  \end{equation}
  exists along every sequence of partitions $\mathcal{P}$ of $[0,T]$ with mesh size $|\mathcal{P}| \to 0$, and comes with the estimate
  \begin{equation*}
    \begin{split}
    &\bigg|\int_s^t F_u \,\d G_u - F_s G_{s,t} - F'_s G'_s \X_{s,t} \bigg|\\
    &\quad \leq C\Big(\|F'\|_\infty (\|G'\|_{p,[s,t)}^p + \|X\|_{p,[s,t)}^p)^{\frac{2}{p}} \|X\|_{p,[s,t]} + \|F\|_{p,[s,t)} \|R^G\|_{\frac{p}{2},[s,t]}\\
    &\quad\quad\quad\quad + \|R^F\|_{\frac{p}{2},[s,t)} \|G'\|_\infty \|X\|_{p,[s,t]} + \|F'G'\|_{p,[s,t)} \|\X\|_{\frac{p}{2},[s,t]}\Big),
    \end{split}
  \end{equation*}
  for every $(s,t) \in \Delta_T$, where the constant $C$ depends only on $p$.
\end{lemma}

\begin{lemma}[Proposition~2.7~(ii) in \cite{Allan2023b}]\label{lem: stability of rough integrals against controlled paths}
  Let $\bX = (X, \X)$, $\tbX = (\tX, \tbbX)$ be c{\`a}dl{\`a}g rough paths, and let $(F,F'), (G,G') \in \crpX$ and $(\widetilde{F}, \widetilde{F}'), (\widetilde{G}, \widetilde{G}') \in \crptX$ be controlled paths. Let $M > 0$ be an upper bound for $\|F\|_{\crpX}$, $\|G\|_{\crpX}$, $\|\widetilde{F}\|_{\crptX}$, $\|\widetilde{G}\|_{\crptX}$, $\|\bX\|_p$ and $\|\tbX\|_p$. Then, there exists a constant $C$, depending only on $p$ and $M$, such that
  \begin{equation*}
\bigg\| \int_0^\cdot F_u \dd G_u - \int_0^\cdot \widetilde{F}_u \dd \widetilde{G}_u \bigg\|_p \leq C \Big(\|F;\widetilde{F}\|_{\crpX,\crptX} + \|G;\widetilde{G}\|_{\crpX,\crptX} + \|\bX;\tbX\|_p\Big),
  \end{equation*}
  where $\int_0^\cdot F_u \dd G_u$ and $\int_0^\cdot \widetilde{F}_u \dd \widetilde{G}_u$ are rough integrals, as defined in \eqref{eq: integration against controlled path}.
\end{lemma}

\subsection{The product of controlled paths}

\begin{lemma}\label{lem: product of controlled paths}
Let $\bX = (X,\X)$ be a c{\`a}dl{\`a}g rough path. The product operator on controlled paths, i.e., the map $\crpX \times \crpX \to \crpX$ given by
\begin{equation*}
((F,F'), (G,G')) \mapsto (FG,(FG)'),
\end{equation*}
where $(FG)' := F' G + F G'$, is defined whenever the dimensions of $F$ and $G$ are such that the multiplication $FG$ is valid. Moreover, it is a continuous bilinear map, and comes with the estimate
\begin{equation*}
\|FG\|_{\crpX} \leq C (1 + \|X\|_{p})^2 \|F\|_{\crpX} \|G\|_{\crpX},
\end{equation*}
where the constant $C$ depends only on $p$ and the dimensions of the paths involved. We call $(FG,(FG)')$ the product of $(F,F')$ and $(G,G')$, which we usually simply denote by $FG$.
\end{lemma}

The proof of Lemma~\ref{lem: product of controlled paths} is the same as the proof of the corresponding statement for continuous paths, which can be found in \cite[Lemma~A.1]{Allan2023a}.

\begin{lemma}\label{lem: estimate product of controlled paths}
  Let $\bX = (X,\X)$, $\tbX = (\tX,\tbbX)$ be c{\`a}dl{\`a}g rough paths and let $(F,F')$, $(G,G') \in \crpX$ and $(\widetilde{F},\widetilde{F}'), (\widetilde{G},\widetilde{G}') \in \crptX$ be controlled paths. Let $M > 0$ be an upper bound for $\|F\|_{\crpX}$, $\|G\|_{\crpX}$, $\|\widetilde{F}\|_{\crptX}$, $\|\widetilde{G}\|_{\crptX}$, $\|X\|_p$ and $\|\tX\|_p$. Then there exists a constant $C$, which depends only on $p$ and $M$, such that
  \begin{equation*}
    \|FG;\widetilde{F}\widetilde{G}\|_{\crpX,\crptX} \leq C \Big( \|F;\widetilde{F}\|_{\crpX,\crptX} + \|G;\widetilde{G}\|_{\crpX,\crptX} + \|X - \tX\|_p \Big).
  \end{equation*}
\end{lemma}

\begin{proof}
We have that
\begin{align*}
|(FG)_0 - (\widetilde{F}\widetilde{G})_0| &= |(F_0 - \widetilde{F}_0) G_0 + \widetilde{F}_0 (G_0 - \widetilde{G}_0)|\\
&\leq \|G\|_{\crpX} \|F;\widetilde{F}\|_{\crpX,\crptX} + \|\widetilde{F}\|_{\crptX} \|G;\widetilde{G}\|_{\crpX,\crptX}
\end{align*}
and
\begin{align*}
&|(FG)'_0 - (\widetilde{F}\widetilde{G})'_0| \leq |F'_0 G_0 - \widetilde{F}'_0 \widetilde{G}_0| + |F_0 G'_0 - \widetilde{F}_0 \widetilde{G}'_0|\\
&\quad \leq |F'_0 - \widetilde{F}'_0| |G_0| + |\widetilde{F}'_0| |G_0 - \widetilde{G}_0| + |F_0 - \widetilde{F}_0| |G'_0| + |\widetilde{F}_0| |G'_0 - \widetilde{G}'_0|\\
&\quad \leq \|G\|_{\crpX} \|F;\widetilde{F}\|_{\crpX,\crptX} + \|\widetilde{F}\|_{\crptX} \|G;\widetilde{G}\|_{\crpX,\crptX}.
\end{align*}
  Further, we have that
  \begin{align*}
    &\|(FG)' - (\widetilde{F}\widetilde{G})'\|_p\\
    &\leq \|F' - \widetilde{F}'\|_p \|G\|_\infty +  \|\widetilde{F}'\|_p \|G - \widetilde{G}\|_\infty + \|F' - \widetilde{F}'\|_\infty \|G\|_p + \|\widetilde{F}'\|_\infty \|G - \widetilde{G}\|_p\\
    &\quad + \|F - \widetilde{F}\|_p \|G'\|_\infty + \|\widetilde{F}\|_p \|G' - \widetilde{G}'\|_\infty + \|F - \widetilde{F}\|_\infty \|G'\|_p + \|\widetilde{F}\|_\infty \|G' - \widetilde{G}'\|_p\\
    &\leq (\|F - \widetilde{F}\|_\infty + \|F - \widetilde{F}\|_p + \|F' - \widetilde{F}'\|_\infty + \|F' - \widetilde{F}'\|_p) (\|G\|_\infty + \|G\|_p + \|G'\|_\infty + \|G'\|_p)\\
    &\quad + (\|\widetilde{F}\|_\infty + \|\widetilde{F}\|_p + \|\widetilde{F}'\|_\infty + \|\widetilde{F}'\|_p) (\|G - \widetilde{G}\|_\infty + \|G - \widetilde{G}\|_p + \|G' - \widetilde{G}'\|_\infty + \|G' - \widetilde{G}'\|_p)\\
    &\lesssim (1 + \|X\|_p) (1 + \|\tX\|_p) (1 + \|F\|_{\crpX}) \|G\|_{\crpX} (\|F;\widetilde{F}\|_{\crpX,\crptX} + \|X - \tX\|_p)\\
    &\quad + (1 + \|X\|_p) (1 + \|\tX\|_p)\|\widetilde{F}\|_{\crptX} (1 + \|\widetilde{G}\|_{\crptX}) (\|G;\widetilde{G}\|_{\crpX,\crptX} + \|X - \tX\|_p).
  \end{align*}

The remainder is given by $R^{FG}_{s,t} = R^F_{s,t} G_s + F_s R^G_{s,t} + F_{s,t} G_{s,t}$ for each $(s,t) \in \Delta_T$ (see the proof of \cite[Lemma~A.1]{Allan2023a}). We have that
  \begin{align*}
    &\|R^{FG}\|_{\frac{p}{2}}\\
    &\leq \|R^F - R^{\widetilde{F}}\|_{\frac{p}{2}} \|G\|_\infty + \|R^{\widetilde{F}}\|_{\frac{p}{2}} \|G - \widetilde{G}\|_\infty + \|F - \widetilde{F}\|_\infty \|R^G\|_{\frac{p}{2}} + \|\widetilde{F}\|_\infty \|R^G - R^{\widetilde{G}}\|_{\frac{p}{2}}\\
    &\quad + \|F - \widetilde{F}\|_p \|G\|_p + \|\widetilde{F}\|_p \|G - \widetilde{G}\|_p\\    
    &\lesssim (1+\|X\|_p)\|F;\widetilde{F}\|_{\crpX,\crptX} \|G\|_{\crpX}\\
    &\quad + (1+\|X\|_p) \|\widetilde{F}\|_{\crptX} (1 + \|\widetilde{G}\|_{\crptX}) (\|G;\widetilde{G}\|_{\crpX,\crptX} + \|X - \tX\|_p)\\
    &\quad + (1+\|\tX\|_p) (1 + \|F\|_{\crpX}) \|G\|_{\crpX} (\|F;\widetilde{F}\|_{\crpX,\crptX} + \|X - \tX\|_p)\\
    &\quad + (1+\|\tX\|_p) \|\widetilde{F}\|_{\crpX} \|G;\widetilde{G}\|_{\crpX,\crptX}\\
    &\quad + (1+\|X\|_p) (1+\|\tX\|_p) (1 + \|F\|_{\crpX}) \|G\|_{\crpX} (\|F;\widetilde{F}\|_{\crpX,\crptX} + \|X - \tX\|_p)\\
    &\quad + (1+\|X\|_p) (1+\|\tX\|_p) \|\widetilde{F}\|_{\crptX} (1 + \|\widetilde{G}\|_{\crptX}) (\|G;\widetilde{G}\|_{\crpX,\crptX} + \|X - \tX\|_p)\\
    &\lesssim (1 + \|X\|_p) (1 + \|\tX\|_p) (1 + \|F\|_{\crpX}) (1 + \|G\|_{\crpX}) (1 + \|\widetilde{F}\|_{\crptX}) (1 + \|\widetilde{G}\|_{\crptX})\\
    &\qquad \times (\|F;\widetilde{F}\|_{\crpX,\crptX} + \|G;\widetilde{G}\|_{\crpX,\crptX} + \|X - \tX\|_p).
  \end{align*}
  Combining the inequalities above, we deduce the desired estimate.
\end{proof}

\subsection{Associativity of rough integration}

The following proposition establishes the associativity of rough integration with respect to c{\`a}dl{\`a}g controlled paths.

\begin{proposition}\label{prop: associativity of rough integration}
  Let $\bX=(X,\X)$ be a c{\`a}dl{\`a}g rough path and let $(Y,Y'), (F,F'), (G,G') \in \crpX$ be controlled paths. Then $(Z,Z') := (\int_0^\cdot F_u \dd G_u, FG') \in \crpX$, and we have that
  \begin{equation*}
    \int_0^\cdot Y_u \dd Z_u = \int_0^\cdot Y_u F_u \dd G_u,
  \end{equation*}
  where on the left-hand side we have the integral of $(Y,Y')$ against $(Z,Z')$, and on the right-hand side we have the integral of $(YF,(YF)')$ against $(G,G')$, each defined in the sense of~\eqref{eq: integration against controlled path}.
\end{proposition}

The proof of Proposition~\ref{prop: associativity of rough integration} is identical to the proof of the corresponding statement for continuous paths, which can be found in \cite[Proposition~A.2]{Allan2023a}.

\subsection{The canonical rough path lift of a controlled path}

The next lemma provides the canonical construction of a c{\`a}dl{\`a}g rough path above a controlled path.

\begin{lemma}\label{lem: consistency of rough integrals}
  Let $\bX = (X, \X)$ be a c{\`a}dl{\`a}g rough path and $(Z,Z') \in \crpX$ be a controlled path. Then, $\bZ = (Z, \Z)$ is a c{\`a}dl{\`a}g rough path, where
  \begin{equation*}
    \Z_{s,t} := \int_s^t Z_u \otimes \d Z_u - Z_s \otimes Z_{s,t}, \qquad (s,t) \in \Delta_T,
  \end{equation*}
  with the integral defined as in~\eqref{eq: integration against controlled path}. We call $\bZ = (Z, \Z)$ the canonical rough path lift of $(Z, Z')$. Moreover, if $(Y, Y') \in \mathcal{V}^p_Z$, then $(Y, Y'Z') \in \crpX$, and
  \begin{equation*}
    \int_0^T Y_u \dd \bZ_u = \int_0^T Y_u \dd Z_u,
  \end{equation*}
  where on the left-hand side we have the rough integral of $(Y,Y')$ against $\bZ$, and on the right-hand side we have the integral of $(Y,Y'Z')$ against $(Z,Z')$ in the sense of \eqref{eq: integration against controlled path}.
\end{lemma}

The proof of Lemma~\ref{lem: consistency of rough integrals} follows the proof of the corresponding statement for continuous paths verbatim; see \cite[Lemma~A.3]{Allan2023a}.

\subsection{The exponential of a rough path}

Recall that, given a c{\`a}dl{\`a}g rough path $\bX = (X, \X)$, one can define the so-called reduced rough path $\bX^r = (X, [\bX])$, where $[\bX]_t := X_{0,t} \otimes X_{0,t} - 2 \Sym(\X_{0,t})$ is the rough path bracket of $\bX$; see, e.g., \cite[Section~2.4]{Friz2018}. If $X$ satisfies Property {\textup(RIE)} relative to $p$ and a sequence of partitions $(\cP^n)_{n\in\N}$, then, by \cite[Proposition~2.18]{Allan2023b}, one can see that the rough path bracket $[\bX]$ coincides with the pathwise quadratic variation $[X]$ of $X$, in the sense of F{\"o}llmer; see \cite{Follmer1981}. Using this notion, one can introduce the rough exponential analogously to the stochastic exponential of It{\^o} calculus.

In the following, given a path $X$, we will write $\Delta X_t := X_{t-,t}$ for the jump of $X$ at time $t$.

\begin{lemma}\label{lem: dynamics of rough exponential}
  Given a one-dimensional c{\`a}dl{\`a}g rough path $\bX = (X,\X)$ (so that in particular $X$ is real-valued), such that $X_0 = 0$, $\Delta [\bX]_t = (\Delta X_t)^2$ for every $t \in (0,T]$, and $\sum_{t \in (0,T]} (\Delta X_t)^2 < \infty$, the rough exponential $V = \cE(X)$ is defined by
  \begin{equation*}
    V_t := \exp \Big(X_t - \frac{1}{2} \Gamma_t\Big) \prod_{0 < s \leq t} (1 + \Delta X_s) \exp(- \Delta X_s), \qquad t \in [0,T],
  \end{equation*}
  where $\Gamma_t := [\bX]_t - \sum_{s \leq t} (\Delta X_s)^2$ for $t \in [0,T]$. We then have that $V$ is the unique controlled path in $\crpX$ satisfying the linear rough differential equation
  \begin{equation}\label{eq: dynamics of rough exponential}
    V_t = 1 + \int_0^t V_s \dd \bX_s, \qquad t \in [0,T],
  \end{equation}
  with Gubinelli derivative $V' = V$.
\end{lemma}

\begin{proof}
  Since we assume that $\sum_{t \in (0,T]} (\Delta X_t)^2 < \infty$, and $\Delta [\bX]_t = (\Delta X_t)^2$ for all $t \in (0,T]$, the path $\Gamma = [\bX] - \sum_{s \leq \cdot} (\Delta X_s)^2$ is continuous and has finite $\frac{p}{2}$-variation. Let $Y := X - \frac{1}{2} \Gamma$ and $A := \prod_{s \leq \cdot} (1 + \Delta X_s) \exp(-\Delta X_s)$. One can verify that $A$ is of finite $1$-variation; see, e.g., the proof of \cite[Ch.~II, Theorem~37]{Protter2005}. Hence, the two-dimensional path $Z := (Y, A)$ admits a rough path lift $\bZ = (Z, \Z)$, such that
  \begin{align*}
    &\Z^{1,1}_{s,t} = \X_{s,t} - \frac{1}{2}\int_s^t X_{s,u} \dd \Gamma_u - \frac{1}{2}\int_s^t \Gamma_{s,u} \dd X_u + \frac{1}{4} \int_s^t \Gamma_{s,u} \dd \Gamma_u,\\
    &\Z^{1,2}_{s,t} = \int_s^t Y_{s,u} \dd A_u, \qquad \Z^{2,1}_{s,t} = \int_s^t A_{s,u} \dd Y_u, \qquad \Z^{2,2}_{s,t} = \int_s^t A_{s,u} \dd A_u,
  \end{align*}
  for $(s,t) \in \Delta_T$, where all the integrals above are interpreted as Young integrals (as in, e.g., \cite[Proposition~2.4]{Friz2018}).

  We now consider the reduced rough path $(Z, [\bZ])$ associated with $\bZ$. By definition, we have that
  \begin{equation*}
    [Y,A]_t := [\bZ]^{1,2}_t = [\bZ]^{2,1}_t = Y_{0,t} A_{0,t} - \bigg(\int_0^t Y_{0,u} \dd A_u + \int_0^t A_{0,u} \dd Y_u \bigg).
  \end{equation*}
  Since $\int_0^t Y_{0,u} \dd A_u$ and $\int_0^t A_{0,u} \dd Y_u$ are Young integrals, for any sequence of partitions $(\cP^n)_{n \in \N}$ of $[0,T]$ with vanishing mesh size, we have that
  \begin{equation*}
    \int_0^t Y_{0,u} \dd A_u = \lim_{n \to \infty} \sum_{[u,v] \in \cP^n} Y_{0,u} A_{u \wedge t,v \wedge t}, \qquad \int_0^t A_{0,u} \dd Y_u = \lim_{n \to \infty} \sum_{[u,v] \in \cP^n} A_{0,u} Y_{u \wedge t,v \wedge t}.
  \end{equation*}
  
  Noting that
  \begin{equation*}
    Y_{0,v \wedge t} A_{0,v \wedge t} - Y_{0,u \wedge t} A_{0,u \wedge t} = Y_{0,u \wedge t} A_{u \wedge t,v \wedge t} + A_{0,u \wedge t} Y_{u \wedge t,v \wedge t} + Y_{u \wedge t,v \wedge t} A_{u \wedge t,v \wedge t} 
  \end{equation*}
  and taking $\lim_{n \to \infty} \sum_{[u,v] \in \cP^n}$ on each side, we obtain
  \begin{equation*}
    [Y,A]_t = \lim_{n \to \infty} \sum_{[u,v] \in \cP^n} Y_{u \wedge t,v \wedge t} A_{u \wedge t,v \wedge t} = \sum_{s \leq t} \Delta Y_s \Delta A_s = \sum_{s \leq t} \Delta X_s \Delta A_s,
  \end{equation*}
  and one can similarly show that $[A]_t := [\bZ]^{2,2}_t = \sum_{s \leq t} (\Delta A_s)^2$.

  Since $\Gamma$ is continuous and of finite $\frac{p}{2}$-variation, one can show, using the integration by parts formula for Young integrals, that $[Y]_t := [\bZ]^{1,1}_t = [\bX]_t$, so that $[Y]_t = \Gamma_t + \sum_{s \leq t} (\Delta X_s)^2$.

  Applying the It{\^o} formula for rough paths (\cite[Theorem~2.12]{Friz2018}) to $V_t = f(Z_t)$, where $f(y,a) := a \exp(y)$, and using the expressions derived above for the rough path bracket $[\bZ]$, a straightforward calculation (similar to the proof of \cite[Ch.~II, Theorem~37]{Protter2005} in the semimartingale setting) establishes that $V_t = 1 + \int_0^t V_s \dd \bX_s$. In particular, this involves noting that $\int_0^t V_s \dd Y_s = \int_0^t V_s \dd \bX_s - \frac{1}{2} \int_0^t V_s \dd \Gamma_s$, where in the first integral on the right-hand side we identify $(V,V)$ as a controlled path with respect to $X$.

  Finally, the uniqueness of solutions to \eqref{eq: dynamics of rough exponential} follows from straightforward estimates using the stability of rough integration (\cite[Lemma~3.4]{Friz2018}).
\end{proof}

\begin{remark}\label{remark: RIE paths satisfy conditions for rough exp}
  We note that, if $X$ is a real-valued c\`adl\`ag path which satisfies Property \textup{(RIE)}, then its canonical rough path lift $\bX = (X,\X)$ satisfies the assumptions of Lemma~\ref{lem: dynamics of rough exponential}. Indeed, by the definition of $\X$ in \eqref{eq: RIE rough path}, it is easy to see that $\X_{t-,t} = 0$ for every $t \in (0,T]$, and, recalling that the bracket $[\bX]$ satisfies $[\bX]_{s,t} = X_{s,t} \otimes X_{s,t} - 2 \Sym(\X_{s,t})$ for all $(s,t) \in \Delta_T$, we have that $\Delta [\bX]_t = (X_{t-,t})^2 - 2 \Sym(\X_{t-,t}) = (\Delta X_t)^2$ for every $t \in (0,T]$. Moreover, a slight generalization of \cite[Proposition~2.18]{Allan2023b} (to allow non-nested partitions) shows that the bracket $[\bX]$ has finite $1$-variation (and is actually increasing, since here $X$ is one-dimensional), and we then have that $\sum_{t \in (0,T]} (\Delta X_t)^2 = \sum_{t \in (0,T]} \Delta [\bX]_t \leq [\bX]_T < \infty$.
\end{remark}

\bibliography{quellen}{}

\def\cprime{$'$}
\providecommand{\bysame}{\leavevmode\hbox to3em{\hrulefill}\thinspace}
\providecommand{\MR}{\relax\ifhmode\unskip\space\fi MR }
\providecommand{\MRhref}[2]{%
  \href{http://www.ams.org/mathscinet-getitem?mr=#1}{#2}
}
\providecommand{\href}[2]{#2}
\begin{thebibliography}{AKLP23}

\bibitem[AC20]{Allan2020}
Andrew~L. Allan and Samuel~N. Cohen, \emph{Pathwise stochastic control with
  applications to robust filtering}, Ann. Appl. Probab. \textbf{30} (2020),
  no.~5, 2274--2310.

\bibitem[ACLP23]{Allan2023a}
Andrew~L. Allan, Christa Cuchiero, Chong Liu, and David~J. Pr\"omel,
  \emph{Model-free portfolio theory: a rough path approach}, Math. Finance
  \textbf{33} (2023), no.~3, 709--765.

\bibitem[AKLP23]{Allan2023c}
Andrew~L. Allan, Anna~P. Kwossek, Chong Liu, and David~J. Prömel,
  \emph{Pathwise convergence of the {E}uler scheme for rough and stochastic
  differential equations}, arXiv preprint arXiv:2309.16489 (2023).

\bibitem[ALP24]{Allan2023b}
Andrew~L. Allan, Chong Liu, and David~J. Pr\"omel, \emph{A c\`adl\`ag rough
  path foundation for robust finance}, Finance Stoch. \textbf{28} (2024),
  no.~1, 215--257.

\bibitem[BP17]{Biagini2017}
Sara Biagini and Mustafa~{\c C}. P{\i}nar, \emph{The robust {M}erton problem of
  an ambiguity averse investor}, Math. Financ. Econ. \textbf{11} (2017), no.~1,
  1--24.

\bibitem[CFRT16]{Cai2016}
Jiatu Cai, Masaaki Fukasawa, Mathieu Rosenbaum, and Peter Tankov, \emph{Optimal
  discretization of hedging strategies with directional views}, SIAM J.
  Financial Math. \textbf{7} (2016), no.~1, 34--69.

\bibitem[CSW19]{Cuchiero2019}
Christa Cuchiero, Walter Schachermayer, and Ting-Kam~Leonard Wong,
  \emph{Cover's universal portfolio, stochastic portfolio theory, and the
  num\'eraire portfolio}, Math. Finance \textbf{29} (2019), no.~3, 773--803.

\bibitem[FH20]{FrizHairer2020}
Peter~K. Friz and Martin Hairer, \emph{A course on rough paths}, second ed.,
  Universitext, Springer, Cham, [2020] \copyright 2020, With an introduction to
  regularity structures.

\bibitem[F{\"o}l81]{Follmer1981}
H.~F{\"o}llmer, \emph{Calcul d'{I}t\^{o} sans probabilit\'{e}s}, Seminar on
  {P}robability, {XV} ({U}niv. {S}trasbourg, {S}trasbourg, 1979/1980), Lecture
  Notes in Math., vol. 850, Springer, Berlin, 1981, pp.~143--150.

\bibitem[FS17]{Friz2017}
Peter~K. Friz and Atul Shekhar, \emph{General rough integration, {L}\'{e}vy
  rough paths and a {L}\'{e}vy-{K}intchine-type formula}, Ann. Probab.
  \textbf{45} (2017), no.~4, 2707--2765.

\bibitem[Fuk14]{Fukasawa2014}
Masaaki Fukasawa, \emph{Efficient discretization of stochastic integrals},
  Finance Stoch. \textbf{18} (2014), no.~1, 175--208.

\bibitem[FV10]{Friz2010}
Peter~K. Friz and Nicolas~B. Victoir, \emph{Multidimensional stochastic
  processes as rough paths}, Cambridge Studies in Advanced Mathematics, vol.
  120, Cambridge University Press, Cambridge, 2010, Theory and applications.

\bibitem[FZ18]{Friz2018}
Peter~K. Friz and Huilin Zhang, \emph{Differential equations driven by rough
  paths with jumps}, J. Differential Equations \textbf{264} (2018), no.~10,
  6226--6301.

\bibitem[GK00]{Goll2000}
Thomas Goll and Jan Kallsen, \emph{Optimal portfolios for logarithmic utility},
  Stochastic Process. Appl. \textbf{89} (2000), no.~1, 31--48.

\bibitem[GK03]{Goll2003}
\bysame, \emph{A complete explicit solution to the log-optimal portfolio
  problem}, Ann. Appl. Probab. \textbf{13} (2003), no.~2, 774--799.

\bibitem[GR12]{Guasoni2012}
Paolo Guasoni and Scott Robertson, \emph{Portfolios and risk premia for the
  long run}, Ann. Appl. Probab. \textbf{22} (2012), no.~1, 239--284.

\bibitem[GT01]{Gobet2001}
Emmanuel Gobet and Emmanuel Temam, \emph{Discrete time hedging errors for
  options with irregular payoffs}, Finance Stoch. \textbf{5} (2001), no.~3,
  357--367.

\bibitem[Jar18]{Jarrow2018}
Robert~A. Jarrow, \emph{Continuous-time asset pricing theory}, Springer Finance
  Textbooks, Springer, Cham, 2018, A martingale-based approach.

\bibitem[KK20]{Karatzas2020}
Ioannis Karatzas and Donghan Kim, \emph{Trading strategies generated pathwise
  by functions of market weights}, Finance Stoch. \textbf{24} (2020), no.~2,
  423--463.

\bibitem[KNP24]{Kwossek2024}
Anna~P. Kwossek, Andreas Neuenkirch, and David~J. Prömel, \emph{Functional
  differential equations driven by c\`adl\`ag rough paths}, arXiv preprint
  arXiv:2403.17573 (2024).

\bibitem[KO96]{Kim1996}
Tong~Suk Kim and Edward Omberg, \emph{Dynamic nonmyopic portfolio behavior},
  The Review of Financial Studies \textbf{9} (1996), no.~1, 141--161.

\bibitem[Kor97]{Korn1997}
Ralf Korn, \emph{Optimal portfolios: stochastic models for optimal investment
  and risk management in continuous time}, World Scientific, 1997.

\bibitem[Kra05]{Kraft2005}
Holger Kraft, \emph{Optimal portfolios and {H}eston's stochastic volatility
  model: an explicit solution for power utility}, Quant. Finance \textbf{5}
  (2005), no.~3, 303--313.

\bibitem[Kv11]{Kardaras2011}
Constantinos Kardaras and Gordan \v{Z}itkovi\'{c}, \emph{Stability of the
  utility maximization problem with random endowment in incomplete markets},
  Math. Finance \textbf{21} (2011), no.~2, 313--333.

\bibitem[Lej12]{Lejay2012}
Antoine Lejay, \emph{Global solutions to rough differential equations with
  unbounded vector fields}, S\'eminaire de {P}robabilit\'es {XLIV}, Lecture
  Notes in Math., vol. 2046, Springer, Heidelberg, 2012, pp.~215--246.

\bibitem[Lv07]{Larsen2007}
Kasper Larsen and Gordan \v{Z}itkovi\'{c}, \emph{Stability of
  utility-maximization in incomplete markets}, Stochastic Process. Appl.
  \textbf{117} (2007), no.~11, 1642--1662.

\bibitem[Mer69]{Merton1969}
Robert~C. Merton, \emph{Lifetime portfolio selection under uncertainty: The
  continuous-time case}, The Review of Economics and Statistics \textbf{51}
  (1969), no.~3, 247--257.

\bibitem[Mer71]{Merton1971}
Robert~C. Merton, \emph{Optimum consumption and portfolio rules in a
  continuous-time model}, J. Econom. Theory \textbf{3} (1971), no.~4, 373--413.

\bibitem[MS19]{Mostovyi2019}
Oleksii Mostovyi and Mihai S{\^i}rbu, \emph{Sensitivity analysis of the utility
  maximisation problem with respect to model perturbations}, Finance Stoch.
  \textbf{23} (2019), no.~3, 595--640.

\bibitem[NN18]{Neufeld2018}
Ariel Neufeld and Marcel Nutz, \emph{Robust utility maximization with {L}\'evy
  processes}, Math. Finance \textbf{28} (2018), no.~1, 82--105.

\bibitem[PP16]{Perkowski2016}
Nicolas Perkowski and David~J. Pr\"{o}mel, \emph{Pathwise stochastic integrals
  for model free finance}, Bernoulli \textbf{22} (2016), no.~4, 2486--2520.

\bibitem[Pro05]{Protter2005}
Philip~E. Protter, \emph{Stochastic integration and differential equations},
  Stochastic Modelling and Applied Probability, vol.~21, Springer-Verlag,
  Berlin, 2005, Second edition. Version 2.1, Corrected third printing.

\bibitem[PWZ22]{Pham2022}
Huy\^en Pham, Xiaoli Wei, and Chao Zhou, \emph{Portfolio diversification and
  model uncertainty: a robust dynamic mean-variance approach}, Math. Finance
  \textbf{32} (2022), no.~1, 349--404.

\bibitem[Rog01]{Rogers2001}
L.~C.~G. Rogers, \emph{The relaxed investor and parameter uncertainty}, Finance
  Stoch. \textbf{5} (2001), no.~2, 131--154.

\bibitem[RT14]{Rosenbaum2014}
Mathieu Rosenbaum and Peter Tankov, \emph{Asymptotically optimal discretization
  of hedging strategies with jumps}, Ann. Appl. Probab. \textbf{24} (2014),
  no.~3, 1002--1048.

\bibitem[SSV18]{Schied2018}
Alexander Schied, Leo Speiser, and Iryna Voloshchenko, \emph{Model-free
  portfolio theory and its functional master formula}, SIAM J. Financial Math.
  \textbf{9} (2018), no.~3, 1074--1101.

\bibitem[TTU13]{Tevzadze2013}
Revaz Tevzadze, Teimuraz Toronjadze, and Tamaz Uzunashvili, \emph{Robust
  utility maximization for a diffusion market model with misspecified
  coefficients}, Finance Stoch. \textbf{17} (2013), no.~3, 535--563.

\bibitem[Wes16]{Weston2016}
Kim Weston, \emph{Stability of utility maximization in nonequivalent markets},
  Finance Stoch. \textbf{20} (2016), no.~2, 511--541.

\bibitem[Zar99]{Zariphopoulou1999}
Thaleia Zariphopoulou, \emph{Optimal investment and consumption models with
  non-linear stock dynamics}, Math. Methods Oper. Res. \textbf{50} (1999),
  no.~2, 271--296.

\end{thebibliography}
\bibliographystyle{amsalpha}

\end{document}